\documentclass[12pt]{article}
\usepackage[utf8]{inputenc}

\usepackage[margin=1.25in]{geometry}
\usepackage[pdftex,dvipsnames,table]{xcolor}
\usepackage{amsmath,amssymb,amsthm,mathtools}
\usepackage{enumitem,comment}
\usepackage[onehalfspacing]{setspace}
\usepackage{graphicx,color,tikz}
\usepackage[font=footnotesize,labelfont=it]{caption}
\usepackage[font=footnotesize]{subcaption}
\usepackage[colorlinks=true,linkcolor=blue,citecolor=gray]{hyperref}
\usepackage{lscape}
\usepackage{authblk}
\usepackage[colon]{natbib}
\usepackage[title]{appendix}
\usepackage{chngcntr}
\usepackage{apptools}
\usepackage{footmisc}
\usepackage[normalem]{ulem}

\setlist{nolistsep}
\allowdisplaybreaks

\newcommand{\ol}{\overline}
\newcommand{\eps}{\varepsilon}

\DeclareMathOperator*{\argmin}{arg\,min}
\DeclareMathOperator*{\I}{1}

\newtheorem{assumption}{Assumption}[section]

\newtheorem{lemma}{Lemma}[section]

\newtheorem{theorem}{Theorem}[section]
\newtheorem{corollary}{Corollary}[section]
\newtheorem{remark}{Remark}[section]
\newtheorem{example}{Example}[section]

\title{Deconvolution from two order statistics\footnotetext{We thank Gaurab Aryal, Christian Gouri\'{e}roux, Philip Haile, Lixiong Li, Jos\'{e} Luis Montiel Olea, Chen Qiu, Cailin Slattery, and conference participants for helpful comments. We thank Cristi{\'a}n Hern{\'a}ndez, Daniel Quint, and Christopher Turansick for sharing their code.}}

\date{\small\today}

\author[1]{JoonHwan Cho\thanks{Department of Economics, Binghamton University, 4400 Vestal Parkway E., PO Box 6000, Binghamton, NY 13902-6000. Email: jcho13@binghamton.edu.}}
\author[2]{Yao Luo\thanks{Department of Economics, University of Toronto, Max Gluskin House, 150 St.~George St., Toronto, ON M5S 3G7. Email: yao.luo@utoronto.ca.}}
\author[3]{Ruli Xiao\thanks{Department of Economics, Indiana University, 100 S.~Woodlawn Ave., Bloomington, IN 47405. Email: rulixiao@iu.edu.}}

\affil[1]{Binghamton University}
\affil[2]{University of Toronto}
\affil[3]{Indiana University}

\begin{document}
\maketitle

\begin{abstract}
  Economic data are often contaminated by measurement errors and truncated by ranking. This paper shows that the classical measurement error model with independent and additive measurement errors is identified nonparametrically using only two order statistics of repeated measurements. The identification result confirms a hypothesis by \cite{athey2002identification} for a symmetric ascending auction model with unobserved heterogeneity. Extensions allow for heterogeneous measurement errors, broadening the applicability to additional empirical settings, including asymmetric auctions and wage offer models. We adapt an existing simulated sieve estimator and illustrate its performance in finite samples.
\end{abstract}

\medskip
\noindent\emph{Keywords}: Measurement error, Order statistics, Nonparametric identification, Spacing, Cross-sum

\thispagestyle{empty}

\clearpage


\newpage

\section{Introduction} \label{section:introduction}

We consider the classical measurement error model with repeated measurements:
\begin{align*}
  X_j = \xi + \eps_j , \qquad j \in \{1,\ldots,n\} ,
\end{align*}
where the latent variable of interest $\xi$ is measured $n$ times with i.i.d.~measurement errors $(\eps_j)_{j=1,\ldots,n}$ that are independent of $\xi$. The identification result under such a model is well-established when at least \emph{two} repeated measurements are observed, known as Kotlarski's lemma. The result has been widely applied in econometrics since its introduction by \cite{li1998nonparametric}.\footnote{Examples include \cite{bonhomme2010generalized}, \cite{krasnokutskaya2011identification} and \cite{grundl2019identification}. An outline of uses of Kotlarski's lemma in econometrics is in \cite{schennach2016recent}.} In practice, the researcher may observe only a subset of order statistics of the measurements, i.e., $(X_{(j)})_{j \in J}$, where $X_{(j)}$ is the $j$\textsuperscript{th} smallest order statistic from a sample of size $n$ and $J \subset \{1,\ldots,n\}$. This paper shows the underlying probability distributions of the latent variable and the measurement errors are identified from the joint distribution of the ordered measurements $(X_{(j)})_{j \in J}$. \par

The problem took its shape in \cite{athey2002identification} as a nonparametric identification problem in ascending auctions, which has particular relevance in auctions with electronic bidding. They conjecture that the model consists of enough structure to attain point identification using \emph{two} order statistics. However, the question has been long-standing for two decades.\footnote{Some positive findings are made recently in the finite mixture context; see \cite{mbakop2017identification} using five order statistics, \cite{luo2023identification} using two and an instrument and \cite{luo2021order} using three. However, these papers do not tackle the original conjecture in the spirit of Kotlarski's lemma.} Importantly, the standard approach using Kotlarski's lemma fails because of dependence in the order statistics of measurement errors. Further, the attempt to difference out the latent variable and exploit the spacing distribution is shown to be insufficient for point identification, evidenced by \cite{rossberg1972characterization}'s counterexample. \par

While the literature has documented the increasing importance of unobserved heterogeneity in auctions,\footnote{For example, see \cite{aradillas2013identification}, \cite{krasnokutskaya2011identification}, and \cite{li2000conditionally}.} the lack of identification results in the existing literature has hindered allowance for such heterogeneity in the classical fashion, unless relying on additional external variations in the data.\footnote{For example, \cite{hernandez2020estimation} uses variation in the number of bidders across auctions. \cite{freyberger2022identification} circumvents the issue of dependence between order statistics using observed reserve prices, rendering the identification problem classical.} We fill this important gap by showing that the underlying distributions are identified with only data on \emph{two} order statistics, which need not be consecutive or extreme, without relying on extra variations. \par

We propose a new identification strategy that exploits both features in \cite{kotlarski1967characterizing} and \cite{rossberg1972characterization}. In particular, we make use of the \emph{within} independence of the latent variable and the additive-separable measurement errors as in \cite{kotlarski1967characterizing} to derive that the model imposes an additional restriction beyond the spacing of order statistics.\footnote{The term \emph{spacing} usually refers to the difference between two \emph{consecutive} order statistics. Here, we use it more broadly as the difference between any two order statistics.} More precisely, we find that two measurement error distributions both rationalize the observed distribution of the ordered measurements only if the joint distributions of \emph{spacing} and \emph{cross-sum} are identical.\footnote{The term \emph{cross-sum} refers to the sum of two order statistics of two different ranks, where the order statistics are of two independent random samples from distinct underlying parent distributions.} Exploiting the structure of order statistics and commonly seen conditions (a support or tail restriction), we show that the spacing in conjunction with this additional restriction is indeed sufficient to point identify the underlying distribution using any two order statistics. The latent variable distribution is subsequently identified by a standard deconvolution argument. \par

We extend our main result to the setting where measurement errors are independent but nonidentically distributed (i.n.i.d.). We show that the underlying distributions are identified when two order statistics are observed, provided \emph{some} measurement errors are known to be identically distributed and the data consist of group identities of high-order statistics. The result applies to asymmetric ascending auctions, where only high-order dropout bids are recorded, and asymmetric first-price auctions and wage offer settings, where consecutive and extreme order statistics are observed. \par

The outline of the paper is as follows. Section \ref{section:motivating_examples} introduces two motivating examples and Section \ref{section:model_and_identification} presents the identification results. To complement the identification result, in Section \ref{section:estimator_main} we adapt \cite{bierens2012semi}'s simulated sieve estimator to consistently estimate the unknown distributions. Section \ref{section:conclusion} concludes. \par


\section{Motivating Examples} \label{section:motivating_examples}

To motivate the identification problem, we introduce two examples: an ascending auction model with auction-specific unobserved heterogeneity and a model of wage offers where wage is determined by worker-specific labor productivity. Throughout the paper, we focus only on unobserved heterogeneity since adding exogenous covariates does not require novel considerations. \par


\begin{example}[Ascending auction]
\label{example:ascending_auction}

Consider an ascending auction with one indivisible good and $n$ bidders. Each bidder's valuation for the item is determined by a set of item features---summarized and denoted by $\xi \in \mathbb{R}$---and a private value component $\eps_j$ that is independent of $\xi$ and across bidders. The valuation of bidder $j$ is given by $X_j = \xi + \eps_j$. Note that the valuation of each bidder can be cast as a measurement of the value $\xi$ with measurement error $\eps_j$. See, e.g., \cite{athey2002identification}. \par

In an ascending auction, the price rises until only one bidder remains. Suppose the auctioneer records prices $P_1 \le P_2 \le \cdots$ at which bidders drop out. Assuming that the bidders play the dominant strategy of remaining in the auction until the price reaches their valuations, the dropout bids reveal the ordered valuations of the bidders, i.e., $P_j = X_{(j)}$. Importantly, the highest valuation $X_{(n)}$ is never observed since the auction ends when the price reaches $P_{n-1} = X_{(n-1)}$, and the bidder with the highest valuation wins. Therefore, in an ascending auction, we observe only an incomplete set of order statistics on bidders' valuations. For example, \cite{kim2014nonparametric} observe at least the three highest dropout prices in ascending used-car auctions; see also \cite{larsen2021efficiency} for used-car auctions in the U.S. Data on timber auctions by the U.S.~Forest Service---analyzed in a number of empirical papers, e.g., \cite{haile2003inference}, \cite{athey2011comparing}, and \cite{aradillas2013identification} to name a few---records at most top twelve bids regardless of the number of bidders. \par

Both the distributions of $\xi$ and $\eps_j$ are of interest in an empirical analysis of auctions. For example, the counterfactual expected revenue to the auctioneer under a hypothetical auction rule requires the knowledge of both distributions. \par

\end{example}


\begin{example}[Wage offer determination]
\label{example:wage_offer}

Consider a simple wage offer model $X_j = \xi + w_j + \eps_j$ where $X_j$ is the $j$\textsuperscript{th} (log-)wage offer an individual receives. $X_j$ is assumed to be composed of the worker's productivity $\xi$, the wage rate $w_j$, and measurement error $\eps_j$ (on, e.g., labor productivity) associated with the offer. It is assumed that the period in consideration is relatively short such that the worker's productivity remains unchanged, but the worker may receive multiple offers; see \cite{guo2021identification}. \par

Data from the Survey of Consumer Expectations (SCE) Labor Market Survey records salary-related responses of all job offers for those who received at most three offers and responses to three best offers for those who received more than three offers within the last four months. In this setting, the researcher observes $X_{(n)},\ldots,X_{(n-k+1)}$ where $k = \min\{3,n\}$ for each respondent. The identification results in the current paper show that the distributions $F_\xi$ and $F_{w_j+\eps_j}$ are nonparametrically identified. \par

\end{example}


\section{The Model and Main Results} \label{section:model_and_identification}

In this section, we first formalize the i.i.d.~framework considered throughout the paper and provide sufficient conditions to identify the latent variable and measurement error distributions. In Section~\ref{subsection:extensions}, we consider i.n.i.d.~measurement errors. \par


\subsection{The Setup} \label{subsection:model_setup}

We begin by stating the sampling process. \par

\begin{assumption}[Sampling process]
\label{assumption:sampling_process}
  For each $1 \le j \le n$, $X_{j} := \xi + \eps_{j}$, where the random variable $\xi$ is independent of the random vector $(\eps_{1},\ldots,\eps_{n})$.\footnote{The identification results herein also applies to a setting with multiplicatively separable measurement error by taking logs when $\xi$ and $\eps_j$'s are positive, as in Example~\ref{example:wage_offer}.}
\end{assumption}
The researcher observes $(X_{(r)},X_{(s)})$, the $r$\textsuperscript{th} and $s$\textsuperscript{th} order statistics from $n$ observations, where $r$, $s$, and $n$ are known and $1 \le r < s \le n$.  That is, while the total number of measurements is known, one only observes two measurements of known ranks. In this section, we identify the unknown distributions of $\xi$ and $(\eps_{1},\ldots,\eps_{n})$ using $F_{X_{(r,s)}}$, the joint distribution of $(X_{(r)},X_{(s)})$, which is estimable from a random sample of the two order statistics. \par

For the benchmark case, we make the following distributional assumptions.
\begin{assumption}[i.i.d.~errors] 
\label{assumption:iid_case}
  \leavevmode
  \begin{enumerate}[label=(\alph*)]
    \item $\eps_{1},\ldots,\eps_{n}$ are i.i.d.~with a common distribution $F_\eps$ on $\mathbb{R}$.
    \item $F_\eps$ is absolutely continuous with a probability density function $f_\eps$ that is light-tailed, i.e., for some $C > 0$, $f_\eps(\epsilon) = O(e^{-C\lvert\epsilon\rvert})$ as $\lvert\epsilon\rvert \rightarrow \infty$.
  \end{enumerate}
\end{assumption}
The tail condition in Assumption~\ref{assumption:iid_case}(b) is trivially satisfied when the support is bounded. If the support is unbounded on either side, the assumption restricts the tail of the density function to decay at an exponential rate.\footnote{Throughout the paper, we use the term ``light-tailed'' to refer to densities with exponentially decaying tails as stated in the assumption. \label{fn:light_tailed}} The same assumption is found in \cite{evdokimov2012some}. Alternatively, we may follow \cite{kotlarski1967characterizing} and \cite{miller1970characterizing} to assume (a.e.-)nonvanishing or analytic characteristic functions (ch.f.) of both $\xi$ and $\eps$. Assumption~\ref{assumption:iid_case}(b) is a sufficient condition for the ch.f.~of $F_\eps$ to be analytic. However, to our knowledge, there exists no result regarding its implication on the joint ch.f.~of order statistics, the property we need for our identification results. In Lemma~\ref{lemma:appendix_joint_analyticity}, we show that the joint ch.f.~of the order statistics $(\eps_{(r)},\eps_{(s)})$ (and thus its marginals) is also analytic under this assumption. \par

In addition to Assumption~\ref{assumption:iid_case}, we further restrict the support of the measurement errors as follows. For any distribution $F$ on $\mathbb{R}$, let $S(F) \subseteq \mathbb{R}$ denote its support.\footnote{The support of $F$ is defined as the smallest closed set $K \subseteq \mathbb{R}$ such that $P_F(K) = 1$, where $P_F$ is the Borel probability measure induced by $F$. Thus, the density $f(x)$ may be zero for some values $x \in K$, including points on the boundary of $K$.}
\begin{assumption}[Support condition] 
\label{assumption:bound}
  The measurement errors are bounded from below, which is normalized to zero, i.e., $\inf S(F_\eps) = 0$.
\end{assumption}
Two aspects are worth mentioning regarding the above assumption. First, we require the measurement error to be bounded at one end; nevertheless, we allow the support to be possibly unbounded from above.\footnote{The results herein apply to the case when only the upper bound is finite, e.g., $S(F_\eps) = (-\infty,0]$, by considering $(X_{(r')}',X_{(s')}^\prime) = (-X_{(s)},-X_{(r)})$ where $r' = n-s+1$ and $s'=n-r+1$. \label{footnote:support_upper_bound}} We do not assume the upper bound of the support is known nor assume any additional structure on the support, e.g., connectedness. Second, as is typical with measurement error models, the underlying distributions are identified only up to location. Although one may alternatively normalize the mean and have an unknown but finite lower bound, normalizing the lower bound turns out to be more convenient for our identification argument. Our identification strategy heavily utilizes the condition that one boundary of the support of $F_\eps$ is finite. On the other hand, the distribution of $\xi$ is left completely unspecified. \par

The support restriction is nonrestrictive in many applications. The literature on games with incomplete information typically assumes bounded support of agent types (\cite{athey2001single}), such as bidder valuation in empirical auctions (\cite{guerre2000optimal}, \cite{athey2007nonparametric}), private costs of exerting efforts in contest games (\cite{huang2021structural}), and private information about variable costs in Cournot games (\cite{aryal2021empirical}). Wage offers are also bounded below in job search models (\cite{burdett1998wage}, \cite{guo2021identification}). \par


\subsection{Nonparametric Identification} \label{subsection:identification}

Throughout the section, we treat the observable joint distribution $F_{X_{(r,s)}}$ as known (i.e., as a datum), and show that the distribution $F_\xi$ of the latent variable $\xi$ and $F_\eps$ are identified under the aforementioned assumptions. The bulk of our identification result is concerned with showing that the measurement error distribution that can rationalize the observed distribution $F_{X_{(r,s)}}$ must be unique. Then the latent variable distribution is identified by a standard deconvolution argument. To achieve such a goal, we begin by isolating the empirical content on measurement errors from the datum. In Lemma~\ref{lemma:rationalizability} below, we do so by exploiting the multiplicative structure of the ch.f.~of a sum of independent random variables.\footnote{Similar ideas are also exploited in the classical repeated measurement error setting. See \cite{hall2003inference} and \cite{evdokimov2012some}.} \par

In the following lemmas, let $\eta_1,\ldots,\eta_n$ be $n$ i.i.d.~copies of $\eta \in \mathbb{R}$ and let  $(\eta_{(r)},\eta_{(s)})$ be order statistics. $\psi_{\eta_{(r,s)}}$ and $\psi_{\eta_{(r)}}$ denote the joint and marginal ch.f., respectively. \par

\begin{lemma}
\label{lemma:rationalizability}
If the sampling process is described as in Assumption~\ref{assumption:sampling_process} and the measurement errors satisfy Assumption~\ref{assumption:iid_case}, a measurement error distribution $F$ on $\mathbb{R}$ is data-consistent\footnote{We say a pair of distribution functions $(G_\xi,G_\eps)$ rationalizes the data, or is data-consistent, if $\xi' \sim G_\xi$ and $(\eps_j')_{j=1,\ldots,n} \sim \times_{j=1}^n G_\eps$, independent of $\xi'$, implies $(\xi'+\eps_{(r)}', \xi'+\eps_{(s)}') =_d (X_{(r)},X_{(s)})$. We say $G_\xi$ (resp., $G_\eps$) is data-consistent if $(G_\xi,G_\eps)$ is data-consistent for some $G_\eps$ (resp., $G_\xi$).} only if 
  \begin{align}
    \frac{\psi_{X_{(r,s)}}(t_r,t_s)}{\psi_{X_{(j)}}(t_r+t_s)} = \frac{\psi_{\eta_{(r,s)}}(t_r,t_s)}{\psi_{\eta_{(j)}}(t_r+t_s)} , \quad \text{for all $(t_r,t_s) \in B_0$, } j \in \{r,s\} , \label{eq:identify_eps}
  \end{align}
  where $\eta \sim F$ and $B_0$ is an open ball in $\mathbb{R}^2$ centered at zero. Further, $F$ induces a unique data-consistent latent variable distribution.
\end{lemma}

Suppose $F$ and $G$ are two measurement error distributions that are data-consistent. Lemma~\ref{lemma:rationalizability} states that the order statistics $(\eta_{(r)},\eta_{(s)})$ from the parent distribution $F$ and $(\eta'_{(r)},\eta'_{(s)})$ from the parent distribution $G$ must have the same ratio of joint and marginal ch.f.s. To obtain identification, it remains to show that such a measurement error distribution is unique, i.e., $F=G$. In contrast to the prevalent identification strategy in the existing literature that expands on \cite{kotlarski1967characterizing}, which begins with identifying the distribution of the latent variable $\xi$, Lemma~\ref{lemma:rationalizability} suggests an identification strategy that first focuses on the measurement error distribution. \par

Note that without any assumption on the latent variable $\xi$, the ch.f.~of $\xi$ may vanish on a nontrivial interval bounded away from the origin.\footnote{Recall that a ch.f.~takes value one at the origin and is uniformly continuous, implying that it does not vanish on some neighborhood about the origin.} Hence, in general, the ratio is only guaranteed to be well-defined on a neighborhood about the origin. Analyticity of the ch.f.~of measurement error order statistics (Lemma \ref{lemma:appendix_joint_analyticity}) thus plays a crucial role in pinning down the entire distribution based only on the information of the ch.f.~about the origin. In particular, Lemma~\ref{lemma:rationalizability} implies that any two data-consistent measurement error distributions must satisfy
\begin{align*}
  \psi_{\eta_{(r,s)}}(t_r,t_s) \psi_{\eta'_{(j)}}(t_r+t_s) = \psi_{\eta'_{(r,s)}}(t_r,t_s) \psi_{\eta_{(j)}}(t_r+t_s) , \quad j \in \{r,s\} ,
\end{align*}
for all $(t_r,t_s) \in B_0$. The equality extends to all of $\mathbb{R}^2$ because a product of two analytic functions is analytic. Finally, noticing that the ch.f.~of the sum of two independent random vectors is multiplicatively separable, the following lemma establishes necessary conditions for any two data-consistent measurement error distributions. \par

\begin{lemma}
\label{lemma:observational_equivalence_eps}
Let $\eta_{(r)}$ and $\eta_{(s)}$ be two order statistics of a random sample with size $n$ from $F$, and $\eta_{(r)}^\prime$ and $\eta_{(s)}^\prime$ be two order statistics of a random sample with size $n$ from $G$, where $(\eta_{(r)},\eta_{(s)})$ and $(\eta'_{(r)},\eta'_{(s)})$ are independent. Under Assumptions~\ref{assumption:sampling_process}, if measurement error distributions $F$ and $G$ admit light-tailed densities,\footnote{Cf.~Assumption~\ref{assumption:iid_case}(b) and footnote~\ref{fn:light_tailed}.} they both rationalize the data only if
  \begin{align}
    Z_{1j}
    := \begin{pmatrix} \eta_{(j)}^\prime + \eta_{(r)} \\ \eta_{(j)}^\prime + \eta_{(s)} \end{pmatrix}
    \stackrel{d}{=}
    \begin{pmatrix} \eta_{(j)} + \eta_{(r)}^\prime \\ \eta_{(j)} + \eta_{(s)}^\prime \end{pmatrix}
    =: Z_{2j} , \quad j \in \{r,s\} . \label{eq:observational_equivalence}
  \end{align}
\end{lemma}

The equivalence of ratios of ch.f.s~(Lemma~\ref{lemma:rationalizability}) implies an equivalence of two joint distributions of particular linear combinations of order statistics that arise from the two measurement error distributions. The linear combinations not only expose the empirical content contained in the spacing of errors (see Corollary~\ref{corollary:observational_equivalence_eps}) as investigated by \cite{hall2003inference} among many others in the classical repeated measurements setting and by \cite{rossberg1972characterization} in the context of order statistics, but they also imposes restrictions on the sum of order statistics from potentially distinct parent distributions, which we call a cross-sum. \par

Because both $Z_{1r}$ and $Z_{2r}$ involve three order statistics that are associated with potentially different parent distributions $F$ and $G$, it appears to be difficult to establish the equivalence of $F$ and $G$ directly from the distributional equivalence of $Z_{1r}$ and $Z_{2r}$. A reasonable attempt to tackle the problem would be to explore features of the distributions of $Z_{1r}$ and $Z_{2r}$ that depend on the parent distributions in a simple manner. Under the support condition in Assumption~\ref{assumption:bound}, we show that the distributions of $Z_{1r}$ and $Z_{2r}$ depend on only one of the two parent distributions upon conditioning on an extreme event.\footnote{The identification argument here and below investigates the implications of the condition $Z_{1r} =_d Z_{2r}$ only. $Z_{1s} =_d Z_{2s}$ in \eqref{eq:observational_equivalence} is the relevant condition to exploit under the assumption, in place of Assumption~\ref{assumption:bound}, that the measurement error is bounded from above (cf.~footnote~\ref{footnote:support_upper_bound}).} Thus the joint distribution of the cross-sum and spacing together with Assumption~\ref{assumption:bound} provide information sufficient to claim equivalence of two data-consistent measurement error distributions and point identify $F_\eps$. \par

\begin{lemma}\label{lemma:limit_of_conditional_distribution}
Let $\eta_{(r)}$, $\eta_{(s)}$, $\eta_{(r)}^\prime$ and $\eta_{(s)}^\prime$ be as specified in Lemma~\ref{lemma:observational_equivalence_eps}. If measurement error distributions $F$ and $G$ admit light-tailed densities and $\inf S(F) = \inf S(G) = 0$, we have, for all $c \in \mathbb{R}$,
\begin{align}
    \lim_{\delta \downarrow 0}  \mathbb{P}(\eta_{(r)}^\prime + \eta_{(s)} \le c \ | \ \eta_{(r)}^\prime + \eta_{(r)}  \le \delta) &= F_{s-r:n-r}(c) , \label{eqn:limitF} \\
    \lim_{\delta \downarrow 0}  \mathbb{P}(\eta_{(r)} + \eta_{(s)}^\prime \le c \ | \ \eta_{(r)} + \eta_{(r)}^\prime \le \delta) &= G_{s-r:n-r}(c) , \label{eqn:limitG}
\end{align}
where $F_{s-r:n-r}$ (resp., $G_{s-r:n-r}$) denotes the distribution of the $(s-r)$\textsuperscript{th} order statistic of a random sample with size $n-r$ from $F$ (resp., $G$).
\end{lemma}
For the sake of intuition, suppose that conditioning on the event $\{\eta_{(r)}^\prime + \eta_{(r)} = 0\}$ is well-defined. The condition implies that both $\eta_{(r)}^\prime = 0$ and $\eta_{(r)} = 0$. Thus, the independence between the pairs of order statistics simplifies the conditional distribution to $\mathbb{P}( \eta_{(s)} \le c \ | \ \eta_{(r)} = 0)$. Standard arguments in order statistics (cf.~Theorem~2.4.2 in \cite{arnold2008first}) suggest that this conditional distribution has a simple form:
\[
  \mathbb{P}( \eta_{(s)} \le c \ | \ \eta_{(r)} = 0) = F_{s-r:n-r}(c) .
\]
In other words, the equality above says that if the lowest $r$ draws all take the minimal possible value, then the distribution of any higher order statistic must equal the distribution of order statistic when the $r$ draws at the minimal value are discarded. \par

However, the event $\{\eta_{(r)}^\prime + \eta_{(r)} = 0\}$ may not be well-defined as it occurs with zero probability. In particular, the event $\{\eta_{(r)} = 0\}$ always has zero density (regardless of the density of $\eta$) unless $r=1$, i.e., the minimum order statistic. Further, because the true density $f_\varepsilon(0)$ may be zero, even the minimum order statistic may have zero density. In Lemma~\ref{lemma:limit_of_conditional_distribution}, we make the intuitive claim made above more rigorous by considering the limiting argument as in \eqref{eqn:limitF} and \eqref{eqn:limitG}, which is well-defined. \par

Lemma~\ref{lemma:observational_equivalence_eps} implies that if $F$ and $G$ are both data-consistent, then the conditional distributions \eqref{eqn:limitF} and \eqref{eqn:limitG} in Lemma~\ref{lemma:limit_of_conditional_distribution} must be equal, i.e., for all $c \in \mathbb{R}$,
\[
  F_{s-r:n-r}(c) = G_{s-r:n-r}(c) .
\]
As the distribution of an order statistic uniquely identifies the parent distribution, we can conclude that $F = G$.\footnote{The one-to-one mapping between the distribution of an order statistic and the parent distribution is a standard result, e.g., see \cite{david2003order}, p.10, (2.1.5). The mapping in (2.1.5) is an invertible map of the c.d.f.~$F$ because $p \mapsto I_p(a, b)$ is strictly monotone.} Combining Lemmas~\ref{lemma:rationalizability}--\ref{lemma:limit_of_conditional_distribution} leads to the following main identification result for the i.i.d.~case. \par

\begin{theorem}
\label{theorem:identification_iid_case}
Under Assumption \ref{assumption:sampling_process}, if the measurement errors satisfy Assumptions~\ref{assumption:iid_case} and \ref{assumption:bound}, both $F_\xi$ and $F_\eps$ are identified.
\end{theorem}

\subsubsection*{Discussion on \cite{rossberg1972characterization}'s Counterexample}

\cite{athey2002identification} pointed out that an identification strategy based on the spacing between order statistics does not yield a positive result. The discussion therein relied on a counterexample by \cite{rossberg1972characterization}. A straightforward corollary to Lemma~\ref{lemma:observational_equivalence_eps} highlights the model restrictions we exploit in addition to the spacing. \par

\begin{corollary}
\label{corollary:observational_equivalence_eps}
Under the same assumptions as in Lemma~\ref{lemma:observational_equivalence_eps}, measurement error distributions $F$ and $G$ both rationalize the data only if
  \begin{align*}
    \begin{pmatrix} \eta_{(s)} - \eta_{(r)} \\ \eta_{(r)}^\prime + \eta_{(s)} \end{pmatrix}
    \stackrel{d}{=} \begin{pmatrix} \eta_{(s)}^\prime - \eta_{(r)}^\prime \\ \eta_{(r)} + \eta_{(s)}^\prime \end{pmatrix}
    \quad \text{and} \quad
    \begin{pmatrix} \eta_{(s)} - \eta_{(r)} \\ \eta_{(r)} + \eta_{(s)}^\prime \end{pmatrix}
    \stackrel{d}{=} \begin{pmatrix} \eta_{(s)}^\prime - \eta_{(r)}^\prime \\ \eta_{(r)}^\prime + \eta_{(s)} \end{pmatrix} .
  \end{align*}
\end{corollary}
In words, the observational equivalence between any two measurement error distributions in our setting implies the equivalence of the joint distributions of the spacing and cross-sum of order statistics. This finding complements the nonidentification result in \cite{rossberg1972characterization} that the distribution of spacing between two order statistics alone cannot identify their parent distribution. The additional constraint on the distribution of cross-sum absent in \cite{rossberg1972characterization} emerges from the within independence assumption. Figure~\ref{fig:df_fig} in Appendix~\ref{appendix:rossberg} shows that, while the spacing distributions for the exponential and Rossberg's counterexample overlap, the cross-sum distributions differ over nontrivial regions. See Appendix~\ref{appendix:rossberg} for more discussion. \par


\subsection{Extensions of the Identification Result} \label{subsection:extensions}

We extend our identification result to the case of independent but nonidentically distributed (i.n.i.d.) measurement errors. To motivate the extension to the problem, we expand on Example~\ref{example:ascending_auction} by introducing bidder asymmetry and discuss additional features available in some auction data.

\begin{example}[Asymmetric auction]
In empirical auctions, bidders are said to be asymmetric if the bidders' private values $\eps_1,\ldots,\eps_n$ have different marginal (or parent, in the case of order statistics) distributions. Such asymmetry arises naturally in procurement auctions, where contractors differ in cost efficiency and productivity (see, e.g., \cite{flambard2006asymmetry}), timber auctions, where mills have the manufacturing capacity and loggers do not (see, e.g., \cite{athey2011comparing})\footnote{In a wage offer setting as in Example~\ref{example:wage_offer} employers may value the same productivity differently, which may result in heterogeneous measurement errors.}. \par

Due to asymmetry, identifying bidder-specific valuation distributions requires knowing their identities, i.e., knowing who participated in each auction. Fortunately, the auctioneer often publishes such identities despite the highest valuation being missing; see, e.g., \cite{athey2011comparing}. Another common practice in the asymmetric auction literature is grouping bidders by commonly known bidder types, such as mills and loggers in \cite{athey2011comparing} and strong and weak bidders in \cite{luo2018auctions}, which leads to more tractable theory and empirics. Such grouping uses additional information about the bidders, typically available as a public record, such as bidders' manufacturing capacity in a timber auction and pre-qualified contractors' experience in Department of Transportation procurement auctions.
\end{example}

In this section, we first show that any two order statistics suffice to point identify the underlying distributions when there is a relatively small number of groups and the group identities of high-order statistics are observable. We then discuss several relevant applications in empirical auctions. \par

\begin{assumption}[i.n.i.d.~errors] \label{assumption:noniid_case}
  \leavevmode
  \begin{enumerate}[label=(\alph*)]
    \item $\eps_{1},\ldots,\eps_{n}$ are independent with distributions $F_{\eps_1},\ldots,F_{\eps_n}$ on $\mathbb{R}$, respectively.
    \item For each $j \in \{1,\ldots,n\}$, $F_{\eps_j}$ admits a probability density function $f_{\eps_j}$ that is light-tailed, i.e., $f_{\eps_j}(\epsilon) = O(e^{-C_j\lvert\epsilon\rvert})$ as $\lvert\epsilon\rvert \rightarrow \infty$ for some $C_j > 0$.
    \item For each $j \in \{1,\ldots,n\}$, $\inf S(F_{\eps_j}) = 0$.
  \end{enumerate}
\end{assumption}

Assumption~\ref{assumption:noniid_case}(c) assumes both finite and common support lower bound across distinct measurement error distributions, the latter of which holds trivially in the i.i.d.~case. Note that apart from the lower boundary, the supports may differ (cf.~Corollary~\ref{corollary:identification_noniid}). 
The cost of relaxing the homogeneity assumption translates to more stringent data requirements. We require observing indices or types associated with all high-order statistics (for ranks above $r$). The following assumption records any prior information on the different types of measurement errors.
\begin{assumption}[Group structure] \label{assumption:group_structure}
  There exists a partition $g_1,\ldots,g_p$ of $\{1,\ldots,n\}$ such that $\eps_j =_d \eps_k$ if $j,k \in g_q$ for some $q \in \{1,\ldots,p\}$. In addition, the measurement errors $\eps_1,\ldots,\eps_n$ have common support.
\end{assumption}

Assumption~\ref{assumption:group_structure} posits there are at most $p$ distinct measurement error distributions. The case $p=1$ corresponds to the i.i.d.~case in Section~\ref{subsection:identification}; and $p=n$ corresponds to the case with no prior information on homogeneity. We do not preclude the possibility that two groups $g_q$ and $g_{q'}$ have the same distribution beyond the researcher's knowledge. E.g., the extreme case $p=n$ subsumes the i.i.d.~setting as a special case. Assumption~\ref{assumption:group_structure} implicitly assumes that the group structure is held constant across observations. In an application to auctions, the assumption may be imposed by restricting to auctions with the same composition of bidder types when all bidder types of participants are available in the dataset. In such a case, the analysis should be interpreted conditionally on the bidder composition. \par

The common support assumption is a sufficient condition to guarantee that all measurement error distributions are identified on the entirety of their support. Otherwise, some distributions may be identified only on a strict subset of their support, as we highlight in a discussion below. Let $R_{(j)} = \{q : X_{(j)} = X_k \text{ for some } k \in g_q\}$ be the group identity of the $j$\textsuperscript{th} order statistic.\footnote{Under Assumption~\ref{assumption:noniid_case}(b), $R_{(j)}$ is singleton with probability $1$. Thus, we abuse notation and use $R_{(j)}$ to denote both the set and the a.s.~unique element in $\{1,\ldots,p\}$.} \par

\begin{theorem} \label{theorem:identification_noniid_group}
  Suppose Assumptions~\ref{assumption:sampling_process}, \ref{assumption:noniid_case}, and \ref{assumption:group_structure} hold. Further, suppose two order statistics and the group identities of high-order statistics $(X_{(r)},X_{(s)},R_{(r+1)},\ldots,R_{(n)})$ are observed. If there exists a group $g_q$ with at least $n-r$ members, both $F_\xi$ and $(F_{\eps_j} : j=1,\ldots,n)$ are identified.
\end{theorem}

Note that the identification result does not require the researcher to observe the group identity $R_{(r)}$ of the $r$\textsuperscript{th} order statistic or those of lower order statistics. An heuristic explanation is provided in a discussion below. Theorem~\ref{theorem:identification_noniid_group} has an important implication when the measurement errors are left completely heterogeneous, i.e., $p=n$. In particular, the theorem implies that the observed order statistics should be consecutive and extreme because there is no group of a larger size, i.e., $n-r=1$.\footnote{The support condition in Assumption~\ref{assumption:group_structure} plays no role in this corollary because the observed order statistics are the largest among all measurement errors.}
\begin{corollary}
\label{corollary:identification_noniid}
  Suppose Assumptions~\ref{assumption:sampling_process} and \ref{assumption:noniid_case} hold. Further, suppose $(X_{(n-1)},X_{(n)},R_{(n)})$ is observed. Both $F_\xi$ and $(F_{\eps_j} : j=1,\ldots,n)$ are identified.    
\end{corollary}

So as to appreciate the additional conditions assumed in Theorem~\ref{theorem:identification_noniid_group}, we first illustrate the identification strategy in the setting of Corollary~\ref{corollary:identification_noniid}, which delivers a simpler argument. Suppose results similar to Lemmas~\ref{lemma:rationalizability}--\ref{lemma:limit_of_conditional_distribution} hold in the i.n.i.d.~case, in the sense that two sets of data-consistent measurement error distributions $(F_j : j=1,\ldots,n)$ and $(G_j : j=1,\ldots,n)$ must satisfy
\begin{align}
  \mathbb{P}(\eta_{(n)} \le c \,\vert\, \eta_{(n-1)} = 0) = \mathbb{P}(\eta'_{(n)} \le c \,\vert\, \eta'_{(n-1)} = 0) , \label{eq:conditional_implication_noniid}
\end{align}
for every constant $c$. Had the measurement errors been i.i.d., \eqref{eq:conditional_implication_noniid} corresponds the equivalence of two parent distributions. In the i.n.i.d.~case, the top-order statistic may arise from any of the parent distributions. Intuition suggests that this should be a mixture of measurement errors from different parent distributions. Since the mixing weights depend on $(F_j : j=1,\ldots,n)$, it appears formidable to show its equivalence with $(G_j : j=1,\ldots,n)$ from \eqref{eq:conditional_implication_noniid} without additional assumptions. Alternatively, suppose---as we formally show in the proof of Theorem~\ref{theorem:identification_noniid_group}---data-consistency implies a condition similar to \eqref{eq:conditional_implication_noniid} conditional on the member index of the highest order statistic. Heuristically speaking, because the index is known, say $j$, the conditional distribution simplifies to the $j$\textsuperscript{th} marginal distribution (i.e., a mixture with trivial weights). Thus $F_j = G_j$ and the $j$\textsuperscript{th} measurement error distribution is identified. Under the assumption of a common support lower bound, each measurement error has a nonzero chance of being the largest, which implies all $n$ measurement error distributions are identified. \par

Corollary~\ref{corollary:identification_noniid} does not grant identification of an ascending auction model when the data on dropout bids is incomplete and only a few high-order dropout bids as in, e.g., \cite{freyberger2022identification}. Nonetheless, group structures as in Theorem~\ref{theorem:identification_noniid_group} are commonly present in the empirical auction literature. \par

To illustrate the identification strategy in the setting of Theorem~\ref{theorem:identification_noniid_group}, suppose out of $n=4$ measurements, one observes $(X_{(2)},X_{(3)},R_{(3)},R_{(4)})$ and it is known \emph{a priori} that two of the measurement errors have a common distribution (call it group 1 and the other groups 2 and 3). By conditioning on the event that $\{R_{(3)} = R_{(4)} = 1\}$, we can homogenize the conditional distribution of the 3\textsuperscript{rd}-order statistic of measurement errors conditional on the 2\textsuperscript{nd}-order statistic taking the smallest possible value, which allows us to identify the measurement error distribution for group 1. This rests on (i) being able to observe the group identities and (ii) group 1 being large enough to condition on such an event. Provided the measurement errors have a common support, the remaining group distributions can be identified sequentially. For example, the group-2 distribution can be identified by conditioning on $\{R_{(3)} = 2, R_{(4)} = 1\}$.\footnote{If one measurement error has larger support than another, the distribution may not be fully identified. If group-1 measurement errors have support $[0,1]$ but group 2 has support $[0,2]$, regardless of whether one conditions on $\{R_{(3)} = 1, R_{(4)} = 2\}$ or $\{R_{(3)} = 2, R_{(4)} = 1\}$, the 3\textsuperscript{rd}-order statistic only has support $[0,1]$. Thus group-2 measurement error distribution is not identified on $(1,2]$.} \par

\begin{remark}
Corollary~\ref{corollary:identification_noniid} has natural applications in first-price auctions and wage offer settings, where consecutive and extreme order statistics are common and identities are observed. For instance, the Washington State Department of Transportation archives three apparent low bids and bidder identities of six months or older online. FDIC auction data contain the winning and second-highest bids and the associated identities (\cite{allen2023resolving}). U.S.~Forest Service timber auctions only record the top twelve bids and bidder identities regardless of the number of bidders. Lastly, data from the Survey of Consumer Expectations (SCE) Labor Market Survey records salary-related responses on the three best offers for those who received more than three offers within the last four months.
\end{remark}

\begin{remark} \label{remark:identification_noniid}
If \emph{all} dropout bids are observed, Corollary~\ref{corollary:identification_noniid} is applicable to ascending auctions. Specifically, if private values have a common upper bound, i.e., $\sup S(F_{\eps_j}) = \ol{\eps} < +\infty$, the lowest order statistics $(X_{(1)},X_{(2)})$ identify the underlying distributions.
\end{remark}


\section{Nonparametric Estimation} \label{section:estimator_main}

We propose a simulated sieve estimator for the i.i.d.~measurement error framework, which can be easily modified for the i.n.i.d.~case, albeit with the curse of dimensionality. The procedure, motivated by the sieve estimator proposed in \cite{bierens2012semi}, estimates the distribution functions of the latent variable and of the measurement error simultaneously. \par

\subsection{A Simulated Sieve Estimator} \label{subsection:estimator_sub}

We follow closely the development in \cite{bierens2008semi} to specify the parameter space and its sieve space. It has the advantage that we can incorporate prior information on the support without having to choose different orthogonal bases. This is an attractive feature for our purpose because, unlike the latent variable $\xi$, the measurement errors are restricted to be nonnegative-valued. The sieve space is constructed using only Legendre polynomials instead of using, e.g., Hermite polynomials for the latent variable and Laguerre polynomials for the measurement errors. \par

The construction of the sieve space begins with the observation that any absolutely continuous c.d.f.~$F$ on $\mathbb{R}$ can be expressed as $F(\cdot) = (H \circ G)(\cdot)$ where $H$ is an absolutely continuous c.d.f.~on $[0,1]$ and $G$ is an absolutely continuous c.d.f.~that is strictly increasing on $S(F)$. Equivalently,
\begin{align} \label{eq:cdf_representation_G}
  F(\cdot) = \int_0^{G(\cdot)} h(u) \, du = \int_0^{G(\cdot)} \pi^2(u) \, du ,
\end{align}
where $h = \pi^2$ is the density of $H$ and $\pi$ is a Borel-measurable square-integrable function on $[0,1]$.\footnote{The density of $F$ may then also expressed as $f(\cdot) = (h \circ G)(\cdot) g(\cdot)$ where $g$ is the density of $G$.} Thus, e.g., with a fixed $G$ with support $S(G) = [0,\infty)$, a large enough set $\mathcal{P}$ of Borel-measurable square-integrable functions maps via \eqref{eq:cdf_representation_G} to a large enough set of distribution functions with support contained in $[0,\infty)$. \par

The sieve space is constructed by using orthonormal polynomials to approximate $\pi$. \cite{bierens2008semi} considers the following compact set of square-integrable functions
\begin{align} \label{eq:legendre_series_pi}
  \mathcal{P} \coloneqq \left\{ \pi(\cdot) = \frac{1 + \sum_{\ell=1}^\infty \delta_\ell \rho_\ell(\cdot)}{\sqrt{1 + \sum_{\ell=1}^\infty \delta_\ell^2}} : \lvert\delta_\ell\rvert \le \frac{c}{1+\sqrt{\ell}\ln \ell} , \, \ell=1,2,\ldots \right\} ,
\end{align}
for some large constant $c > 0$, where $\rho_k$, defined on $[0,1]$, is a recentered and rescaled Legendre polynomial of order $k$: if $\ell_k$ is a Legendre polynomial of order $k$ on $[-1,1]$, then $\rho_k(u) := \sqrt{2k+1}\ell_k(2u-1)$ for $u \in [0,1]$ (see Section~2.2 in \cite{bierens2008semi}). This implicitly defines a compact parameter space $\mathcal{F}$ for the unknown c.d.f.~$F$ by \eqref{eq:cdf_representation_G} for some fixed $G$ and $\pi \in \mathcal{P}$. The sieve $\{\mathcal{F}_k\}_k$ is constructed by truncating the series in \eqref{eq:legendre_series_pi} at order $k$. Since we estimate two distribution functions $F_\xi$ and $F_\eps$, we consider a product sieve space $\{\mathcal{F}^\xi_k \times \mathcal{F}^\eps_k\}_k$ for two choices of $G$, denoted $G_\xi$ and $G_\eps$. \par

As \cite{bierens2008semi} notes, the c.d.f.~$G$ not only restricts the support but also acts as an initial guess of the unknown c.d.f.\footnote{The flexibility of choosing a base distribution $G$ in \cite{bierens2008semi} is more a norm than an exception in nonparametric methods using orthogonal bases. The standard polynomial sieve on the unit interval may be seen to have the uniform distribution as an initial guess. The Hermite and Laguerre polynomial sieves take normal and exponential distributions as initial guesses, respectively. While these ``initial guesses'' are determined naturally by the weight function associated with the orthogonal bases, the construction in \cite{bierens2008semi} allows for an explicit choice by the researcher.} With prior information on the shape of the distribution functions, an educated initial guess helps reduce the approximation error resulting from low-order sieve spaces. Without any prior information, one clearly cannot expect any \emph{a priori} advantage to choosing a particular distribution. In light of the often-used standard, Hermite, and Laguerre polynomial sieves, it may be reasonable to choose a uniform distribution if the support is known to be contained in a bounded interval, a normal distribution if one is agnostic about the support, and an exponential distribution if the support is known to be contained in $[0,\infty)$. \par

We consider a sieve extremum estimator that minimizes the average (squared) contrast between two empirical ch.f.s: one based on the factual sample and the other based on simulated draws. The population criterion function is devised as follows. For any candidate pair of distribution functions $F = (F_\xi,F_\eps)$, let
\[
  X_{(r,s)}(F)
    = (X_{(r)}(F), X_{(s)}(F))
    := (\xi(F_\xi) + \eps_{(r)}(F_\eps), \xi(F_\xi) + \eps_{(s)}(F_\eps)) ,
\]
where $\xi(F_\xi)$ is a random draw from $F_\xi$ and $\eps_{(j)}(F_\eps)$ is the $j$\textsuperscript{th} order statistic of $n$ i.i.d.~draws from $F_\eps$. For any $t=(t_r,t_s) \in \mathbb{R}^2$, let $\varphi(t;F) \coloneqq \mathbb{E} e^{it^\top X_{(r,s)}(F)}$ denote its ch.f.~where the expectation is induced by the $n+1$ uniform draws in the simulation process described below. We consider the following population criterion function:
\begin{align}
  Q(F)
  \coloneqq \frac{1}{4\kappa^2} \int \I\left\{t \in (-\kappa,\kappa)^2\right\} \left\lvert \psi_{X_{(r,s)}}(t) - \varphi(t;F) \right\rvert^2 \, dt , \label{eq:population_criterion_function}
\end{align}
where $\kappa > 0$ is a tuning parameter that determines the integration region. In place of the box weight $\I\{(t_r,t_s) \in (-\kappa,\kappa)^2\}$, one may also consider a smooth weight. We choose the latter because it has a closed-form expression for the empirical criterion function (see Appendix~\ref{appendix:estimation}).\footnote{While a data-driven choice of $\kappa$ (as well as the polynomial order $k_N$ in Theorem~\ref{theorem:estimator_consistency} below) may be of interest, we do not explore its possibility here.} We define a simulated sieve extremum estimator:
\begin{align}
  \widehat F_N \coloneqq (\widehat F_{\xi,N}, \widehat F_{\eps,N}) \in \argmin_{F \in \mathcal{F}^\xi_{k_N} \times \mathcal{F}^\eps_{k_N}} \widehat{Q}_N(F) , \label{eq:sieve_estimator_def}
\end{align}
where $\{k_N\}$ is an arbitrary sequence of positive integers such that $k_N \rightarrow \infty$ and $\widehat{Q}_N(F)$ is the empirical criterion function with the empirical ch.f.~$\widehat\psi_N$ and the simulated ch.f.~$\widehat\varphi_N(\cdot;F)$ in place of $\psi_{X_{(r,s:n)}}$ and $\varphi(\cdot;F)$, respectively. $\widehat\varphi_N(\cdot;F)$ is constructed from $N$ repeated draws of $(n+1)$ uniform random variables $(V,U_1,\ldots,U_n)$ and computing the inverse transform $X_{(j)}(F_k) = F_{\xi,k}^{-1}(V) + F_{\eps,k}^{-1}(U_{(j)})$ for $j \in \{r,s\}$. \par

We show that the estimator is consistent under the following set of assumptions.
\begin{assumption}[Consistency]
\label{assumption:consistency}
  \leavevmode
  \begin{enumerate}[label=(\alph*)]
    \item $\{(X_{(r),i},X_{(s),i})\}_{i=1,\ldots,N}$ are $N$ i.i.d.~realizations of $(X_{(r)},X_{(s)})$, where $(X_{(r)},X_{(s)})$ has a bounded support.
    \item $G_\xi$ and $G_\eps$ are known absolutely continuous c.d.f.s with support on $\mathbb{R}$ and $[0,\infty)$, respectively.
    \item Given $G_\xi$ and $G_\eps$ in part (b), the pair of true underlying distributions $(F_\xi,F_\eps)$ is in the closure of $\bigcup_k \mathcal{F}^\xi_k \times \mathcal{F}^\eps_k$.
    \item $\{(V_i,U_{1,i},\ldots,U_{n,i})\}_{i=1,\ldots,N}$ are $N$ i.i.d.~draws from $\mathcal{U}(0,1)^{n+1}$, independent of the sampling process.
  \end{enumerate}
\end{assumption}

The support restriction in Assumption~\ref{assumption:consistency}(a) is a sufficient condition to ensure that the measurement errors have light tails. Despite being restrictive, this appears to be the most straightforward assumption to impose on the \emph{observables} to guarantee identification.\footnote{Note that as remarked in Section~\ref{section:model_and_identification}, one may consider alternatives to Assumption~\ref{assumption:iid_case}(b), e.g., (a.e.-)nonvanishing or analytic ch.f.s, and still achieve the same identification result. Thus, one may also consider estimating the ch.f.s over a space of analytic functions and then estimate the densities via inverse Fourier transform. Clearly, various other estimation methods can be considered. For example, one may consider a sequential approach where one estimates the measurement error distribution via \eqref{eq:identify_eps} in the first stage and then estimate the latent variable distribution using nonparametric deconvolution in the second stage. Alternatively, one may consider estimating the densities via a sieve maximum likelihood. We consider the proposed simulated sieve extremum estimator as it both allows to estimate the underlying distributions simultaneously and admits a closed-form objective function.} Note that the researcher does not have to know \emph{a priori} the true support of the observables, nor the implied bounded supports for $F_\xi$ and $F_\eps$. \par

Assumption~\ref{assumption:consistency}(b) restricts the support of the base distribution for $F_\eps$ to be in line with the normalization in Assumption~\ref{assumption:iid_case}(b). Assumption~\ref{assumption:consistency}(c) is a standard assumption that the model is correctly specified. Assumption~\ref{assumption:consistency}(d) specifies the simulation process. Note that the random draws $(V_i,U_{j,i})_{j,i}$ are obtained once and not repeatedly drawn across different candidate parameter values $F_{\xi,k}$ and $F_{\eps,k}$ in order to ensure that the criterion function is continuous with respect to the parameter. \par

\begin{theorem}
\label{theorem:estimator_consistency}
  Let $\kappa > 0$ and let $\{k_N\}_N$ be any sequence of positive integers such that $k_N \rightarrow \infty$. Under Assumptions \ref{assumption:sampling_process}--\ref{assumption:bound} and \ref{assumption:consistency}, the estimator in \eqref{eq:sieve_estimator_def} is strongly uniformly consistent, i.e.,
  \[
  \max\left\{ \lVert \widehat{F}_{\xi,N} - F_\xi \rVert_\infty, \lVert \widehat{F}_{\eps,N} - F_\eps \rVert_\infty \right\} \stackrel{\text{a.s.}}{\longrightarrow} 0 .
  \]
\end{theorem}

Steps for implementing the estimator is described in Appendix~\ref{appendix:estimation} and its finite sample performance is illustrated below. As is standard in nonparametric estimation, the choice of the sieve order $k_N$ affects the approximation bias and variance of the estimator. An information-criterion-based approach analogous to that found in \cite{bierens2012semi} may be used to choose the sieve order $k_N$, although the procedure may be computationally intensive. With larger $\kappa$, the estimator is expected to perform better as it accounts for more discrepancy between the two ch.f.s. There appears to be no theoretical reason to restrict $\kappa$ except to ensure that the criterion function is well-defined, but limited simulation suggests the aggregation may come with larger variance. We do not have a useful criterion, but in light of the discussion, one may consider minimizing $\widehat{Q}_N$ over $\kappa$ as well as $F$ with a large upper bound on the parameter space for $\kappa$. From simulation studies, we find that the estimator is less sensitive to the choice of $\kappa$ as long as $\kappa$ is not too small. So we set $\kappa=1$ as its baseline value in this paper, as is done in \cite{bierens2012semi}, and explore its properties in a separate paper. We also leave inference procedures for future research.\footnote{Valid inference procedures exist in similar settings, for example, with independent measurement errors \citep{kato2021robust} or order statistics without unobserved heterogeneity \citep{menzel2013large}. A common feature they tackle is a certain lack of continuity in the inverse problems. Similar irregularity concerns may have to be addressed here.} \par


\subsection{Monte Carlo Evidence} \label{subsection:montecarlo}

To illustrate the performance of the proposed estimator, we conduct a simple Monte Carlo experiment. We begin by describing the data generating process (DGP). For observation $i$, the variable of interest $\xi_i$ is measured $n=3$ times with i.i.d.~measurement errors $\eps_{1,i},\ldots,\eps_{3,i}$. The measurement is constructed as $X_{j,i} = \xi_i + \eps_{j,i}$. We assume that only two order statistics of ranks $r=1$ and $s=2$ remain. That is, only $X_{(1),i}$ and $X_{(2),i}$ are recorded. We repeat the process to obtain $N$ pairs of observations. The experiment is replicated $R = 500$ times to obtain $500$ random samples of size $N$ of the form $\{x_{(1),i}^{(r)},x_{(2),i}^{(r)}\}_{i=1,\ldots,N}$. \par

For the distributions of $\xi_i$ and $\eps_{j,i}$, we set up a design that resembles \cite{hernandez2020estimation}'s application on eBay Motors auctions. Specifically, we use their estimated distributions of unobserved heterogeneity ($F_\xi$ in our set-up) and private values ($F_\varepsilon$ in our set-up) to calibrate the DGP for our simulation exercise. To construct these two distributions, we approximate the estimates in Figure~4 of \cite{hernandez2020estimation} with a sieve of order $6$, which resulted in almost identical distributions to those in the original article. \par

We then simulate data from the DGP and investigate finite-sample performance of our estimator. We consider sample sizes $N = 1000$, $2000$, and $4000$ with $k = 4$, $5$, and $6$, respectively. Note that by construction, there is no sieve approximation error when $N=4000$, i.e., any estimation error is associated only with sampling error. Furthermore, we set the bandwidth $\kappa =1$, $3.14$, and $5$, and the base distribution $G_\xi = \mathcal{N}(0,1/4)$ for $\xi$ and $G_\eps = \mathcal{N}(2,1)_+$ for $\varepsilon$, where $\mathcal{N}(2,1)_+$ denotes the truncated normal between 0 and $\infty$. \par

We present the estimation results for $F_\xi$ and $F_\eps$ with $\kappa=1$ in Figure~\ref{fig:montecarlo_maintext}. Simulation results with $\kappa=3.14$ and $5$ are similar and are presented in Figures~\ref{fig:montecarlo1} and \ref{fig:montecarlo2} in Appendix~\ref{appendix:estimation}. The true distribution functions are shown in black. We also plot some randomly selected estimates along with some box plots that illustrate the pointwise sampling error of $\widehat{F}_\xi$ and $\widehat{F}_\eps$ at various evaluation points. The figure suggest that the estimator performs reasonably well under all three sample sizes, and the performance improves with larger sample size. Although the choice of $\kappa$ does not seem to affect the behavior of the estimator in any ill-behaved manner, there appear to be larger pointwise variance when $\kappa = 3.14$ and $5$ relative to the case when $\kappa = 1$. \par

\begin{figure}[h!]
  \centering
  \begin{subfigure}[!t]{0.33\textwidth}
    \centering
    \includegraphics[width=\textwidth]{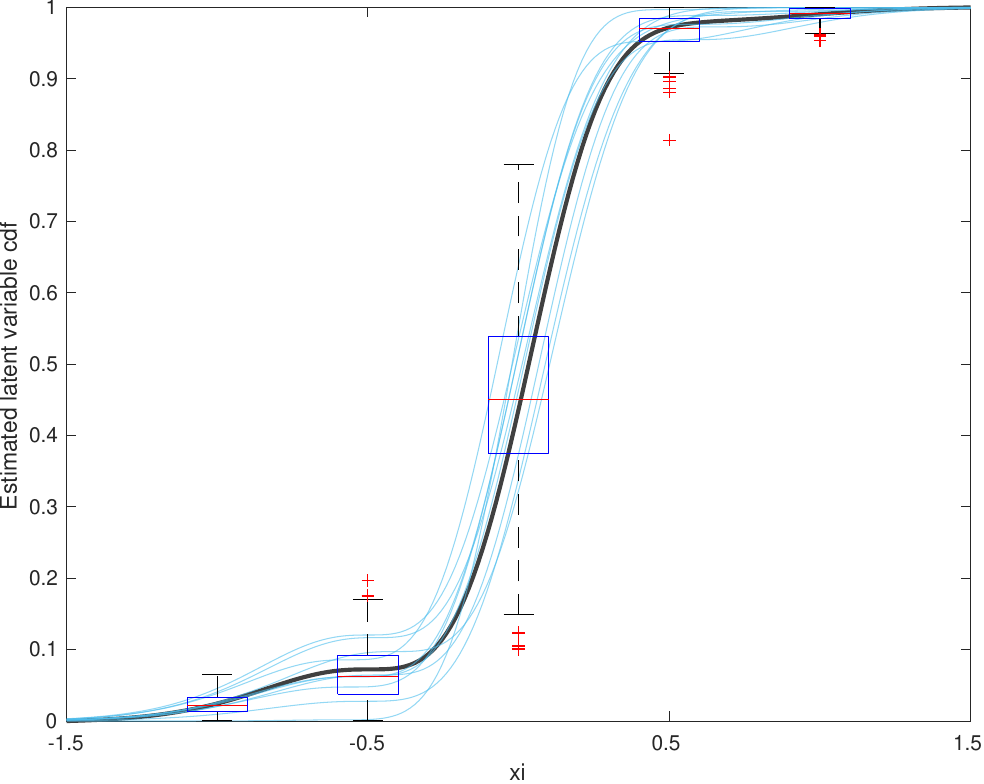}
    \caption*{\emph{Panel 1(a)}}
  \end{subfigure}%
  \hfill%
  \begin{subfigure}[!t]{0.33\textwidth}
    \centering
    \includegraphics[width=\textwidth]{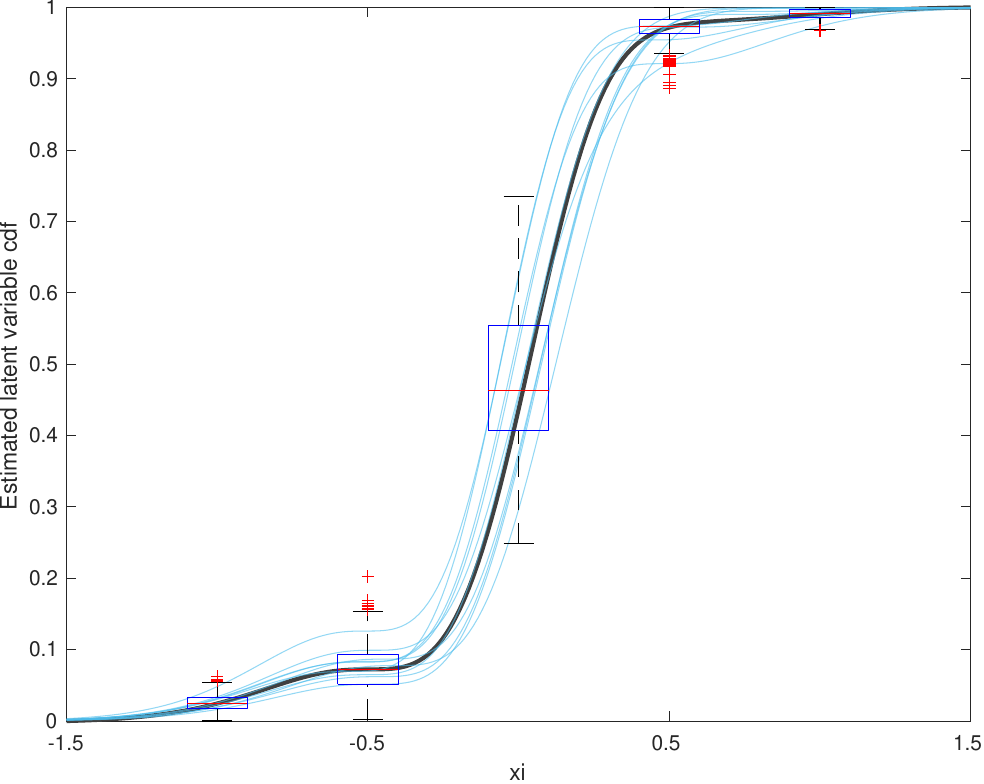}
    \caption*{\emph{Panel 1(b)}}
  \end{subfigure}%
  \hfill%
  \begin{subfigure}[!t]{0.33\textwidth}
    \centering
    \includegraphics[width=\textwidth]{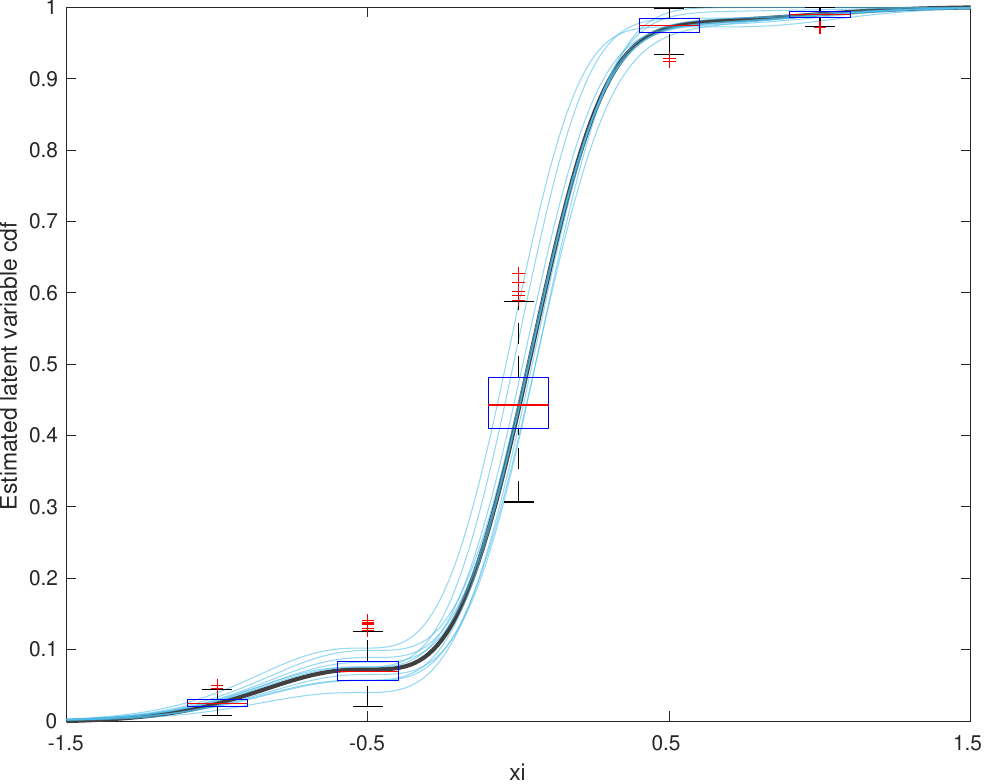}
    \caption*{\emph{Panel 1(c)}}
  \end{subfigure}
  \begin{subfigure}[!t]{0.33\textwidth}
    \centering
    \includegraphics[width=\textwidth]{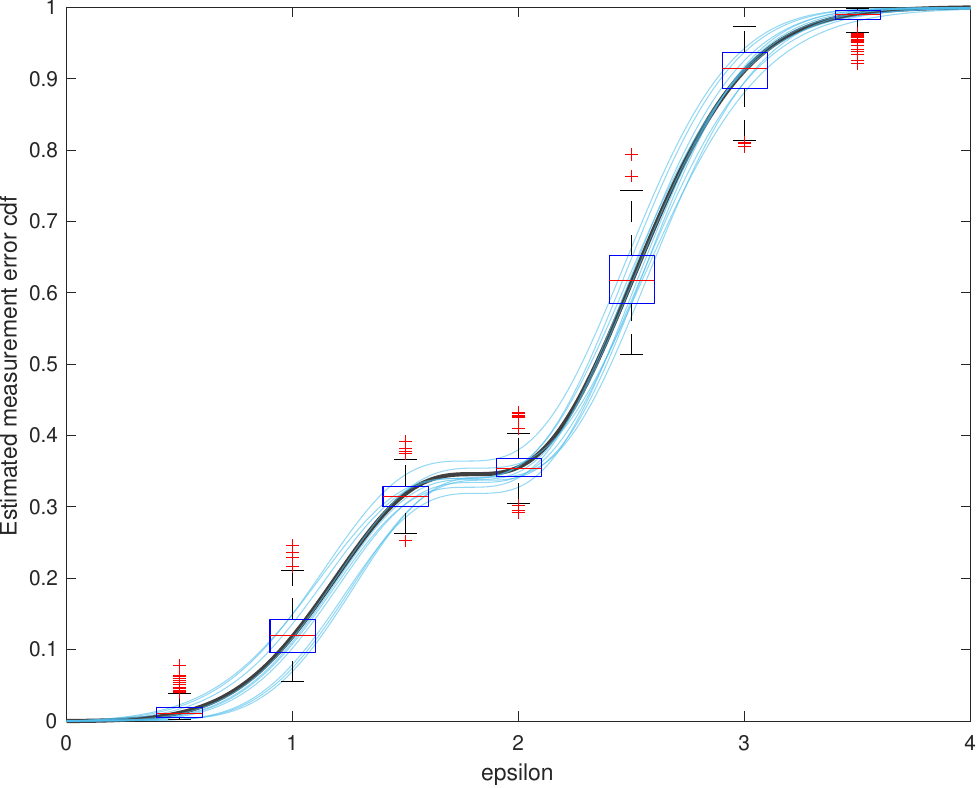}
    \caption*{\emph{Panel 2(a)}}
  \end{subfigure}%
  \hfill%
  \begin{subfigure}[!t]{0.33\textwidth}
    \centering
    \includegraphics[width=\textwidth]{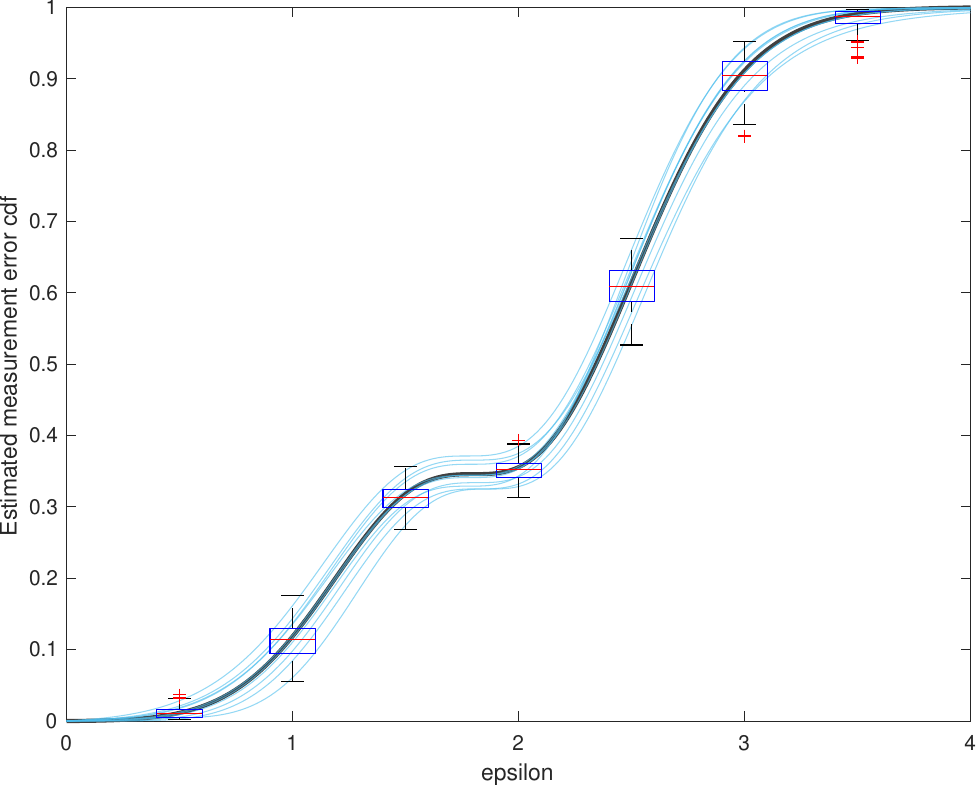}
    \caption*{\emph{Panel 2(b)}}
  \end{subfigure}%
  \hfill%
  \begin{subfigure}[!t]{0.33\textwidth}
    \centering
    \includegraphics[width=\textwidth]{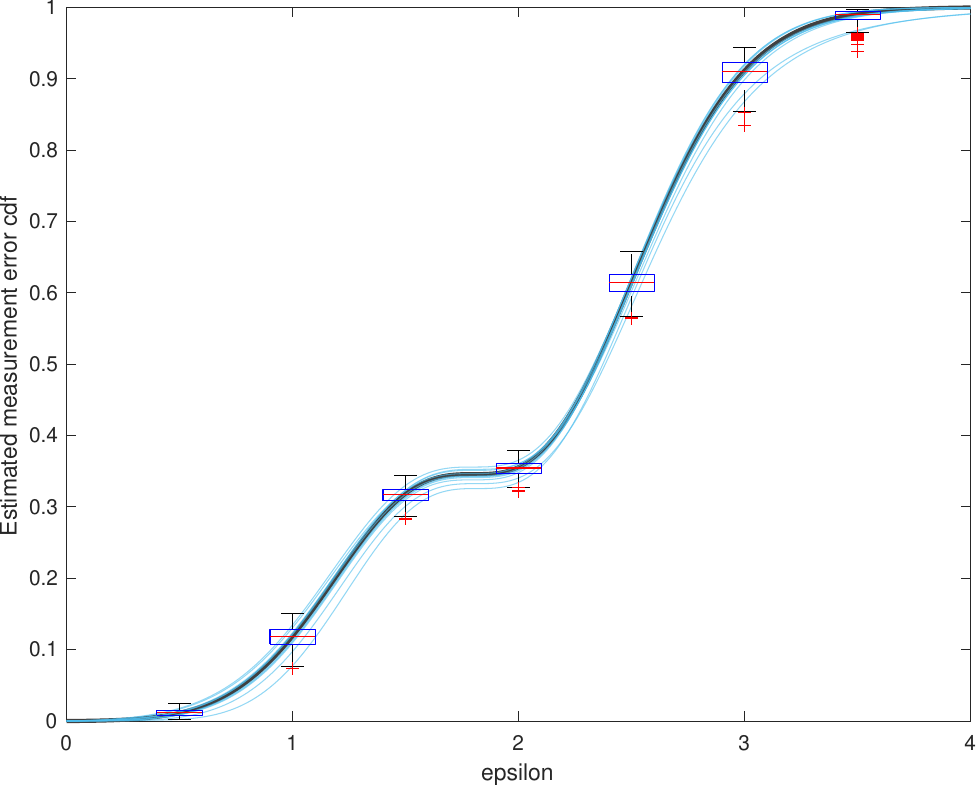}
    \caption*{\emph{Panel 2(c)}}
  \end{subfigure}
  \caption{Monte Carlo simulation results ($\kappa=1$). Panel 1 displays results for $F_\eps$ and Panel 2 for $F_\xi$. Subpanels \emph{(a)}, \emph{(b)}, and \emph{(c)} correspond to sample sizes $N=1000$, $2000$, and $4000$, respectively.} \label{fig:montecarlo_maintext}
\end{figure}


\section{Conclusion} \label{section:conclusion}

This paper shows that distributions of the latent variable and measurement errors are identified nonparametrically under mild assumptions when two or more order statistics are recorded from repeated measurements with independent errors, providing a positive answer to the hypothesis in \cite{athey2002identification} for an ascending auction with unobserved heterogeneity. Our results are also applicable to other applications with unobserved heterogeneity when order statistics are observed, survey data on wage offers being a notable example. More examples include repeated experiments with type II censoring, such as in reliability testing, where consecutive low-order failure times are recorded, and estimating the effects and damages of collusion in auctions, a setting in which \cite{asker2010study} emphasizes the importance of accounting for unobserved heterogeneity. Relatedly, the identification result may be applied to extend the framework in \cite{marmer2017identifying} for testing collusion in ascending auctions. \par


\newpage

\bibliographystyle{apalike}
\bibliography{bibfile}

\begin{thebibliography}{}

\bibitem[Allen et~al., 2023]{allen2023resolving}
Allen, J., Clark, R., Hickman, B., and Richert, E. (2023).
\newblock Resolving failed banks: Uncertainty, multiple bidding \& auction
  design.
\newblock {\em Review of Economic Studies}, page rdad062.

\bibitem[Aradillas-L{\'o}pez et~al., 2013]{aradillas2013identification}
Aradillas-L{\'o}pez, A., Gandhi, A., and Quint, D. (2013).
\newblock Identification and inference in ascending auctions with correlated
  private values.
\newblock {\em Econometrica}, 81(2):489--534.

\bibitem[Arnold et~al., 2008]{arnold2008first}
Arnold, B.~C., Balakrishnan, N., and Nagaraja, H.~N. (2008).
\newblock {\em A First Course in Order Statistics}.
\newblock SIAM.

\bibitem[Aryal and Zincenko, 2021]{aryal2021empirical}
Aryal, G. and Zincenko, F. (2021).
\newblock Empirical framework for {C}ournot oligopoly with private information.
\newblock {\em arXiv preprint arXiv:2106.15035}.

\bibitem[Asker, 2010]{asker2010study}
Asker, J. (2010).
\newblock A study of the internal organization of a bidding cartel.
\newblock {\em American Economic Review}, 100(3):724--62.

\bibitem[Athey, 2001]{athey2001single}
Athey, S. (2001).
\newblock Single crossing properties and the existence of pure strategy
  equilibria in games of incomplete information.
\newblock {\em Econometrica}, 69(4):861--889.

\bibitem[Athey and Haile, 2002]{athey2002identification}
Athey, S. and Haile, P.~A. (2002).
\newblock Identification of standard auction models.
\newblock {\em Econometrica}, 70(6):2107--2140.

\bibitem[Athey and Haile, 2007]{athey2007nonparametric}
Athey, S. and Haile, P.~A. (2007).
\newblock Nonparametric approaches to auctions.
\newblock {\em Handbook of Econometrics}, 6:3847--3965.

\bibitem[Athey et~al., 2011]{athey2011comparing}
Athey, S., Levin, J., and Seira, E. (2011).
\newblock Comparing open and sealed bid auctions: Evidence from timber
  auctions.
\newblock {\em The Quarterly Journal of Economics}, 126(1):207--257.

\bibitem[Bierens, 2008]{bierens2008semi}
Bierens, H.~J. (2008).
\newblock Semi-nonparametric interval-censored mixed proportional hazard
  models: {Identification} and consistency results.
\newblock {\em Econometric Theory}, 24(3):749--794.

\bibitem[Bierens and Song, 2012]{bierens2012semi}
Bierens, H.~J. and Song, H. (2012).
\newblock Semi-nonparametric estimation of independently and identically
  repeated first-price auctions via an integrated simulated moments method.
\newblock {\em Journal of Econometrics}, 168(1):108--119.

\bibitem[Bonhomme and Robin, 2010]{bonhomme2010generalized}
Bonhomme, S. and Robin, J.-M. (2010).
\newblock Generalized non-parametric deconvolution with an application to
  earnings dynamics.
\newblock {\em The Review of Economic Studies}, 77(2):491--533.

\bibitem[Burdett and Mortensen, 1998]{burdett1998wage}
Burdett, K. and Mortensen, D.~T. (1998).
\newblock Wage differentials, employer size, and unemployment.
\newblock {\em International Economic Review}, pages 257--273.

\bibitem[David and Nagaraja, 2003]{david2003order}
David, H.~A. and Nagaraja, H.~N. (2003).
\newblock {\em Order Statistics}.
\newblock John Wiley \& Sons.

\bibitem[Evdokimov and White, 2012]{evdokimov2012some}
Evdokimov, K. and White, H. (2012).
\newblock Some extensions of a lemma of {Kotlarski}.
\newblock {\em Econometric Theory}, 28(4):925--932.

\bibitem[Flambard and Perrigne, 2006]{flambard2006asymmetry}
Flambard, V. and Perrigne, I. (2006).
\newblock Asymmetry in procurement auctions: Evidence from snow removal
  contracts.
\newblock {\em The Economic Journal}, 116(514):1014--1036.

\bibitem[Freyberger and Larsen, 2022]{freyberger2022identification}
Freyberger, J. and Larsen, B.~J. (2022).
\newblock Identification in ascending auctions, with an application to digital
  rights management.
\newblock {\em Quantitative Economics}, 13(2):505--543.

\bibitem[Gallant, 1987]{gallant1987identification}
Gallant, A.~R. (1987).
\newblock Identification and consistency in semi-nonparametric regression.
\newblock {\em Advances in Econometrics: Fifth World Congress}, 1:145--170.

\bibitem[Gallant and Nychka, 1987]{gallant1987semi}
Gallant, A.~R. and Nychka, D.~W. (1987).
\newblock Semi-nonparametric maximum likelihood estimation.
\newblock {\em Econometrica}, 55(2):363--390.

\bibitem[Grundl and Zhu, 2019]{grundl2019identification}
Grundl, S. and Zhu, Y. (2019).
\newblock Identification and estimation of risk aversion in first-price
  auctions with unobserved auction heterogeneity.
\newblock {\em Journal of Econometrics}, 210(2):363--378.

\bibitem[Guerre et~al., 2000]{guerre2000optimal}
Guerre, E., Perrigne, I., and Vuong, Q. (2000).
\newblock Optimal nonparametric estimation of first-price auctions.
\newblock {\em Econometrica}, 68(3):525--574.

\bibitem[Guo, 2021]{guo2021identification}
Guo, J. (2021).
\newblock Identification of the wage offer distribution using order statistics.
\newblock Working Paper 3910119, SSRN.

\bibitem[Haile and Tamer, 2003]{haile2003inference}
Haile, P.~A. and Tamer, E. (2003).
\newblock Inference with an incomplete model of english auctions.
\newblock {\em Journal of Political Economy}, 111(1):1--51.

\bibitem[Hall and Yao, 2003]{hall2003inference}
Hall, P. and Yao, Q. (2003).
\newblock Inference in components of variance models with low replication.
\newblock {\em The Annals of Statistics}, 31(2):414--441.

\bibitem[Hern{\'a}ndez et~al., 2020]{hernandez2020estimation}
Hern{\'a}ndez, C., Quint, D., and Turansick, C. (2020).
\newblock Estimation in {English} auctions with unobserved heterogeneity.
\newblock {\em The RAND Journal of Economics}, 51(3):868--904.

\bibitem[H{\"o}rmander, 1973]{hormander1973introduction}
H{\"o}rmander, L. (1973).
\newblock {\em An Introduction to Complex Analysis in Several Variables}.
\newblock Elsevier.

\bibitem[Huang and He, 2021]{huang2021structural}
Huang, Y. and He, M. (2021).
\newblock Structural analysis of {T}ullock contests with an application to us
  house of representatives elections.
\newblock {\em International Economic Review}, 62(3):1011--1054.

\bibitem[Kato et~al., 2021]{kato2021robust}
Kato, K., Sasaki, Y., and Ura, T. (2021).
\newblock Robust inference in deconvolution.
\newblock {\em Quantitative Economics}, 12(1):109--142.

\bibitem[Kim and Lee, 2014]{kim2014nonparametric}
Kim, K.~i. and Lee, J. (2014).
\newblock Nonparametric estimation and testing of the symmetric {IPV} framework
  with unknown number of bidders.
\newblock Working paper.

\bibitem[Kotlarski, 1967]{kotlarski1967characterizing}
Kotlarski, I. (1967).
\newblock On characterizing the gamma and the normal distribution.
\newblock {\em Pacific Journal of Mathematics}, 20(1):69--76.

\bibitem[Krasnokutskaya, 2011]{krasnokutskaya2011identification}
Krasnokutskaya, E. (2011).
\newblock Identification and estimation of auction models with unobserved
  heterogeneity.
\newblock {\em The Review of Economic Studies}, 78(1):293--327.

\bibitem[Kruse and Deely, 1969]{kruse1969joint}
Kruse, R. and Deely, J. (1969).
\newblock Joint continuity of monotonic functions.
\newblock {\em The American Mathematical Monthly}, 76(1):74--76.

\bibitem[Larsen, 2021]{larsen2021efficiency}
Larsen, B.~J. (2021).
\newblock The efficiency of real-world bargaining: Evidence from wholesale
  used-auto auctions.
\newblock {\em The Review of Economic Studies}, 88(2):851--882.

\bibitem[Li et~al., 2000]{li2000conditionally}
Li, T., Perrigne, I., and Vuong, Q. (2000).
\newblock Conditionally independent private information in {OCS} wildcat
  auctions.
\newblock {\em Journal of Econometrics}, 98(1):129--161.

\bibitem[Li and Vuong, 1998]{li1998nonparametric}
Li, T. and Vuong, Q. (1998).
\newblock Nonparametric estimation of the measurement error model using
  multiple indicators.
\newblock {\em Journal of Multivariate Analysis}, 65(2):139--165.

\bibitem[Luo et~al., 2018]{luo2018auctions}
Luo, Y., Perrigne, I., and Vuong, Q. (2018).
\newblock Auctions with ex post uncertainty.
\newblock {\em The RAND Journal of Economics}, 49(3):574--593.

\bibitem[Luo et~al., 2021]{luo2021order}
Luo, Y., Sang, P., and Xiao, R. (2021).
\newblock Order statistics approaches to unobserved heterogeneity in auctions.
\newblock Working Paper 3704644, SSRN.

\bibitem[Luo and Xiao, 2023]{luo2023identification}
Luo, Y. and Xiao, R. (2023).
\newblock Identification of auction models using order statistics.
\newblock {\em Journal of Econometrics}, 236(1):105457.

\bibitem[Marmer et~al., 2017]{marmer2017identifying}
Marmer, V., Shneyerov, A.~A., and Kaplan, U. (2017).
\newblock Identifying collusion in {English} auctions.
\newblock Working Paper 2738789, SSRN.

\bibitem[Mbakop, 2017]{mbakop2017identification}
Mbakop, E. (2017).
\newblock Identification of auctions with incomplete bid data in the presence
  of unobserved heterogeneity.
\newblock Working paper, Northwestern University.

\bibitem[Menzel and Morganti, 2013]{menzel2013large}
Menzel, K. and Morganti, P. (2013).
\newblock Large sample properties for estimators based on the order statistics
  approach in auctions.
\newblock {\em Quantitative Economics}, 4(2):329--375.

\bibitem[Miller, 1970]{miller1970characterizing}
Miller, P. (1970).
\newblock Characterizing the distributions of three independent
  {$n$}-dimensional random variables, {$X_1, X_2, X_3$}, having analytic
  characteristic functions by the joint distribution of {$(X_1 + X_3, X_2 +
  X_3)$}.
\newblock {\em Pacific Journal of Mathematics}, 35(2):487--491.

\bibitem[Rossberg, 1972]{rossberg1972characterization}
Rossberg, H.~J. (1972).
\newblock Characterization of the exponential and the {Pareto} distributions by
  means of some properties of the distributions which the differences and
  quotients of order statistics are subject to.
\newblock {\em Statistics: A Journal of Theoretical and Applied Statistics},
  3(3):207--216.

\bibitem[Schennach, 2016]{schennach2016recent}
Schennach, S.~M. (2016).
\newblock Recent advances in the measurement error literature.
\newblock {\em Annual Review of Economics}, 8:341--377.

\bibitem[Stein and Shakarchi, 2010]{stein2010complex}
Stein, E.~M. and Shakarchi, R. (2010).
\newblock {\em Complex Analysis}, volume~2.
\newblock Princeton University Press.

\end{thebibliography}


\newpage

\begin{appendices}

\counterwithin*{equation}{section}
\renewcommand\theequation{\thesection\arabic{equation}}


\section{Proofs} \label{appendix:proofs}


\subsection{Supplementary Results}

\begin{lemma}
\label{lemma:appendix_joint_analyticity}
Let $(\eps_{(r)},\eps_{(s)})$ be $r$\textsuperscript{th} and $s$\textsuperscript{th} order statistics from $\eps_1,\ldots,\eps_n$ which are independent with parent distributions $F_{\eps_1},\ldots,F_{\eps_n}$. If every parent distribution $F_{\eps_j}$ has a density function $f_{\eps_j}$ that is light-tailed (i.e., for some $C_j > 0$, $f_{\eps_j}(\epsilon) = O(e^{-C_j\lvert\epsilon\rvert})$ as $\lvert\epsilon\rvert \to \infty$), then the ch.f.~of order statistics $\psi_{\eps_{(r,s)}} : \mathbb{R}^2 \rightarrow \mathbb{C}$ is (jointly) analytic.
\end{lemma}

\begin{proof}[Proof of Lemma~\ref{lemma:appendix_joint_analyticity}]
Let $C_0 < \min_{1 \le j \le n} C_j$ be a positive constant and define $\Omega = \{z=(z_r,z_s) \in \mathbb{C}^2 : \lvert \Im z_j \rvert < C_0 , \, j \in \{r,s\}\}$ be an open set in $\mathbb{C}^2$, where $\Im z_j$ denotes the imaginary part of $z_j$. Consider
\[ \psi^*_{\eps_{(r,s)}} : z \in \Omega \mapsto \int_{\mathbb{R}^2} e^{iz^\top\epsilon} \, dF_{\eps_{(r,s)}}(\epsilon) . \]
To show that $\psi_{\eps_{(r,s)}}$ is analytic on $\mathbb{R}^2$, it suffices to show that $\psi^*_{\eps_{(r,s)}}$ is analytic on the open set $\Omega$ (since $\mathbb{R}^2 \subset \Omega$). By Hartogs' theorem on separate analyticity (cf.~\cite{hormander1973introduction}, Theorem~2.2.8), it suffices to show that $\psi^*_{\eps_{(r,s)}}$ is separately analytic on a strip $\{z_r \in \mathbb{C} : \lvert\Im z_r\rvert < C_0\}$ for any fixed value of $z_s$, and vice versa on a strip for $z_s$ for any fixed value of $z_r$. The remainder of the proof shows that $\psi^*_{\eps_{(r,s)}}$ is (1) indeed well-defined on $\Omega$ and (2) separately analytic with respect to each variable. \par

For any complex vector $z \in \mathbb{C}^d$, let $\Re z \in \mathbb{R}^d$ denote the real part of the vector and $\Im z \in \mathbb{R}^d$ the imaginary part. To show that $\psi^*_{\eps_{(r,s)}}$ is well-defined on $\Omega$, it suffices to show that $\int_{\mathbb{R}^2} e^{-\Im z^\top\epsilon} f_{\eps_{(r,s)}}(\epsilon) \, d\epsilon < \infty$ for any $z \in \Omega$ since
\begin{align*}
  \psi^*_{\eps_{(r,s)}}(z) = \int_{\mathbb{R}^2} e^{-\Im z^\top\epsilon}e^{i\Re z^\top\epsilon} f_{\eps_{(r,s)}}(\epsilon) \, d\epsilon .
\end{align*}
From the definition of the joint density $f_{\eps_{(r,s)}}$ (see, e.g., (5.2.8) in \cite{david2003order}), there exists a positive constant $K_{n,r,s}$ such that for every $\epsilon = (\epsilon_r,\epsilon_s) \in \mathbb{R}^2$,
\begin{align*}
  f_{\eps_{(r,s)}}(\epsilon) \le K_{n,r,s} \sum_{1 \le k,\ell \le n} f_{\eps_k}(\epsilon_r) f_{\eps_\ell}(\epsilon_s) .
\end{align*}
where by assumption, $f_{\eps_k}(\epsilon) = O(e^{-C_k\lvert\epsilon\rvert})$ as $\lvert\epsilon\rvert \rightarrow \infty$ for every $k$. Hence, it follows from the fact $\lvert \Im z_r \rvert < C_0$ and $\lvert \Im z_s \rvert < C_0$ in $\Omega$ that
\begin{align*}
  \int_{\mathbb{R}^2} e^{-\Im z^\top\epsilon} f_{\eps_{(r,s)}}(\epsilon) \, d\epsilon \le K_{n,r,s} \sum_{1 \le k,\ell \le n} \int_{\mathbb{R}^2} e^{-\Im z_r\epsilon_r} f_{\eps_k}(\epsilon_r) e^{-\Im z_s\epsilon_s} f_{\eps_\ell}(\epsilon_s) \, d\epsilon < \infty,
\end{align*}
which concludes that $\psi^*_{\eps_{(r,s)}}$ is well-defined on $\Omega$. \par

Now we show that $\psi^*_{\eps_{(r,s)}}(\cdot,z_s)$ is analytic on $\Omega_r = \{z_r \in \mathbb{C} : \lvert\Im z_r\rvert < C_0\}$ for every $z_s$ in the domain. By Theorem~5.2 in \cite{stein2010complex}, it suffices to construct a sequence of analytic functions $\{\psi^*_m\}_m$ that converges uniformly to $\psi^*_{\eps_{(r,s)}}(\cdot,z_s)$ in every compact subset of $\Omega_r$. We show that the sequence
\begin{align*}
  \psi^*_m(z_r) = \int_{[-m,m]^2} e^{i(z_r\epsilon_r+z_s\epsilon_s)} \, dF_{\eps_{(r,s)}}(\epsilon)
\end{align*}
satisfies the criteria above. Since the region of integration is bounded, it follows from the dominated convergence theorem that for every $m$,
\begin{align*}
  \psi^*_m(z_r) = \int_{[-m,m]^2} \sum_{\ell=0}^\infty \frac{(iz^\top\epsilon)^\ell}{\ell!} \, dF_{\eps_{(r,s)}}(\epsilon)
  = \sum_{\ell=0}^\infty \int_{[-m,m]^2} \frac{(iz^\top\epsilon)^\ell}{\ell!} \, dF_{\eps_{(r,s)}}(\epsilon)
\end{align*}
is entire on the complex plane and thus analytic on $\Omega_r \subset \mathbb{C}$. Further,
\begin{align*}
  \left\lvert \psi^*_{\eps_{(r,s)}}(z_r,z_s)-\psi^*_m(z_r) \right\rvert
    &\le \int_{\mathbb{R}^2\backslash[-m,m]^2} \left\lvert e^{iz^\top\epsilon} \right\rvert \, dF_{\eps_{(r,s)}}(\epsilon) \\
    &= \int_{\mathbb{R}^2\backslash[-m,m]^2} e^{-\Im z^\top\epsilon} \, dF_{\eps_{(r,s)}}(\epsilon) \\
    &\le K_{n,r,s} \sum_{1 \le k,\ell \le n} \int_{\mathbb{R}^2\backslash[-m,m]^2} e^{-\Im z^\top\epsilon} f_{\eps_k}(\epsilon_r) f_{\eps_\ell}(\epsilon_s) \, d\epsilon \\
    &\le K_{n,r,s} \sum_{1 \le k,\ell \le n} \int_{\mathbb{R} \times (\mathbb{R}\backslash[-m,m])} e^{-\Im z^\top\epsilon} f_{\eps_k}(\epsilon_r) f_{\eps_\ell}(\epsilon_s) \, d\epsilon \\
    &\qquad + K_{n,r,s} \sum_{1 \le k,\ell \le n} \int_{(\mathbb{R}\backslash[-m,m]) \times \mathbb{R}} e^{-\Im z^\top\epsilon} f_{\eps_k}(\epsilon_r) f_{\eps_\ell}(\epsilon_s) \, d\epsilon .
\end{align*}
Since $f_{\eps_k}(\epsilon_r) = O(e^{-C_k\lvert\epsilon_r\rvert})$ for every $k$ by assumption and $\lvert \Im z_r \rvert < C_0$, for any compact subset $D \subset \Omega_r$, we have
\begin{align*}
  \sup_{z_r \in D} K_k(z_r)
    \coloneqq \sup_{z_r \in D} \int_{\mathbb{R}} e^{-\Im z_r\epsilon_r} f_{\eps_k}(\epsilon_r) \, d\epsilon_r
    \le \int_{\mathbb{R}} e^{\sup_{z_r \in D} \lvert\Im z_r\rvert\epsilon_r} f_{\eps_k}(\epsilon_r) \, d\epsilon_r
    < \infty .
\end{align*}
In addition, for $m$ sufficiently large (independent of $z_r$), there exists a constant $K_{\ell,0}$ such that
\begin{align*}
  &\int_{\mathbb{R} \times (\mathbb{R}\backslash[-m,m])} e^{-\Im z^\top\epsilon} f_{\eps_k}(\epsilon_r) f_{\eps_\ell}(\epsilon_s) \, d\epsilon \\
    &\quad \le K_k(z_r) K_0 \int_{\mathbb{R}\backslash[-m,m]} e^{-C_0\lvert\epsilon_s\rvert-\Im z_s\epsilon_s} \, d\epsilon_s . \\
    &\quad = K_k(z_r) K_0 \left( \int_{-\infty}^{-k} e^{-(C_0-\Im z_s)\lvert\epsilon_s\rvert} \, d\epsilon_s + \int_{k}^{\infty} e^{-(C_0+\Im z_s)\lvert\epsilon_s\rvert} \, d\epsilon_s \right)
    \longrightarrow 0 ,
\end{align*}
as $m \rightarrow \infty$. Similarly, we have
\begin{align*}
  \int_{(\mathbb{R}\backslash[-m,m]) \times \mathbb{R}} e^{-\Im z^\top\epsilon} f_{\eps_k}(\epsilon_r) f_{\eps_\ell}(\epsilon_s) \, d\epsilon \le K_\ell(z_r) K_{k,0} \int_{\mathbb{R}\backslash[-m,m]} e^{-C_0\lvert\epsilon_r\rvert-\Im z_r\epsilon_r} \, d\epsilon_r ,
\end{align*}
where, as $m \rightarrow \infty$,
\begin{align*}
  &\sup_{z_r \in D} \int_{\mathbb{R}\backslash[-m,m]} e^{-C_0\lvert\epsilon_r\rvert-\Im z_r\epsilon_r} \, d\epsilon_r \longrightarrow 0 .
\end{align*}
Therefore, we conclude that $\{\psi^*_m\}_m$ converges uniformly to $\psi^*_{\eps_{(r,s)}}(\cdot,z_s)$ on any compact subset of $\Omega_r$ for any fixed $z_s$ in the domain. Therefore, $\psi^*_{\eps_{(r,s)}}(\cdot,z_s)$ is analytic on $\Omega_r$. An analogous proof shows that $\psi^*_{\eps_{(r,s)}}(z_r,\cdot)$ is analytic on $\Omega_s = \{z_s \in \mathbb{C} : \lvert\Im z_s\rvert < C_0\}$ for any fixed $z_r$ in the domain. \par

This concludes the proof that $\psi_{\eps_{(r,s)}}$ is (jointly) analytic on $\Omega = \{ z = (z_r,z_s) \in \mathbb{C}^2 : z_r \in \Omega_r, z_s \in \Omega_s \}$.
\end{proof}


\begin{lemma}
\label{lemma:appendix_conditional_order_statistic}
Suppose the measurement errors satisfy Assumption~\ref{assumption:iid_case}(a) and \ref{assumption:bound}. Define, for every $\epsilon_s \in \mathbb{R}$ and $\epsilon_r > 0$,
\[
  F_{\eps_{(s\vert r)}}(\epsilon_s;\epsilon_r) = \frac{F_{\eps_{(r,s)}}(\epsilon_r,\epsilon_s)}{F_{\eps_{(r)}}(\epsilon_r)} .
\]
Then $\lim_{\epsilon_r \downarrow 0} F_{\eps_{(s\vert r)}}(\cdot;\epsilon_r) = F_{\eps_{s-r:n-r}}(\cdot)$, where the latter denotes the distribution of the $(s-r)$\textsuperscript{th} order statistic of a random sample with size $n-r$ from $F_\eps$.
\end{lemma}

\begin{proof}
When $\epsilon_s \le 0$, clearly $\lim_{\epsilon_r \downarrow 0} F_{\eps_{(s\vert r)}}(\epsilon_s;\epsilon_r) = F_{\eps_{s-r:n-r}}(\epsilon_s) = 0$ by Assumption~\ref{assumption:bound}. Thus, fix any $\epsilon_s > 0$ and consider small enough $\epsilon_r$ such that $\epsilon_s > \epsilon_r \downarrow 0$. \par

It follows from standard results (e.g., see (2.1.3) and (2.2.4) in \cite{david2003order}) that the joint and marginal distribution functions may be expressed as
\begin{align*}
  F_{\eps_{(s\vert r)}}(\epsilon_s;\epsilon_r)
    = \frac{\sum_{k=s}^n \sum_{j=r}^k C^n_{j,k-j,n-k} F_\eps(\epsilon_r)^j(F_\eps(\epsilon_s)-F_\eps(\epsilon_r))^{k-j}(1-F_\eps(\epsilon_s))^{n-k}}{\sum_{\ell=r}^n C^n_\ell F_\eps(\epsilon_r)^\ell(1-F_\eps(\epsilon_r))^{n-\ell}} ,
\end{align*}
where $C^n_{j,k-j,n-k}$ and $C^n_\ell$ are multinomial and binomial coefficients, respectively. Observe that
\begin{align*}
  F_{\eps_{(s\vert r)}}(\epsilon_s;\epsilon_r)
    &= \frac{\sum_{j=r}^n \sum_{k=\max(s,j)}^n C^n_{j,k-j,n-k} F_\eps(\epsilon_r)^j(F_\eps(\epsilon_s)-F_\eps(\epsilon_r))^{k-j}(1-F_\eps(\epsilon_s))^{n-k}}{\sum_{\ell=r}^n C^n_\ell F_\eps(\epsilon_r)^\ell(1-F_\eps(\epsilon_r))^{n-\ell}} \\
    &= \sum_{j=r}^n \frac{C^n_j F_\eps(\epsilon_r)^j(1-F_\eps(\epsilon_r))^{n-j}}{\sum_{\ell=r}^n C^n_\ell F_\eps(\epsilon_r)^\ell(1-F_\eps(\epsilon_r))^{n-\ell}} \\
    &\hspace{0.2\textwidth} \sum_{k=\max(s,j)}^n C^{n-j}_{k-j} \left(\frac{F_\eps(\epsilon_s)-F_\eps(\epsilon_r)}{1-F_\eps(\epsilon_r)}\right)^{k-j}\left(\frac{1-F_\eps(\epsilon_s)}{1-F_\eps(\epsilon_r)}\right)^{n-k} \\
    &= \sum_{j=r}^n \frac{C^n_j F_\eps(\epsilon_r)^j(1-F_\eps(\epsilon_r))^{n-j}}{\sum_{\ell=r}^n C^n_\ell F_\eps(\epsilon_r)^\ell(1-F_\eps(\epsilon_r))^{n-\ell}} \\
    &\hspace{0.2\textwidth} \sum_{k=\max(s-j,0)}^{n-j} C^{n-j}_{k} \left(\frac{F_\eps(\epsilon_s)-F_\eps(\epsilon_r)}{1-F_\eps(\epsilon_r)}\right)^k\left(\frac{1-F_\eps(\epsilon_s)}{1-F_\eps(\epsilon_r)}\right)^{n-j-k} .
\end{align*}
As $\epsilon_r \downarrow 0$, the weight component vanishes for all but $j=r$. In particular,
\begin{align*}
  \frac{C^n_j F_\eps(\epsilon_r)^j(1-F_\eps(\epsilon_r))^{n-j}}{\sum_{\ell=r}^n C^n_\ell F_\eps(\epsilon_r)^\ell(1-F_\eps(\epsilon_r))^{n-\ell}}
    & = \left(\sum_{\ell=r}^n \frac{C^n_\ell}{C^n_j} \left(\frac{F_\eps(\epsilon_r)}{1-F_\eps(\epsilon_r)}\right)^{\ell-j}\right)^{-1} \\
    & = \left(\sum_{\ell=r-j}^{n-j} \frac{C^n_{\ell+j}}{C^n_j} \left(\frac{F_\eps(\epsilon_r)}{1-F_\eps(\epsilon_r)}\right)^{\ell}\right)^{-1}
    \longrightarrow \begin{cases} 1 &\text{if } j = r , \\ 0 &\text{if } j \neq r . \end{cases}
\end{align*}
This implies that for any $\epsilon_s > 0$,
\begin{align*}
  \lim_{\epsilon_r \downarrow 0} F_{\eps_{(s\vert r)}}(\epsilon_s;\epsilon_r)
    = \sum_{k=s-r}^{n-r} C^{n-r}_{k} F_\eps(\epsilon_s)^k(1-F_\eps(\epsilon_s))^{n-r-k}
    = F_{\eps_{s-r:n-r}}(\epsilon_s) .
\end{align*}
Therefore, we conclude that $\lim_{\epsilon_r \downarrow 0} F_{\eps_{(s\vert r)}}(\cdot;\epsilon_r) = F_{\eps_{s-r:n-r}}(\cdot)$ on $\mathbb{R}$.
\end{proof}

\begin{remark}
The conditional distribution $F_{\eps_{(s\vert r)}}(\epsilon_s;\epsilon_r)$ is well-defined for any $\epsilon_r > 0$ since $F_{\eps_{(r)}}(\epsilon_r) > 0$ under the support normalization in Assumption~\ref{assumption:bound}, but $F_{\eps_{(r)}}(0)=0$. Lemma \ref{lemma:appendix_conditional_order_statistic} allows one to define $F_{\eps_{(s\vert r)}}(\cdot;0)$ by a continuous extension from above.
\end{remark}


Suppose there exists a known group structure for the measurement errors as in Assumption~\ref{assumption:group_structure}, where $n_q$ denotes the size of group $q$. Without loss of generality, let $(\eps_1,\ldots,\eps_n)$ be ordered such that the first $n_1$ random variables are from group $1$, next $n_2$ random variables from group $2$, etc. Further without loss of generality, let $g_1$ be a group with at least $n-r$ members, i.e., $n_1 \ge n-r$. Let $R_{(j)} = \{q : X_{(j)} = X_k \text{ for some } k \in g_q\}$ be the group identity of the $j$\textsuperscript{th} order statistic. We abuse notation and use $R_{(j)}$ to denote both the set and the a.s.~unique element in $\{1,\ldots,p\}$. Further, let $E_{q}$ denote the event where the top $n-r$ order statistics are all from group $1$ except for the $s$\textsuperscript{th} order statistic, which belongs to group $q$ (where it may be that $q=1$), i.e.,
\begin{align*}
  E_{q} = \{ R_{(r+1)}=1, \ldots, R_{(s-1)}=1, R_{(s)}=q, R_{(s+1)}=1, \ldots, R_{(n)}=1 \} .
\end{align*}
In the following two lemmas, we derive the distribution of order statistics conditional on the event $E_1$ and $E_q$ ($q \neq 1$), respectively.

\begin{lemma}
\label{lemma:appendix_conditional_cdf_1}
Let $(\eps_{(r)},\eps_{(s)})$ be order statistics from an independent but nonidentically distributed sample $(\eps_1,\ldots,\eps_n) \sim \times_{j=1}^n F_{\eps_j}$. Let $\zeta$ be the maximum of $\{\eps_{n-r+1},\ldots,\eps_n\}$. Then
\begin{align*}
  &\mathbb{P}(\eps_{(r)} \le \epsilon_r, \eps_{(s)} \le \epsilon_s \,\vert\, E_1) \\
  &\qquad\propto \sum_{k=s}^n \sum_{j=r}^k C^{n_1}_{j-r,k-j,n-k} \mathbb{E}_{\zeta} \Bigl(\I\{\zeta \le \epsilon_r\} (F_{\eps_1}(\epsilon_r)-F_{\eps_1}(\zeta))^{j-r} \Bigr) (F_{\eps_1}(\epsilon_s)-F_{\eps_1}(\epsilon_r))^{k-j} (1-F_{\eps_1}(\epsilon_s))^{n-k} ,
\end{align*}
and
\begin{align*}
  \mathbb{P}(\eps_{(r)} \le \epsilon_r \,\vert\, E_1) \propto \sum_{j=r}^n C^{n_1}_{j-r,n-j} \mathbb{E}_{\zeta} \Bigl( \I\{\zeta \le \epsilon_r\} (F_{\eps_1}(\epsilon_r)-F_{\eps_1}(\zeta))^{j-r} \Bigr) (1-F_{\eps_1}(\epsilon_r))^{n-j} .
\end{align*}
\end{lemma}

\begin{proof}
Note that
\begin{align*}
  \{ \eps_{(r)} \le \epsilon_r, \eps_{(s)} \le \epsilon_s , E_1 \}
    = \bigcup_{k=s}^n \bigcup_{j=r}^k \{ \eps_{(j)} \le \epsilon_r, \eps_{(j+1)} > \epsilon_r, \eps_{(k)} \le \epsilon_s, \eps_{(k+1)} > \epsilon_s , E_1 \} .
\end{align*}
Let $\mathbb{S}_{j,k}$ be a collection of reordered vectors of $(1,\ldots,n)$ where each vector $\sigma \in \mathbb{S}_{j,k}$ uniquely partition $\{1,\ldots,n\}$ in the sense that
\[
  \{1,\ldots,n\} = \{\sigma_1,\ldots,\sigma_r\} \cup \{\sigma_{r+1},\ldots,\sigma_j\} \cup \{\sigma_{j+1},\ldots,\sigma_k\} \cup \{\sigma_{k+1},\ldots,\sigma_n\} ,
\]
and $\sigma_{\ell} \in g_1$ for all $\ell \ge r+1$. In words, $\sigma$ divides group $1$ so that all top $n-r$ order statistics belong to group $1$ and are in the ranges $(-\infty,\epsilon_r]$, $(\epsilon_r,\epsilon_s]$, and $(\epsilon_s,\infty)$, respectively. The remaining members of group $1$ and all other groups are associated with the lowest $r$ order statistics. Then, we have
\begin{align*}
  &\{ \eps_{(r)} \le \epsilon_r, \eps_{(s)} \le \epsilon_s , E_1 \} \\
  &\qquad= \bigcup_{k=s}^n \bigcup_{j=r}^k \bigcup_{\sigma \in \mathbb{S}_{j,k}} \{\eps_{\sigma_{(r)}} \le \epsilon_r, \eps_{\sigma_{(r)}} < \eps_{r+1}^j \le \epsilon_r, \epsilon_r < \eps_{j+1}^k \le \epsilon_s, \epsilon_s < \eps_{k+1}^n\} ,
\end{align*}
where $\eps_{\sigma_{(r)}} = \max_{1 \le j \le r} \eps_{\sigma_r}$ and $\eps_{\ell}^{m}$ is a shorthand for the $\eps_{\sigma_\ell},\ldots,\eps_{\sigma_m}$. Denote by $\mathbb{E}_{\sigma_{(r)}}$ the expectation with respect to $\eps_{\sigma_{(r)}}$, which combines information of the distributions of all other groups. It follows that
\begin{align*}
  &\mathbb{P}(\eps_{(r)} \le \epsilon_r, \eps_{(s)} \le \epsilon_s , E_1) \\
  &\qquad= \sum_{k=s}^n \sum_{j=r}^k \sum_{\sigma \in \mathbb{S}_{j,k}} \mathbb{E}_{\sigma_{(r)}}\mathbb{P}(\eps_{\sigma_{(r)}} \le \epsilon_r, \eps_{\sigma_{(r)}} < \eps_{r+1}^j \le \epsilon_r, \epsilon_r < \eps_{j+1}^k \le \epsilon_s, \epsilon_s < \eps_{k+1}^n \,\vert\, \eps_{\sigma_{(r)}}) \\
  &\qquad= \sum_{k=s}^n \sum_{j=r}^k \sum_{\sigma \in \mathbb{S}_{j,k}} \mathbb{E}_{\sigma_{(r)}} \Bigl( \I\{\eps_{\sigma_{(r)}} \le \epsilon_r\} (F_{\eps_1}(\epsilon_r)-F_{\eps_1}(\eps_{\sigma_{(r)}}))^{j-r} (F_{\eps_1}(\epsilon_s)-F_{\eps_1}(\epsilon_r))^{k-j} (1-F_{\eps_1}(\epsilon_s))^{n-k} \Bigr) \\
  &\qquad= \sum_{k=s}^n \sum_{j=r}^k C^{n_1}_{j-r,k-j,n-k} \mathbb{E}_{\sigma_{(r)}} \Bigl( \I\{\eps_{\sigma_{(r)}} \le \epsilon_r\} (F_{\eps_1}(\epsilon_r)-F_{\eps_1}(\eps_{\sigma_{(r)}}))^{j-r} \Bigr) \\
  &\hspace{0.6\textwidth} (F_{\eps_1}(\epsilon_s)-F_{\eps_1}(\epsilon_r))^{k-j} (1-F_{\eps_1}(\epsilon_s))^{n-k} ,
\end{align*}
where $C^{n_1}_{j-r,k-j,n-k} = n_1!/[(n_1-(n-j))!(j-r)!(k-j)!(n-k)!]$ is the multinomial coefficient equal to the size of $\mathbb{S}_{j,k}$ (the number of ways to classify members of group $1$ to four sets that partition $\{1,\ldots,n\}$). The last equality holds because the distribution of $\eps_{\sigma_{(r)}}$ is invariant with respect to $\sigma$. This proves the original statement for the joint distribution in the lemma since $\zeta =_d \eps_{\sigma_{(r)}}$ because $\{\eps_{n-r+1},\ldots,\eps_n\}$ consists of all measurement errors except $r$ members from group $1$ as is the definition of $\eps_{\sigma_{(r)}}$. \par

An analogous derivation shows that
\begin{align*}
  \mathbb{P}(\eps_{(r)} \le \epsilon_r , E_1)
  = \sum_{j=r}^n C^{n_1}_{j-r,n-j} \mathbb{E}_{\sigma_{(r)}} \Bigl( \I\{\eps_{\sigma_{(r)}} \le \epsilon_r\} (F_{\eps_1}(\epsilon_r)-F_{\eps_1}(\eps_{\sigma_{(r)}}))^{j-r} \Bigr) (1-F_{\eps_1}(\epsilon_r))^{n-j} .
\end{align*}
\end{proof}


\begin{lemma}
\label{lemma:appendix_conditional_cdf_q}
Let $(\eps_{(r)},\eps_{(s)})$ be order statistics from an independent but nonidentically distributed sample $(\eps_1,\ldots,\eps_n) \sim \times_{j=1}^n F_{\eps_j}$. If $q \neq 1$, for $m \in g_q$ and $\zeta=(\zeta_r,\zeta_s)$ where $\zeta_r$ is the maximum of $\{\eps_{n-r},\ldots,\eps_n\} \backslash \{\eps_m\}$ and $\zeta_s=\eps_m$,
\begin{align*}
  &\mathbb{P}(\eps_{(r)} \le \epsilon_r, \eps_{(s)} \le \epsilon_s \,\vert\, E_q) \\
  &\qquad\propto \sum_{k=s}^n \sum_{j=r}^{s-1} n_q C^{n_1}_{j-r,s-j-1,k-s,n-k} \mathbb{E}_{\zeta} \Bigl( \I\{\zeta_r \le \epsilon_r < \zeta_s \le \epsilon_s\} (F_{\eps_1}(\epsilon_r)-F_{\eps_1}(\zeta_r))^{j-r} \\
  &\hspace{0.45\textwidth} (F_{\eps_1}(\zeta_s)-F_{\eps_1}(\epsilon_r))^{s-j-1} (F_{\eps_1}(\epsilon_s)-F_{\eps_1}(\zeta_s))^{k-s} \Bigr) (1-F_{\eps_1}(\epsilon_s))^{n-k} \\
  &\qquad\qquad+ \sum_{k=s}^n \sum_{j=s}^k n_q C^{n_1}_{s-r-1,j-s,k-j,n-k} \mathbb{E}_{\zeta} \Bigl( \I\{\zeta_r \le \zeta_s \le \epsilon_r\} (F_{\eps_1}(\zeta_s)-F_{\eps_1}(\zeta_r))^{s-r-1} \\
  &\hspace{0.5\textwidth} (F_{\eps_1}(\epsilon_r)-F_{\eps_1}(\zeta_s))^{j-s} \Bigr) (F_{\eps_1}(\epsilon_s)-F_{\eps_1}(\epsilon_r))^{k-j} (1-F_{\eps_1}(\epsilon_s))^{n-k} ,
\end{align*}
and
\begin{align*}
  &\mathbb{P}(\eps_{(r)} \le \epsilon_r \,\vert\, E_q) \\
  &\qquad\propto \sum_{j=r}^{s-1} n_q C^{n_1}_{j-r,s-j-1,n-s} \mathbb{E}_{\zeta} \Bigl( \I\{\zeta_r \le \epsilon_r < \zeta_s\} (F_{\eps_1}(\epsilon_r)-F_{\eps_1}(\zeta_r))^{j-r} \\
  &\hspace{0.5\textwidth} (F_{\eps_1}(\zeta_s)-F_{\eps_1}(\epsilon_r))^{s-j-1} (1-F_{\eps_1}(\zeta_s))^{n-s} \Bigr) \\
  &\qquad\qquad+ \sum_{j=s}^n n_q C^{n_1}_{s-r-1,j-s,n-j} \mathbb{E}_{\zeta} \Bigl( \I\{\zeta_r \le \zeta_s \le \epsilon_r\} (F_{\eps_1}(\zeta_s)-F_{\eps_1}(\zeta_r))^{s-r-1} \\
  &\hspace{0.55\textwidth} (F_{\eps_1}(\epsilon_r)-F_{\eps_1}(\zeta_s))^{j-s} \Bigr) (1-F_{\eps_1}(\epsilon_r))^{n-j} .
\end{align*}
\end{lemma}

\begin{proof}
Note that
\begin{align*}
  \{ \eps_{(r)} \le \epsilon_r, \eps_{(s)} \le \epsilon_s , E_q \}
    &= \bigcup_{k=s}^n \bigcup_{j=r}^k \{ \eps_{(j)} \le \epsilon_r, \eps_{(j+1)} > \epsilon_r, \eps_{(k)} \le \epsilon_s, \eps_{(k+1)} > \epsilon_s , E_q \} \\
    &= \left(\bigcup_{k=s}^n \bigcup_{j=r}^{s-1} \{ \eps_{(j)} \le \epsilon_r, \eps_{(j+1)} > \epsilon_r, \eps_{(k)} \le \epsilon_s, \eps_{(k+1)} > \epsilon_s , E_q \}\right) \\
    &\qquad \bigcup \left(\bigcup_{k=s}^n \bigcup_{j=s}^{k} \{ \eps_{(j)} \le \epsilon_r, \eps_{(j+1)} > \epsilon_r, \eps_{(k)} \le \epsilon_s, \eps_{(k+1)} > \epsilon_s , E_q \}\right) .
\end{align*}
The set is partitioned into two parts where $\eps_{(s)} > \epsilon_r$ (in the first part of the partition) and where $\eps_{(s)} \le \epsilon_r$ (in the second part of the partition). We partition the event into two different events because, as will be evident in the derivations below, the functional form for the probability of these events differs depending on the location of $\eps_{(s)}$ relative to $\epsilon_r$. \par

For $j < s$, let $\mathbb{S}^1_{j,k}$ be a collection of reordered vectors of $(1,\ldots,n)$ where each vector $\sigma \in \mathbb{S}_{j,k}$ uniquely partition $\{1,\ldots,n\}$ in the sense that
\[
  \{1,\ldots,n\} = \{\sigma_\ell\}_{\ell=1}^r \cup \{\sigma_\ell\}_{\ell=r+1}^j \cup \{\sigma_\ell\}_{\ell=j+1}^{s-1} \cup \{\sigma_s\} \cup \{\sigma_\ell\}_{\ell=s+1}^{k} \cup \{\sigma_\ell\}_{\ell=k+1}^{n} ,
\]
where $\sigma_s \in g_{q}$ and $\sigma_{\ell} \in g_1$ for all $\ell \ge r+1$ such that $\ell \neq s$. In words, $\sigma$ divides group $1$ so that all top $n-r$ order statistics except for the $s$\textsuperscript{th} belong to group $1$ and are in the ranges $(-\infty,\epsilon_r]$, $(\epsilon_r,\eps_{\sigma_s}]$ $(\eps_{\sigma_s},\epsilon_s]$, and $(\epsilon_s,\infty)$, respectively. The $s$\textsuperscript{th} order statistic $\eps_{\sigma_s}$ is in $(\epsilon_r,\epsilon_s]$ (because $j < s \le k$), and the remaining members are associated with the lowest $r$ order statistics. \par

Similarly, for $j \ge s$, let $\mathbb{S}^2_{j,k}$ be a collection of reordered vectors of $(1,\ldots,n)$ where each vector $\sigma \in \mathbb{S}_{j,k}$ uniquely partition $\{1,\ldots,n\}$ in the sense that
\[
  \{1,\ldots,n\} = \{\sigma_\ell\}_{\ell=1}^r \cup \{\sigma_\ell\}_{\ell=r+1}^{s-1} \cup \{\sigma_s\} \cup \{\sigma_\ell\}_{\ell=s+1}^{j} \cup \{\sigma_\ell\}_{\ell=j+1}^{k} \cup \{\sigma_\ell\}_{\ell=k+1}^{n} ,
\]
where $\sigma_s \in g_{q}$ and $\sigma_{\ell} \in g_1$ for all $\ell \ge r+1$ such that $\ell \neq s$. In words, $\sigma$ divides group $1$ so that all top $n-r$ order statistics except for the $s$\textsuperscript{th} belong to group $1$ and are in the ranges $(-\infty,\eps_{\sigma_s}]$, $(\eps_{\sigma_s},\epsilon_r]$ $(\epsilon_r,\epsilon_s]$, and $(\epsilon_s,\infty)$, respectively. The $s$\textsuperscript{th} order statistic $\eps_{\sigma_s}$ is in $(-\infty,\epsilon_r]$ (because $s \le j$), and the remaining members are associated with the lowest $r$ order statistics. \par

It follows that
\begin{align*}
  &\{ \eps_{(r)} \le \epsilon_r, \eps_{(s)} \le \epsilon_s , E_q \} \\
  &\qquad= \Biggl(\bigcup_{k=s}^n \bigcup_{j=r}^{s-1} \bigcup_{\sigma \in \mathbb{S}^1_{j,k}} \{\eps_{\sigma_{(r)}} \le \epsilon_r, \eps_{\sigma_{(r)}} < \eps_{r+1}^j \le \epsilon_r, \epsilon_r < \eps_{j+1}^{s-1} \le \eps_{\sigma_s}, \\
  &\hspace{0.45\textwidth} \epsilon_r < \eps_{\sigma_s} \le \epsilon_s, \eps_{\sigma_s} \le \eps_{s+1}^k \le \epsilon_s, \epsilon_s < \eps_{k+1}^n\}\Biggr) \\
  &\qquad\qquad \bigcup \Biggl(\bigcup_{k=s}^n \bigcup_{j=s}^k \bigcup_{\sigma \in \mathbb{S}^2_{j,k}} \{\eps_{\sigma_{(r)}} \le \eps_{\sigma_s}, \eps_{\sigma_{(r)}} < \eps_{r+1}^{s-1} \le \eps_{\sigma_s}, \eps_{\sigma_s} \le \epsilon_r, \\
  &\hspace{0.5\textwidth} \eps_{\sigma_s} < \eps_{s+1}^{j} \le \epsilon_r, \epsilon_r < \eps_{j+1}^k \le \epsilon_s, \epsilon_s < \eps_{k+1}^n\}\Biggr) ,
\end{align*}
where $\eps_{\sigma_{(r)}} = \max_{1 \le j \le r} \eps_{\sigma_r}$ and $\eps_{\ell}^{m}$ is a shorthand for the $\eps_{\sigma_\ell},\ldots,\eps_{\sigma_m}$. Denote by $\mathbb{E}_{\sigma_{(r),s}}$ the expectation with respect to $(\eps_{\sigma_{(r)}},\eps_{\sigma_s})$. Then, we have
\begin{align*}
  &\mathbb{P}(\eps_{(r)} \le \epsilon_r, \eps_{(s)} \le \epsilon_s , E_q) \\
  &\qquad= \sum_{k=s}^n \sum_{j=r}^{s-1} \sum_{\sigma \in \mathbb{S}^1_{j,k}} \mathbb{E}_{\sigma_{(r),s}} \Bigl( \I\{\eps_{\sigma_{(r)}} \le \epsilon_r < \eps_{\sigma_s} \le \epsilon_s\} (F_{\eps_1}(\epsilon_r)-F_{\eps_1}(\eps_{\sigma_{(r)}}))^{j-r} \\
  &\hspace{0.35\textwidth} (F_{\eps_1}(\eps_{\sigma_s})-F_{\eps_1}(\epsilon_r))^{s-j-1} (F_{\eps_1}(\epsilon_s)-F_{\eps_1}(\eps_{\sigma_s}))^{k-s} (1-F_{\eps_1}(\epsilon_s))^{n-k} \Bigr) \\
  &\qquad\qquad+ \sum_{k=s}^n \sum_{j=s}^k \sum_{\sigma \in \mathbb{S}^2_{j,k}} \mathbb{E}_{\sigma_{(r),s}} \Bigl( \I\{\eps_{\sigma_{(r)}} \le \eps_{\sigma_s} \le \epsilon_r\} (F_{\eps_1}(\eps_{\sigma_s})-F_{\eps_1}(\eps_{\sigma_{(r)}}))^{s-r-1} \\
  &\hspace{0.4\textwidth} (F_{\eps_1}(\epsilon_r)-F_{\eps_1}(\eps_{\sigma_s}))^{j-s} (F_{\eps_1}(\epsilon_s)-F_{\eps_1}(\epsilon_r))^{k-j} (1-F_{\eps_1}(\epsilon_s))^{n-k} \Bigr) \\
  &\qquad= \sum_{k=s}^n \sum_{j=r}^{s-1} n_q C^{n_1}_{j-r,s-j-1,k-s,n-k} \mathbb{E}_{\sigma_{(r),s}} \Bigl( \I\{\eps_{\sigma_{(r)}} \le \epsilon_r < \eps_{\sigma_s} \le \epsilon_s\} (F_{\eps_1}(\epsilon_r)-F_{\eps_1}(\eps_{\sigma_{(r)}}))^{j-r} \\
  &\hspace{0.35\textwidth} (F_{\eps_1}(\eps_{\sigma_s})-F_{\eps_1}(\epsilon_r))^{s-j-1} (F_{\eps_1}(\epsilon_s)-F_{\eps_1}(\eps_{\sigma_s}))^{k-s} \Bigr) (1-F_{\eps_1}(\epsilon_s))^{n-k} \\
  &\qquad\qquad+ \sum_{k=s}^n \sum_{j=s}^k n_q C^{n_1}_{s-r-1,j-s,k-j,n-k} \mathbb{E}_{\sigma_{(r),s}} \Bigl( \I\{\eps_{\sigma_{(r)}} \le \eps_{\sigma_s} \le \epsilon_r\} (F_{\eps_1}(\eps_{\sigma_s})-F_{\eps_1}(\eps_{\sigma_{(r)}}))^{s-r-1} \\
  &\hspace{0.4\textwidth} (F_{\eps_1}(\epsilon_r)-F_{\eps_1}(\eps_{\sigma_s}))^{j-s} \Bigr) (F_{\eps_1}(\epsilon_s)-F_{\eps_1}(\epsilon_r))^{k-j} (1-F_{\eps_1}(\epsilon_s))^{n-k} .
\end{align*}
The last equality holds because the distribution of $(\eps_{\sigma_{(r)}},\eps_{\sigma_s})$ is invariant with respect to $\sigma$. This proves the original statement for the joint distribution in the lemma since $\zeta_r =_d \eps_{\sigma_{(r)}}$ because $\{\eps_{n-r},\ldots,\eps_n\} \backslash \{\eps_m\}$ consists of all measurement errors except $r$ members from group $1$ and $1$ member from group $q$---as is the definition of $\eps_{\sigma_{(r)}}$---and $\zeta_s =_d \eps_m \sim F_q$. \par

An analogous derivation shows that
\begin{align*}
  &\mathbb{P}(\eps_{(r)} \le \epsilon_r , E_q) \\
  &\qquad= \sum_{j=r}^{s-1} n_q C^{n_1}_{j-r,s-j-1,n-s} \mathbb{E}_{\sigma_{(r),s}} \Bigl( \I\{\eps_{\sigma_{(r)}} \le \epsilon_r < \eps_{\sigma_s}\} (F_{\eps_1}(\epsilon_r)-F_{\eps_1}(\eps_{\sigma_{(r)}}))^{j-r} \\
  &\hspace{0.55\textwidth} (F_{\eps_1}(\eps_{\sigma_s})-F_{\eps_1}(\epsilon_r))^{s-j-1} (1-F_{\eps_1}(\eps_{\sigma_s}))^{n-s} \Bigr) \\
  &\qquad\qquad+ \sum_{j=s}^n n_q C^{n_1}_{s-r-1,j-s,n-j} \mathbb{E}_{\sigma_{(r),s}} \Bigl( \I\{\eps_{\sigma_{(r)}} \le \eps_{\sigma_s} \le \epsilon_r\} (F_{\eps_1}(\eps_{\sigma_s})-F_{\eps_1}(\eps_{\sigma_{(r)}}))^{s-r-1} \\
  &\hspace{0.6\textwidth} (F_{\eps_1}(\epsilon_r)-F_{\eps_1}(\eps_{\sigma_s}))^{j-s} \Bigr) (1-F_{\eps_1}(\epsilon_r))^{n-j} .
\end{align*}
\end{proof}


\subsection{Proofs of Main Results}


\begin{proof}[Proof of Lemma~\ref{lemma:rationalizability}]
Let $F$ be a data-consistent measurement error distribution satisfying Assumption~\ref{assumption:iid_case} and $\eta \sim F$. We first prove the second part of the lemma. Note that the distribution of $j$\textsuperscript{th} order statistic $\eta_{(j)}$ is uniquely determined by $F$ and, by independence,
\begin{align*}
  \psi_\xi(t;F) = \psi_{X_{(j)}}(t)/\psi_{\eta_{(j)}}(t) , \quad \text{for all } t \in \mathbb{R} ,
\end{align*}
for any $j \in \{r,s\}$, where $\psi_\xi(t;F)$ is the induced latent variable distribution implied by $F$. Since $\psi_{\eta_{(j)}}$ is analytic (Lemma~\ref{lemma:appendix_joint_analyticity}), it has isolated real zeros. Thus, by the continuity of $\psi_\xi(\cdot;F)$, the equality is defined by the continuous extension at $t_0$ whenever $\psi_{\eta_{(j)}}(t_0) = 0$. \par
Now consider the first part of the lemma and note that because $\xi$ is independent of the measurement errors, we have
\[
  \psi_{X_{(r,s)}}(t_r,t_s) = \psi_\xi(t_r+t_s;F)\psi_{\eta_{(r,s)}}(t_r,t_s) ,
\]
and likewise for the marginal ch.f. Since $\psi_\xi(\cdot;F)$ is a ch.f., there exists $t_\xi > 0$ such that $\psi_\xi(t;F) \neq 0$ for all $t \in (-t_\xi,t_\xi)$. Similarly, there exists $t_\eta > 0$ such that $\psi_{\eta_{(j)}}(t) \neq 0$ for all $t \in (-t_\eta,t_\eta)$. Pick any positive $t_0 \le \min(t_\xi,t_\eta)$ and let $B_0$ be an open ball around zero contained in $\{(t_r,t_s) \in \mathbb{R}^2 : |t_r+t_s| < t_0\}$. Thus, $\psi_\xi(t_r+t_s;F) \neq 0$ and $\psi_{\eta_{(j)}}(t_r+t_s) \neq 0$ for all $(t_r,t_s) \in B_0$. Then, on $B_0$, we have
\begin{align*}
  \frac{\psi_{X_{(r,s)}}(t_r,t_s)}{\psi_{X_{(j)}}(t_r+t_s)}
  = \frac{\psi_\xi(t_r+t_s;F)\psi_{\eta_{(r,s)}}(t_r,t_s)}{\psi_\xi(t_r+t_s;F)\psi_{\eta_{(j)}}(t_r+t_s)}
  = \frac{\psi_{\eta_{(r,s)}}(t_r,t_s)}{\psi_{\eta_{(j)}}(t_r+t_s)} .
\end{align*}
This concludes the proof.
\end{proof}


\begin{proof}[Proof of Lemma~\ref{lemma:observational_equivalence_eps}]
By Lemma~\ref{lemma:rationalizability}, the distribution functions $F$ and $G$ are data-consistent only if
\begin{align}
  \psi_{\eta_{(r,s)}}(t_r,t_s) \psi_{\eta'_{(j)}}(t_r+t_s) = \psi_{\eta'_{(r,s)}}(t_r,t_s) \psi_{\eta_{(j)}}(t_r+t_s) , \label{eq:sum_of_os}
\end{align}
for $j \in \{r,s\}$ and for all $(t_r,t_s) \in B_0$. Observe that the two products in \eqref{eq:sum_of_os} are ch.f.s of $Z_{1j}$ and $Z_{2j}$, respectively. By Assumption Lemma~\ref{lemma:appendix_joint_analyticity}, all ch.f.s in \eqref{eq:sum_of_os} are analytic. Since the product of two analytic functions is also analytic, the equality of ch.f.s on $B_0$ implies their equality on all of $\mathbb{R}^2$. Hence, $F$ and $G$ are both data-consistent only if the condition in \eqref{eq:observational_equivalence} holds for $j \in \{r,s\}$.
\end{proof}


\begin{proof}[Proof of Lemma~\ref{lemma:limit_of_conditional_distribution}]
We only prove the equality in \eqref{eqn:limitF} as the proof for \eqref{eqn:limitG} is analogous. Let $F_{(j)}$ and $f_{(j)}$ (resp., $G_{(j)}$ and $g_{(j)}$) denote the marginal distribution and density function of the $j$\textsuperscript{th} order statistic of a random sample with size $n$ from $F$ (resp. $G$), respectively. Also, we denote the joint distribution and density functions of order statistics from $F$ by $F_{(r,s)}$ and $f_{(r,s)}$. Further, for any $y_r \ge 0$, define $F_{(s\vert r)}(\cdot;y_r)$ to be the distribution of the $s$\textsuperscript{th} order statistic conditional on the cumulative event that the $r$\textsuperscript{th} order statistic $\eta_{(r)} \le y_r$ as in Lemma~\ref{lemma:appendix_conditional_order_statistic}. \par

For $c \le 0$, the equality in \eqref{eqn:limitF} is an obvious consequence of the normalization in Assumption~\ref{assumption:bound}. Fix any $c > 0$ and consider a small enough $\delta$ such that $c > \delta \downarrow 0$. The conditional probability in \eqref{eqn:limitF} is a weighted average of the conditional distribution $F_{(s\vert r)}$:
\begin{align*}
  \mathbb{P}(\eta_{(r)}^\prime + \eta_{(s)} \le c \ | \ \eta_{(r)} + \eta_{(r)}^\prime \le \delta)
    &= \frac{\mathbb{P}(\eta_{(r)} + \eta_{(r)}^\prime \le \delta \ , \ \eta_{(s)} + \eta_{(r)}^\prime \le c)}{\mathbb{P}(\eta_{(r)} + \eta_{(r)}^\prime \le \delta)} \\
    &= \frac{\int_0^\delta F_{(r,s)}(\delta-x,c-x)g_{(r)}(x) \, dx}{\int_0^\delta F_{(r)}(\delta-x)g_{(r)}(x) \, dx} \\
    &= \int_0^\delta \frac{F_{(r)}(\delta-x)g_{(r)}(x)}{\int_0^\delta F_{(r)}(\delta-x)g_{(r)}(x) \, dx} F_{(s\vert r)}(c-x;\delta-x) \, dx .
\end{align*}
Note that since $F$ is absolutely continuous by Assumption~\ref{assumption:iid_case}, $F_{(s\vert r)}(y_s;\cdot)$ is continuous on $(0,\infty)$ for every $y_s \in \mathbb{R}$. Further, by the definition of $F_{(s\vert r)}(y_s;0)$ as in Lemma~\ref{lemma:appendix_conditional_order_statistic}, we conclude that $F_{(s\vert r)}(y_s;\cdot)$ is continuous on $[0,\infty)$ for every $y_s \in \mathbb{R}$. Likewise, for every $y_r \in [0,\infty)$, $F_{(s\vert r)}(\cdot;y_r)$ is continuous. Finally, it follows from the monotonicity of $F_{(s\vert r)}(\cdot;y_r)$ for every $y_r \in [0,\infty)$ that the function $F_{(s\vert r)}(\cdot;\cdot)$ is (jointly) continuous on $\mathbb{R} \times [0,\infty)$ (see, e.g., \cite{kruse1969joint}). \par

Therefore, it follows that as $\delta \downarrow 0$,
\begin{align*}
  &\int_0^\delta \frac{F_{(r)}(\delta-x)g_{(r)}(x)}{\int_0^\delta F_{(r)}(\delta-x)g_{(r)}(x) \, dx} F_{(s\vert r)}(c-x;\delta-x) \, dx \\
  &\qquad \le F_{(s\vert r)}(c;\delta) + \int_0^\delta \frac{F_{(r)}(\delta-x)g_{(r)}(x)}{\int_0^\delta F_{(r)}(\delta-x)g_{(r)}(x) \, dx} \left\lvert F_{(s\vert r)}(c-x;\delta-x) - F_{(s\vert r)}(c;\delta) \right\rvert \, dx \\
  &\qquad \le F_{(s\vert r)}(c;\delta) + \int_0^\delta \frac{F_{(r)}(\delta-x)g_{(r)}(x)}{\int_0^\delta F_{(r)}(\delta-x)g_{(r)}(x) \, dx} \sup_{0 \le x \le \delta} \left\lvert F_{(s\vert r)}(c-x;\delta-x) - F_{(s\vert r)}(c;\delta) \right\rvert \, dx \\
  &\qquad = F_{(s\vert r)}(c;\delta) + \sup_{0 \le x \le \delta} \left\lvert F_{(s\vert r)}(c-x;\delta-x) - F_{(s\vert r)}(c;\delta) \right\rvert \\
  &\qquad \longrightarrow F_{(s\vert r)}(c;0) ,
\end{align*}
The convergence of the upper bound is due to the (joint) continuity of $F_{(s\vert r)}(\cdot;\cdot)$. \par

A similar derivation for the lower bound shows that the lower bound approaches the same limit. By the squeeze theorem and Lemma~\ref{lemma:appendix_conditional_order_statistic}, we conclude that for any $c > 0$,
\[
  \mathbb{P}(\eta_{(r)}^\prime + \eta_{(s)} \le c \ | \ \eta_{(r)} + \eta_{(r)}^\prime \le \delta)
    \longrightarrow F_{(s\vert r)}(c;0)
    = F_{s-r:n-r}(c) .
\]
Thus, we conclude that the statement of the lemma holds for all $c \in \mathbb{R}$.
\end{proof}


\begin{proof}[Proof of Theorem~\ref{theorem:identification_iid_case}]
Lemmas~\ref{lemma:observational_equivalence_eps} and \ref{lemma:limit_of_conditional_distribution} imply that any two data-consistent measurement errors satisfying Assumptions~\ref{assumption:iid_case} and \ref{assumption:bound} must have the same distribution of order statistics:
\[
  F_{s-r:n-r} = G_{s-r:n-r} .
\]
It then follows from the one-to-one mapping between the distribution of an order statistic and the parent distribution (see, e.g., \cite{david2003order}, p.10, (2.1.5)) that $F = G$, i.e., the measurement error distribution $F_\eps = F = G$ is identified. Therefore, by Lemma~\ref{lemma:rationalizability}, the latent variable distribution $F_\xi$ is also identified.
\end{proof}


\begin{proof} [Proof of Corollary~\ref{corollary:observational_equivalence_eps}]
The result follows directly from \eqref{eq:observational_equivalence} in Lemma~\ref{lemma:observational_equivalence_eps}. Applying linear transformations $T_r : (z_1,z_2) \mapsto (z_2-z_1,z_2)$ and $T_s : (z_1,z_2) \mapsto (z_2-z_1,z_1)$ to \eqref{eq:observational_equivalence} for $j \in \{r,s\}$, respectively, shows that 
\begin{align*}
  \begin{pmatrix} \eta_{(s)} - \eta_{(r)} \\ \eta_{(r)}^\prime + \eta_{(s)} \end{pmatrix}
  \stackrel{d}{=} \begin{pmatrix} \eta_{(s)}^\prime - \eta_{(r)}^\prime \\ \eta_{(r)} + \eta_{(s)}^\prime \end{pmatrix}
  \quad \text{and} \quad
  \begin{pmatrix} \eta_{(s)} - \eta_{(r)} \\ \eta_{(r)} + \eta_{(s)}^\prime \end{pmatrix}
  \stackrel{d}{=} \begin{pmatrix} \eta_{(s)}^\prime - \eta_{(r)}^\prime \\ \eta_{(r)}^\prime + \eta_{(s)} \end{pmatrix} .
\end{align*}
\end{proof}


\begin{proof}[Proof of Theorem~\ref{theorem:identification_noniid_group}]
Without loss of generality, let $g_1$ be a group with at least $n-r$ members. For any $q \in \{1,\ldots,p\}$ (see Assumption~\ref{assumption:group_structure}), let $E_{q}$ denote the event where the top $n-r$ order statistics are all from group $1$ except for the $s$\textsuperscript{th} order statistic, which belong to group $q$ (where it may be that $q=1$), i.e.,
\begin{align} \label{eq:group_identifier}
  E_{q} = \{ R_{(r+1)}=1, \ldots, R_{(s-1)}=1, R_{(s)}=q, R_{(s+1)}=1, \ldots, R_{(n)}=1 \} .
\end{align}
The identification proof proceeds in three steps. First, we claim that the distribution $F_{\epsilon_j}$ for $j \in g_1$---the distribution of a group with many members---is identified. Then we show that $F_{\epsilon_j}$ for $j \in g_q$ is identified for each $q \neq 1$. Finally, $F_\xi$ is identified by a standard deconvolution argument. \par

\textbf{Step 1.} For notational simplicity, suppose $1 \in g_1$. A close inspection of the proof of Lemmas~\ref{lemma:rationalizability} and \ref{lemma:observational_equivalence_eps} reveals that the necessary condition \eqref{eq:observational_equivalence} does not rely on the hypothesis that the measurement errors are identically distributed. Therefore, we can make use of a similar argument as in the proof of Lemmas~\ref{lemma:rationalizability} and \ref{lemma:observational_equivalence_eps}, provided the ch.f.~of $(\eps_{(r)},\eps_{(s)})$ is analytic in the i.n.i.d.~setting. Indeed, Lemma~\ref{lemma:appendix_joint_analyticity} confirms $f_{(r,s)}(\cdot,\cdot;E_1)$ is analytic. Therefore, analogous to the proof of Lemmas~\ref{lemma:rationalizability} and \ref{lemma:observational_equivalence_eps} but conditional on the event $E_1$, two data-consistent measurement errors $\eta = (\eta_1,\ldots,\eta_n) \sim \times_{j=1}^n F_j$ and $\eta' = (\eta'_1,\ldots,\eta'_n) \sim \times_{j=1}^n G_j$ satisfying Assumption~\ref{assumption:noniid_case} must have the same joint distribution of sums:
\[
  \begin{pmatrix} \eta_{(j)}^\prime + \eta_{(r)} \\ \eta_{(j)}^\prime + \eta_{(s)} \end{pmatrix}
  \stackrel{d}{=}
  \begin{pmatrix} \eta_{(j)} + \eta_{(r)}^\prime \\ \eta_{(j)} + \eta_{(s)}^\prime \end{pmatrix}
  \, \Big\vert \, E_1^\eta \cap E_1^{\eta'} , \quad j \in \{r,s\} ,
\]
where $E_1^\eta$, as in \eqref{eq:group_identifier}, denotes the event that identifies the group association of order statistics that originate from the random vector $\eta$; and likewise for $E_1^{\eta'}$. Consider the left-hand side with $j=r$. Following a similar derivation as in the proof of Lemma~\ref{lemma:limit_of_conditional_distribution}, for any $c \in \mathbb{R}$, we have
\begin{align*}
  &\mathbb{P}(\eta_{(r)}^\prime + \eta_{(s)} \le c \ | \ \eta_{(r)} + \eta_{(r)}^\prime \le \delta, E_1^\eta \cap E_1^{\eta'}) \\
    &\qquad \le F_{s-r:n-r}(c) + \sup_{0 \le x \le \delta} \left\lvert F_{(s\vert r)}(c-x;\delta-x, E_1^\eta) - F_{s-r:n-r}(c) \right\rvert ,
\end{align*}
where $F_{s-r:n-r}$ is the distribution of the $(s-r)$\textsuperscript{th} order statistic of a random sample with size $n-r$ from the parent distribution $F_1$. To show that the limit of the left-hand side as $\delta \downarrow 0$ is indeed $F_{s-r:n-r}(c)$, it suffices to show, in combination with a similar lower bound, that $F_{(s \vert r)}(\cdot;\cdot,E_1^\eta)$ is continuous and $\lim_{\delta \downarrow 0} F_{(s \vert r)}(c;\delta,E_1^\eta) = F_{s-r:n-r}(c)$. Following the same proof as in Lemma~\ref{lemma:appendix_conditional_order_statistic} using the expressions in Lemma~\ref{lemma:appendix_conditional_cdf_1}, we show that
\begin{align*}
  \lim_{\delta \downarrow 0} F_{(s \vert r)}(c;\delta,E_1^\eta) = \sum_{k=s}^n C^{n-r}_{k-r} F_1(c)^{k-r}(1-F_1(c))^{n-k} ,
\end{align*}
for every $c \in \mathbb{R}$. The right-hand side equals $F_{s-r:n-r}(c)$. As in the proof of Lemma~\ref{lemma:limit_of_conditional_distribution}, the continuity of $F_{(s \vert r)}(\cdot;\cdot,E_1^\eta)$ on $\mathbb{R} \times [0,\infty)$ follows by elementwise continuity of $F_{(s \vert r)}(\cdot;\cdot,E_1^\eta)$ and monotonicity of $F_{(s \vert r)}(\cdot;\epsilon_r,E_1^\eta)$ for every $\epsilon_r \in [0,\infty)$ (cf.~\cite{kruse1969joint}). Therefore, we conclude that for any $c \in \mathbb{R}$,
\[
  \lim_{\delta \downarrow 0} \mathbb{P}(\eta_{(r)}^\prime + \eta_{(s)} \le c \ | \ \eta_{(r)} + \eta_{(r)}^\prime \le \delta, E_1^\eta \cap E_1^{\eta'}) = F_{s-r:n-r}(c) .
\]
A similar derivation for the right-hand side shows that
\[
  \lim_{\delta \downarrow 0} \mathbb{P}(\eta_{(r)} + \eta_{(s)}^\prime \le c \ | \ \eta_{(r)} + \eta_{(r)}^\prime \le \delta, E_1^\eta \cap E_1^{\eta'}) = \sum_{k=s}^n C^{n-r}_{k-r} G_1(c)^{k-r}(1-G_1(c))^{n-k} =: G_{s-r:n-r}(c) ,
\]
for every $c$. Therefore, we conclude that $F_{s-r:n-r} = G_{s-r:n-r}$ and thus $F_1 = G_1$, i.e., the distribution for group $1$ (a group with large size) is identified. \par

\textbf{Step 2.} For notational simplicity, suppose $q \in g_q$. A sequence of arguments analogous to Step~1 remains to hold conditional on the event $E_q$ for any $q \neq 1$. Therefore, it suffices to show that the equality
\begin{align} \label{eq:limit_equivalence_noniid}
  \lim_{\delta \downarrow 0} F_{(s \vert r)}(c;\delta,E_q^\eta) = \lim_{\delta \downarrow 0} G_{(s \vert r)}(c;\delta,E_q^{\eta'}) , \qquad \text{for all } c \in \mathbb{R} ,
\end{align}
implies that $F_q = G_q$, where
\[
  F_{(s \vert r)}(c;\delta,E_q^\eta) = \mathbb{P}(\eta_{(s)} \le c \,\vert\, \eta_{(r)} \le \delta, E_q^\eta) ,
\]
and $G_{(s \vert r)}(\cdot;\cdot,E_q^{\eta'})$ is defined similarly. Following the same proof as in Lemma~\ref{lemma:appendix_conditional_order_statistic} using the expression in Lemma~\ref{lemma:appendix_conditional_cdf_q}, one can show that
\begin{align} \label{eq:limit_conditional_noniid}
  F_{(s \vert r)}(c;0,E_q^\eta)
    &:= \lim_{\delta \downarrow 0} F_{(s \vert r)}(c;\delta,E_q^\eta) \nonumber \\
    &= \sum_{k=s}^n C^{n-s}_{k-s} \frac{\mathbb{E}_{\zeta_s}\I\{\zeta_s \le c\} F_1(\zeta_s)^{s-r-1} (F_1(c)-F_1(\zeta_s))^{k-s} (1-F_1(c))^{n-k}}{\mathbb{E}_{\zeta_s} F_1(\zeta_s)^{s-r-1} (1-F_1(\zeta_s))^{n-s}} ,
\end{align}
where $\zeta_s \sim F_q$, the group-$q$ distribution. Differentiating \eqref{eq:limit_conditional_noniid} with respect to $c$, the density function is given by
\begin{align*}
  f_{(s \vert r)}(c;0,E_q^\eta)
    &= \frac{n-s}{\mathbb{E}_{\zeta_s} F_1(\zeta_s)^{s-r-1} (1-F_1(\zeta_s))^{n-s}} F_1(c)^{s-r-1}(1-F_1(c))^{n-s} f_q(c) \\
    &\propto F_1(c)^{s-r-1}(1-F_1(c))^{n-s} f_q(c) ,
\end{align*}
for almost all $c \in \mathbb{R}$. Similarly, the density $g_{(s \vert r)}(\cdot;0,E_q^{\eta'})$ is given by
\begin{align*}
  g_{(s \vert r)}(c;0,E_q^\eta)
    &= \frac{n-s}{\mathbb{E}_{\zeta'_s} F_1(\zeta'_s)^{s-r-1} (1-F_1(\zeta'_s))^{n-s}} F_1(c)^{s-r-1}(1-F_1(c))^{n-s} g_q(c) \\
    &\propto F_1(c)^{s-r-1}(1-F_1(c))^{n-s} g_q(c) ,
\end{align*}
for almost all $c \in \mathbb{R}$, where $\zeta'_s \sim G_q$. Since \eqref{eq:limit_equivalence_noniid} implies $f_{(s \vert r)}(\cdot;0,E_q^\eta) = g_{(s \vert r)}(\cdot;0,E_q^\eta)$ a.e.~and $F_1(c)^{s-r-1}(1-F_1(c))^{n-s} > 0$ for all $c$ in the interior of the support of $F_q$ (common support assumption in Assumption~\ref{assumption:group_structure}), we conclude from the above expression of the densities that $f_q = g_q$ a.e., and thus $F_q = G_q$ for all $q \neq 1$. Therefore, the distributions for all measurement error groups are identified. \par

\textbf{Step 3.} Since $F_{\eps_1},\ldots,F_{\eps_n}$ are identified, so is the ch.f.~$\psi_{\eps_{(j)}}$. It follows that the ch.f.~$\psi_\xi = \psi_{X_{(j)}}/\psi_{\eps_{(j)}}$ is identified, where $\psi_\xi(t)$ is defined by the continuous extension whenever $\psi_{\eps_{(j)}}(t) = 0$. It is well-defined since $\psi_{\eps_{(j)}}$ is analytic and hence has isolated zeros. Thus, $F_\xi$ is identified.
\end{proof}


\begin{proof}[Proof of Corollary~\ref{corollary:identification_noniid}]
When $s=n$ and $r=n-1$, the proof of Theorem~\ref{theorem:identification_noniid_group} reveals that
\begin{align*}
  f_{(s \vert r)}(c;0,E_q^\eta) \propto f_q(c) , \qquad g_{(s \vert r)}(c;0,E_q^\eta) \propto g_q(c) .
\end{align*}
Above implies $f_q = g_q$ a.e.~in the absence of a common support assumption.
\end{proof}


\begin{proof}[Proof of Theorem \ref{theorem:estimator_consistency}]
It follows from Theorem~3 in \cite{bierens2008semi} that the sieve $\{\mathcal{H}_k\}_k$ is dense in $\mathcal{H} = \{h \in \pi^2 : \pi \in \mathcal{P}\}$. This implies that $\{\mathcal{F}_k\}_k := \{\mathcal{F}^\xi_k \times \mathcal{F}^\eps_k\}_k$ is dense in $\mathcal{F} := \mathcal{F}^\xi \times \mathcal{F}^\eps$ where
\begin{align*}
  &\mathcal{F}^\xi = \{F(\cdot) = \int_0^{G_\xi(\cdot)} h(u) \, du \, : \, h \in \mathcal{H} \} ,
  &&\mathcal{F}^\eps = \{F(\cdot) = \int_0^{G_\eps(\cdot)} h(u) \, du \, : \, h \in \mathcal{H} \} , \\
  &\mathcal{F}^\xi_k = \{F(\cdot) = \int_0^{G_\xi(\cdot)} h(u) \, du \, : \, h \in \mathcal{H}_k \} ,
  &&\mathcal{F}^\eps_k = \{F(\cdot) = \int_0^{G_\eps(\cdot)} h(u) \, du \, : \, h \in \mathcal{H}_k \} .
\end{align*}
We endow $\mathcal{F}$ with the supremum norm $\lVert\,\cdot\,\rVert_\infty$. Further, because $\mathcal{F}$ is compact, it suffices to show that (1) $Q(\cdot)$ is continuous on $\mathcal{F}$, (2) $Q(\cdot)$ has a unique minimizer on $\mathcal{F}$, and (3) the following uniform a.s.~convergence holds:
\[
  \sup_{F \in \mathcal{F}} \left\lvert \widehat{Q}_N(F) - Q(F) \right\rvert \stackrel{\text{a.s.}}{\longrightarrow} 0 .
\]
See \cite{gallant1987identification} or \cite{gallant1987semi}. A proof of each of (1)--(3) follows. \par

\textbf{(1)} Given a complex number $z \in \mathbb{C}$, let $\Re(z)$ and $\Im(z)$ denote the real and imaginary part, respectively. Let $F_k = (F_{\xi,k},F_{\eps,k}) \in \mathcal{F}$ such that $F_k \rightarrow F \in \mathcal{F}$. Consider:
\begin{align*}
  \lvert Q(F) - Q(F_k) \rvert
  &= \frac{1}{4\kappa^2} \Biggl\lvert \int_{(-\kappa,\kappa)^2} \Re\{\psi_{X_{(r,s)}}(t)-\varphi(t;F)\}^2 + \Im\{\psi_{X_{(r,s)}}(t)-\varphi(t;F)\}^2 \, dt \\
  &\hspace{0.15\textwidth} - \int_{(-\kappa,\kappa)^2} \Re\{\psi_{X_{(r,s)}}(t)-\varphi(t;F_k)\}^2 + \Im\{\psi_{X_{(r,s)}}(t)-\varphi(t;F_k)\}^2 \, dt \Biggr\rvert \\
  &= \frac{1}{4\kappa^2} \Biggl\lvert \int_{(-\kappa,\kappa)^2} \Re\{2\psi_{X_{(r,s)}}(t)-\varphi(t;F)-\varphi(t;F_k)\}\Re\{\varphi(t;F)-\varphi(t;F_k)\} \\
  &\hspace{0.15\textwidth} + \Im\{2\psi_{X_{(r,s)}}(t)-\varphi(t;F)-\varphi(t;F_k)\}\Im\{\varphi(t;F)-\varphi(t;F_k)\} \, dt \Biggr\rvert \\
  &\le \frac{1}{4\kappa^2} \int_{(-\kappa,\kappa)^2} \Bigl\lvert \Re\{2\psi_{X_{(r,s)}}(t)-\varphi(t;F)-\varphi(t;F_k)\} \Bigr\rvert \Bigl\lvert \Re\{\varphi(t;F)-\varphi(t;F_k)\} \Bigr\rvert \\
  &\hspace{0.15\textwidth}+ \Bigl\lvert \Im\{2\psi_{X_{(r,s)}}(t)-\varphi(t;F)-\varphi(t;F_k)\} \Bigl\lvert \Bigr\rvert \Im\{\varphi(t;F)-\varphi(t;F_k)\} \Bigr\rvert \, dt \\
  &\le \frac{1}{\kappa^2} \int_{(-\kappa,\kappa)^2} \Bigl\lvert \Re\{\varphi(t;F)-\varphi(t;F_k)\} \Bigr\rvert + \Bigl\lvert \Im\{\varphi(t;F)-\varphi(t;F_k)\} \Bigr\rvert \, dt \\
  &\le \frac{\sqrt{2}}{\kappa^2} \int_{(-\kappa,\kappa)^2} \Bigl\lvert \varphi(t;F)-\varphi(t;F_k) \Bigr\rvert \, dt \\
  &\le 4\sqrt{2} \sup_{t \in (-\kappa,\kappa)^2} \Bigl\lvert \varphi(t;F)-\varphi(t;F_k) \Bigr\rvert .
\end{align*}
Therefore, to show $\lvert Q(F) - Q(F_k) \rvert \rightarrow 0$, it suffices to show that
\[
  (\xi(F_{\xi,k}),\eps_{1}(F_{\eps,k}),\ldots,\eps_{n}(F_{\eps,k})) \stackrel{\text{a.s.}}{\longrightarrow} (\xi(F_{\xi}),\eps_{1}(F_{\eps}),\ldots,\eps_{n}(F_{\eps})) ,
\]
as this implies a.s.~convergence of $(\xi(F_{\xi,k}),\eps_{(r)}(F_{\eps,k}),\eps_{(s)}(F_{\eps,k}))$ to $(\xi(F_{\xi}),\eps_{(r)}(F_{\eps}),\eps_{(s)}(F_{\eps}))$ and thus uniform convergence of $\varphi(t;F_k)$ to $\varphi(t;F)$ on any compact set. For each $k$, let $H_k$ denote the c.d.f.~such that $F_{\xi,k}(\cdot) = H_k(G_\xi(\cdot))$. $F_{\xi,k} \rightarrow F_\xi$ implies $H_k \rightarrow H$ where $F_\xi(\cdot) = H(G_\xi(\cdot))$. Therefore, since $G_\xi^{-1}$ and $H^{-1}$ are monotone and hence have at most countable discontinuity points, we conclude that $V$--a.s.,
\[
  \lim_{k \rightarrow \infty} \xi(F_{\xi,k})
    = \lim_{k \rightarrow \infty} F_{\xi,k}^{-1}(V)
    = \lim_{k \rightarrow \infty} G_\xi^{-1}(H_k^{-1}(V))
    = G_\xi^{-1}(H^{-1}(V))
    = F_\xi^{-1}(V)
    = \xi(F_\xi) .
\]
Similarly, we have $\lim_{k \rightarrow \infty} \eps_j(F_{\eps,k}) = \eps_j(F_\eps)$, $U_j$--a.s. We thus conclude that $Q(\cdot)$ is continuous on $\mathcal{F}$. \par

\textbf{(2)} Existence of $F=(F_\xi,F_\eps) \in \mathcal{F}$ satisfying $Q(F) = 0$ is obvious from Assumption~\ref{assumption:noniid_case}(b). Uniqueness of the minimizer follows from Theorem~\ref{theorem:identification_iid_case}. \par

\textbf{(3)} Note that we have
\begin{align*}
  &\sup_{F \in \mathcal{F}} \left\rvert \widehat{Q}_N(F) - Q(F) \right\lvert \\
  &\qquad= \sup_{F \in \mathcal{F}} \frac{1}{4\kappa^2} \Biggl\rvert \int_{(-\kappa,\kappa)^2} \Re\{\widehat\psi_N(t)-\widehat\varphi_N(t;F)\}^2 + \Im\{\widehat\psi_N(t)-\widehat\varphi_N(t;F)\}^2 \, dt \\
  &\hspace{0.2\textwidth} - \int_{(-\kappa,\kappa)^2} \Re\{\psi_{X_{(r,s)}}(t)-\varphi(t;F)\}^2 + \Im\{\psi_{X_{(r,s)}}(t)-\varphi(t;F)\}^2 \, dt \Biggr\lvert \\
  &\qquad= \sup_{F \in \mathcal{F}} \frac{1}{4\kappa^2} \Biggl\rvert \int_{(-\kappa,\kappa)^2} \Re\{\widehat\psi_N(t)-\widehat\varphi_N(t;F)+\psi_{X_{(r,s)}}(t)-\varphi(t;F)\} \\
  &\hspace{0.4\textwidth} \Re\{(\widehat\psi_N(t)-\psi_{X_{(r,s)}}(t))-(\widehat\varphi_N(t;F)-\varphi(t;F))\} \\
  &\hspace{0.25\textwidth} + \Im\{\widehat\psi_N(t)-\widehat\varphi_N(t;F)+\psi_{X_{(r,s)}}(t)-\varphi(t;F)\} \\
  &\hspace{0.4\textwidth} \Im\{(\widehat\psi_N(t)-\psi_{X_{(r,s)}}(t))-(\widehat\varphi_N(t;F)-\varphi(t;F))\} \, dt \Biggr\lvert \\
  &\qquad\le \sup_{F \in \mathcal{F}} \frac{1}{\kappa^2} \int_{(-\kappa,\kappa)^2} \left\lvert \Re\{(\widehat\psi_N(t)-\psi_{X_{(r,s)}}(t))-\Re(\widehat\varphi_N(t;F)-\varphi(t;F))\} \right\rvert \\
  &\hspace{0.3\textwidth} +  \left\lvert \Im\{(\widehat\psi_N(t)-\psi_{X_{(r,s)}}(t))-\Im(\widehat\varphi_N(t;F)-\varphi(t;F))\} \right\rvert \, dt \\
  &\qquad\le \sup_{F \in \mathcal{F}} \frac{2}{\kappa^2} \int_{(-\kappa,\kappa)^2} \Bigl\lvert \widehat\psi_N(t)-\psi_{X_{(r,s)}}(t) \Bigr\rvert + \Bigl\lvert \widehat\varphi_N(t;F)-\varphi(t;F) \Bigr\rvert \, dt \\
  &\qquad\le \frac{2}{\kappa^2} \int_{(-\kappa,\kappa)^2} \Bigl\lvert \widehat\psi_N(t)-\psi_{X_{(r,s)}}(t) \Bigr\rvert \, dt + \frac{2}{\kappa^2} \int_{(-\kappa,\kappa)^2} \sup_{F \in \mathcal{F}} \Bigl\lvert \widehat\varphi_N(t;F)-\varphi(t;F) \Bigr\rvert \, dt .
\end{align*}
Recall ch.f.s are bounded uniformly by $1$. Since $\widehat\psi_N \rightarrow_{\text{a.s.}} \psi_{X_{(r,s)}}$ pointwise on $\mathbb{R}^2$, the first integral on the right-hand side a.s.~vanishes asymptotically. It remains to be shown that $\widehat\varphi_N(t;\cdot) \rightarrow_{\text{a.s.}} \varphi(t;\cdot)$ uniformly in $F$. the supremum in the last integral vanishes pointwise on $(-\kappa,\kappa)^2$. Since, by definition
\begin{align*}
  \Bigl\lvert \widehat\varphi_N(t;F)-\varphi(t;F) \Bigr\rvert
    &= \Bigl\lvert \mathbb{E}_N e^{it^\top X(F)} - \mathbb{E} e^{it^\top X(F)} \Bigr\rvert \\
    &\le \Bigl\lvert \mathbb{E}_N \cos(t^\top X(F)) - \mathbb{E} \cos(t^\top X(F)) \Bigr\rvert + \Bigl\lvert \mathbb{E}_N \sin(t^\top X(F)) - \mathbb{E} \sin(t^\top X(F)) \Bigr\rvert ,
\end{align*}
where both terms are bounded, we conclude from the Borel--Cantelli lemma and Hoeffding's inequality that
\begin{align} \label{eq:varphi_convergence}
  \Bigl\lvert \widehat\varphi_N(t;F)-\varphi(t;F) \Bigr\rvert
    \stackrel{\text{a.s.}}{\longrightarrow} 0 , \quad \text{for all } F \in \mathcal{F} .
\end{align}
We complete the proof by establishing its strong stochastic equicontinuity. Note that a.s.~continuity of $(\xi(F),\eps_1(F),\ldots,\eps_n(F))$ in part (1) of the proof implies a.s.~continuity of $X(F)$. Thus, for any $\eta > 0$, there exists $\delta > 0$ such that
\begin{align*}
  d(F,F') = \max \bigl\{ \lVert F_\xi-F_\xi'\rVert_\infty,\lVert F_\eps-F_\eps'\rVert_\infty \bigr\} < \delta
\end{align*}
implies, a.s.,
\begin{align*}
  \bigl\lvert \cos(t^\top X(F)) - \cos(t^\top X(F')) \bigr\rvert < \eta/4 , \quad
  \bigl\lvert \sin(t^\top X(F)) - \sin(t^\top X(F')) \bigr\rvert < \eta/4 .
\end{align*}
Therefore,
\begin{align*}
  &\sup_{d(F,F') < \delta} \Bigl\lvert \Bigl(\widehat\varphi_N(t;F) - \varphi(t;F)\Bigr) - \Bigl(\widehat\varphi_N(t;F') - \varphi(t;F')\Bigr) \Bigr\rvert \\
  &\qquad= \sup_{d(F,F') < \delta} \Bigl\lvert \Bigl(\mathbb{E}_N e^{it^\top X(F)} - \mathbb{E} e^{it^\top X(F)}\Bigr) - \Bigl(\mathbb{E}_N e^{it^\top X(F')} - \mathbb{E} e^{it^\top X(F')}\Bigr) \Bigr\rvert \\
  &\qquad\le \sup_{d(F,F') < \delta} \Bigl\lvert \Bigl(\mathbb{E}_N e^{it^\top X(F)} - \mathbb{E}_N e^{it^\top X(F')}\Bigr) - \Bigl(\mathbb{E} e^{it^\top X(F)} - \mathbb{E} e^{it^\top X(F')}\Bigr) \Bigr\rvert \\
  &\qquad\le \mathbb{E}_N \sup_{d(F,F') < \delta} \Bigl( \Bigl\lvert \cos(t^\top X(F)) - \cos(t^\top X(F')) \Bigr\rvert + \Bigl\lvert \sin(t^\top X(F)) - \sin(t^\top X(F')) \Bigr\rvert \Bigr) \\
  &\hspace{0.1\textwidth} + \mathbb{E} \sup_{d(F,F') < \delta} \Bigl( \Bigl\lvert \cos(t^\top X(F)) - \cos(t^\top X(F')) \Bigr\rvert + \Bigl\lvert \sin(t^\top X(F)) - \sin(t^\top X(F')) \Bigr\rvert \Bigr) \\
  &\qquad< \eta .
\end{align*}
Since the inequality holds for all $N$, we conclude by strong stochastic equicontinuity that the a.s.~convergence in \eqref{eq:varphi_convergence} holds uniformly on $\mathcal{F}$.
\end{proof}


\section{Rossberg's Counterexample} \label{appendix:rossberg}

One natural approach towards investigating the identification problem of the measurement error distribution is to exploit the spacing between two order statistics: $X_{(s)} - X_{(r)} = \eps_{(s)} - \eps_{(r)}$. However, without additional restrictions, knowledge of the spacing distribution is not sufficient to identify $F_\eps$. A constructive example is provided by \cite{rossberg1972characterization}. Specifically, suppose $\eps_1$ and $\eps_2$ are i.i.d.~standard exponential random variables. Then the spacing $\eps_{(2)} - \eps_{(1)}$ is a folded standard Laplace random variable. \cite{rossberg1972characterization} shows there are infinitely many parent distributions that deliver a folded standard Laplace spacing distribution, one of them being
\begin{align*}
  G(x) = 1 - e^{-x} \big[ 1 + \pi^{-2} (1 - \cos 2 \pi x) \big] ,
\end{align*}
with support $S(G) = [0,\infty)$. \par

However, our Corollary \ref{corollary:observational_equivalence_eps} requires that, for Rossberg's distribution $G$ to be observationally equivalent to the true exponential distribution, the distributions of cross-sums of order statistics must be the same in addition to their spacing distributions. This additional condition arises from the model setup, in particular, from the independence assumption between the latent variable of interest, $\xi$, and the measurement errors. Whereas the empirical content comes from the random vector $(X_{(1)},X_{(2)})$ in our context, the spacing is the only---and complete---empirical content available in the setting of \cite{rossberg1972characterization}. As discussed in Section \ref{subsection:identification}, the independence assumption allows us to exploit the ratio of ch.f.s in a tractable manner despite the dependence between measurement errors of the observed. \par

\begin{figure}[!t]
  \centering
  \begin{subfigure}[!t]{0.49\textwidth}
    \centering
    \includegraphics[width=\textwidth]{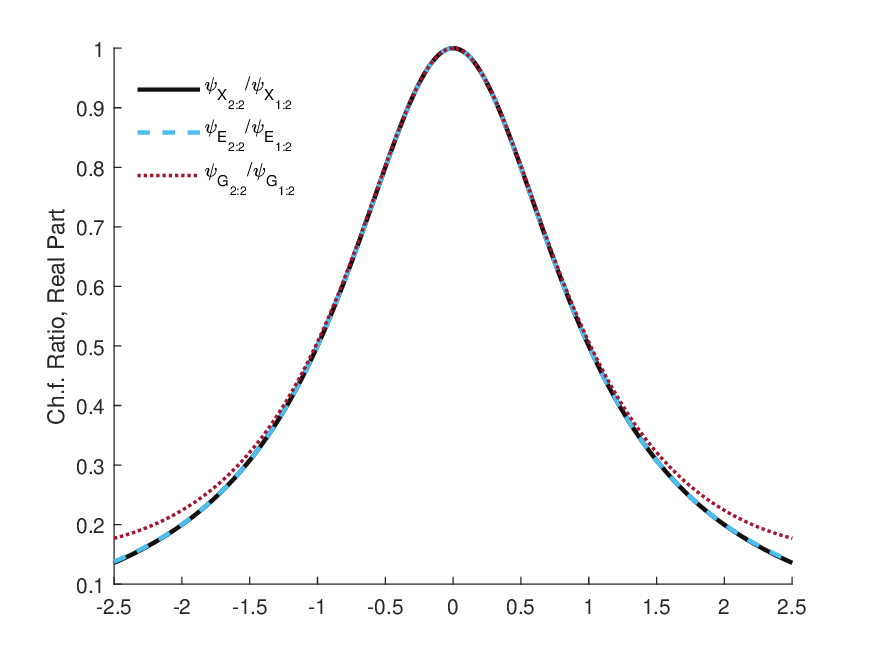}
    \caption{Real part of the ratios of ch.f.s.} \label{fig:chf_real_fig}
  \end{subfigure}%
  \hfill%
  \begin{subfigure}[!t]{0.49\textwidth}
    \centering
    \includegraphics[width=\textwidth]{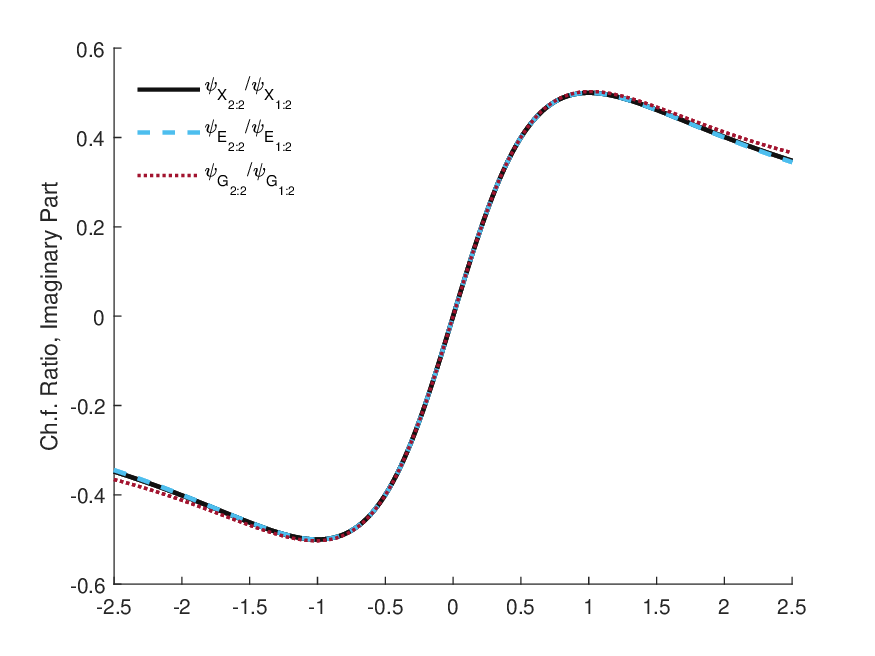}
    \caption{Imaginary part of the ratios of ch.f.s.} \label{fig:chf_imag_fig}
  \end{subfigure}
  \caption{An illustration of the departure of the ratio of ch.f.s of second and first Rossberg order statistics (\emph{red dotted line}) from that of the measurement (observed) order statistics (\emph{black line}). The ratio of ch.f.s of true measurement error (exponential) order statistics (\emph{blue dashed line}) aligns with that of the measurement order statistics.} \label{fig:chf_fig}
\end{figure}

\begin{figure}[!t]
  \centering
  \begin{subfigure}[!t]{0.5\textwidth}
    \centering
    \includegraphics[width=\textwidth]{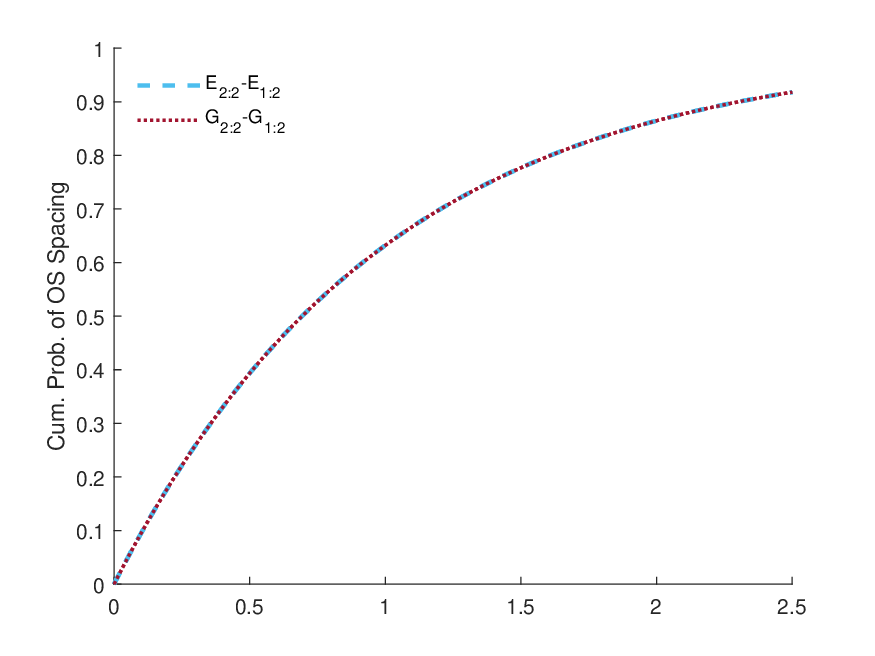}
    \caption{Cumulative distributions of the spacings. (\emph{blue dashed line}: exponential spacing; \emph{red dotted line}: Rossberg spacing)} \label{fig:df_spacing_fig}
  \end{subfigure}%
  \hfill%
  \begin{subfigure}[!t]{0.5\textwidth}
    \centering
    \includegraphics[width=\textwidth]{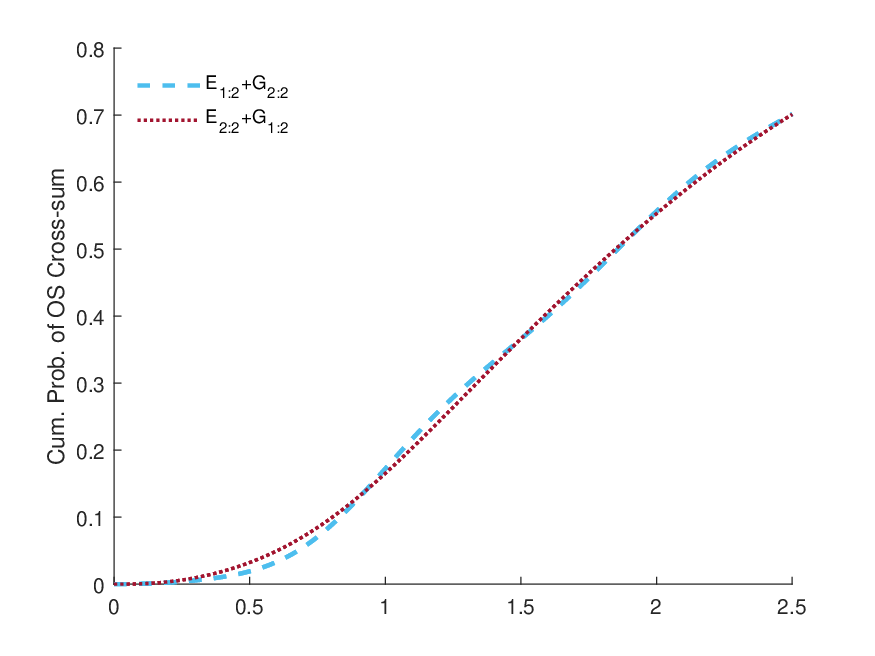}
    \caption{Cumulative distributions of the cross-sums. (\emph{blue dashed line}: exponential first order statistic; \emph{red dotted line}: Rossberg first order statistic)} \label{fig:df_crosssum_fig}
  \end{subfigure}
  \caption{An illustration of the dissimilarity between the probability distributions of cross-sums of two exponential and Rossberg order statistics (Figure \ref{fig:df_crosssum_fig}). Consistent with \cite{rossberg1972characterization}, the probability distributions of exponential and Rossberg spacings are aligned (Figure \ref{fig:df_spacing_fig}).} \label{fig:df_fig}
\end{figure}

Figure \ref{fig:chf_fig} compares the ratios of ch.f.s of order statistics.\footnote{Due to the lack of analytical tractability of the ratio for Rossberg's distribution, all functions are approximated by Monte Carlo integration with $N=10^8$ pairs of random draws. To pin down the specification of $X_j = \xi + \eps_j$, random quantity $\xi$ is drawn from the standard normal distribution. The same draws are used to plot the distribution functions in Figure \ref{fig:df_fig}.} In order to make the minor departures clear, we forfeit the three-dimensional display of the ratios of ch.f.s in \eqref{eq:identify_eps}. The ratios are evaluated at $\psi_{X_{(1,2)}}(0,t) / \psi_{X_{(1)}}(t) = \psi_{X_{(2)}}(t) / \psi_{X_{(1)}}(t)$, and the real and imaginary parts are displayed separately. The ratios for the observations $(X_{(1)},X_{(2)})$ and the true measurement errors $(\eps_{(1)},\eps_{(2)})$ from the standard exponential distribution coincide. However, the ratio for $G$ begins to depart substantially when $|t| > 1$, indicating that $G$ cannot rationalize the observed data. Consequently, in Figure \ref{fig:df_fig}, while the spacing distributions for the exponential and Rossberg's random variables overlap, the cross-sum distributions differ over nontrivial regions. This connection between ratios of ch.f.s and probability distributions of spacings and cross-sums is established in Corollary \ref{corollary:observational_equivalence_eps}. \par


\section{Implementation of the Estimator} \label{appendix:estimation}

This section provides a closed form for the empirical criterion function and describes how the estimation procedure is implemented in practice to calculate the estimator proposed in Section~\ref{section:estimator_main}. Additional simulation results are also reported. While it is left implicit in the main text whether the vectorized form of an element is a column or row vector, it is made clear here. For instance, $x = (x_1,\ldots,x_n)$ is a row vector of length $n$ and $x^\top$ a column vector. \par

\subsection{Closed form for the Criterion Function}

The simulated sieve estimator minimizes the empirical analogue of \eqref{eq:population_criterion_function}. By simulating the model-implied ch.f.~$\varphi$, we avoid having to numerically integrate the criterion function over $(-\kappa,\kappa)^2$. Precisely, the empirical criterion function $\widehat{Q}_N(\cdot)$ has the following closed-form:
\begin{align}
  \widehat{Q}_N(F_{k_N};\kappa)
  = \frac{2}{N} + \frac{2}{N^2} \sum_{i > j}^N q(X_i-X_j) &+ \frac{2}{N^2} \sum_{i > j}^N q(X_i(F_{k_N})-X_j(F_{k_N})) \nonumber \\
  &\qquad- \frac{2}{N^2} \sum_{i, j}^N q(X_i-X_j(F_{k_N})) , \label{eq:empirical_criterion_function}
\end{align}
where $X_i = (X_{(r),i},X_{(s),i})^\top$ is the $i$\textsuperscript{th} observation, $X_i(F_{k_N})$ is the $i$\textsuperscript{th} simulated column vector given c.d.f.s $F_{\xi,k_N}$ and $F_{\eps,k_N}$, to be made precise below (see also Assumption~\ref{assumption:noniid_case}(d)), and $q(x,y) = \sin(\kappa x)\sin(\kappa y) / (\kappa^2xy)$, defined everywhere by the continuous extension. The above closed form is straightforward to derive using the identity $e^{ix} = \cos(x) + i \sin(x)$ and two trigonometric identities, $\cos(x-y) = \cos(x)\cos(y) + \sin(x)\sin(y)$ and $2\sin(x)\sin(y) = \cos(x-y)-\cos(x+y)$. \par

\subsection{Estimation Procedure}

Note that the sieve space $\{\mathcal{F}_k\}_k := \{\mathcal{F}^\xi_k \times \mathcal{F}^\eps_k\}_k$ introduced in Section~\ref{section:estimator_main} is constructed by $\{\mathcal{P}_k\}_k$, a sequence of spaces of truncated series based on the series representation $\mathcal{P}$ in \eqref{eq:legendre_series_pi}. That is, the infinite-dimensional spaces $\{\mathcal{P}_k\}_k$ are determined by a sequence of finite-dimensional sieves $\{\mathcal{D}_k\}_k$, where
\[
  \mathcal{D}_k := \left\{ \delta \in \mathbb{R}^k : \lvert\delta_\ell\rvert \le \frac{c}{1+\sqrt{\ell}\ln\ell} , \, \ell=1,\ldots,k \right\} .
\]
Thus, so as to minimize $\widehat{Q}_N$ in \eqref{eq:empirical_criterion_function} with respect to $\mathcal{F}_{k_N}$, one can minimize
\[
  \widehat{Q}_N(F_{k_N}) = \widehat{Q}_N\bigl(H_k(G_\xi(\cdot);\delta_{\xi}), H_k(G_\eps(\cdot);\delta_{\eps})\bigr)
\]
with respect to $(\delta_{\xi},\delta_{\eps}) \in \mathcal{D}_{k_N}^2$, where
\[
  H_k(v;\delta) = \frac{\int_0^v \bigl( 1+\sum_{\ell=1}^k \delta_\ell \rho_\ell(v) \bigr)^2 \, dv}{1+\sum_{\ell=1}^k \delta_\ell^2} .
\]

To alleviate computational concerns associated with solving the above finite-dimensional optimization problem with the parameterization in \eqref{eq:legendre_series_pi}, \cite{bierens2008semi} suggests a reparametrization under which the c.d.f.~$H_k(\cdot;\delta)$ on $[0,1]$ can be expressed as (see Sections~3 and 7 in \cite{bierens2008semi}):
\begin{align*}
  H_k(v;\theta) = \frac{(1-\pi_k^\top\theta,\theta^\top) \Pi_{k+1}(v) (1-\pi_k^\top\theta,\theta^\top)^\top}{(1-\pi_k^\top\theta,\theta^\top) \Pi_{k+1}(1) (1-\pi_k^\top\theta,\theta^\top)^\top} , \quad v \in [0,1] ,
\end{align*}
where $\Pi_{k+1}(v)$ is a $(k+1)$-square matrix and $\pi_k$ a $k$-dimensional vector defined as
\begin{align*}
  \Pi_{k+1}(v) = \left( \frac{v^{i+j+1}}{i+j+1} \right)_{i,j = 0,1,\ldots,k} \qquad \text{and} \qquad
  \pi_k = \left( \frac{1}{i+1} \right)_{i = 1,\ldots,k} .
\end{align*}
The compactifying restrictions on $\delta \in \mathcal{D}_k$ in the definition of $\mathcal{D}_k$ translates to the following constraints on the space $\Theta_k$ for $\theta$:
\[
  \Theta_k = \left\{\theta \in \mathbb{R}^k : \left\lvert \sum_{m=0}^{k-\ell} \theta_{\ell+m} \mu_\ell(m) \right\rvert \le \frac{c}{1+\sqrt{\ell}\ln\ell} , \, \ell=1,\ldots,k \right\} ,
\]
where $\mu_\ell(m) := \int_0^1 u^{\ell+m}\rho_\ell(u) \, du$.

We use the finite-dimensional sieves $\{\Theta_k^2\}_k$ for $(\theta_\xi,\theta_\eps)$ to estimate the distribution functions in the Monte Carlo study in Section~\ref{subsection:montecarlo}. For further details regarding the sieve space, we refer the readers to the original article by \cite{bierens2008semi}. \par

\subsection{Additional Simulation Results}

In this section, we report additional results from the simulation exercise described in Section~\ref{subsection:montecarlo}. Figures~\ref{fig:montecarlo1} and \ref{fig:montecarlo2} display pointwise box plots that summarizes the simulated coverage of the estimator for $F_\xi$ and $F_\eps$. Sample sizes $N = 1000$, $2000$, and $4000$ with $k=4$, $5$, and $6$, respectively, are considered with $\kappa = 1$, $3.14$, and $5$ for each sample size.\footnote{Panel 1 of Figures~\ref{fig:montecarlo1} and \ref{fig:montecarlo2} are reproduced in Figure~\ref{fig:montecarlo_maintext}.} For each $\kappa$, the result shows that the estimator improves with larger $N$. The pointwise variance, as indicated by more values outside the interquartile range, appears to be higher when $\kappa$ is set to a larger value.

\begin{figure}[h!]
  \centering
  \begin{subfigure}[!t]{0.33\textwidth}
    \centering
    \includegraphics[width=\textwidth]{figures/cdfxi_N=1000_k=4_kappa=100_R=500.pdf}
    \caption*{\emph{Panel 1(a)}}
  \end{subfigure}%
  \hfill%
  \begin{subfigure}[!t]{0.33\textwidth}
    \centering
    \includegraphics[width=\textwidth]{figures/cdfxi_N=2000_k=5_kappa=100_R=500.pdf}
    \caption*{\emph{Panel 1(b)}}
  \end{subfigure}%
  \hfill%
  \begin{subfigure}[!t]{0.33\textwidth}
    \centering
    \includegraphics[width=\textwidth]{figures/cdfxi_N=4000_k=6_kappa=100_R=500.pdf}
    \caption*{\emph{Panel 1(c)}}
  \end{subfigure}
  \begin{subfigure}[!t]{0.33\textwidth}
    \centering
    \includegraphics[width=\textwidth]{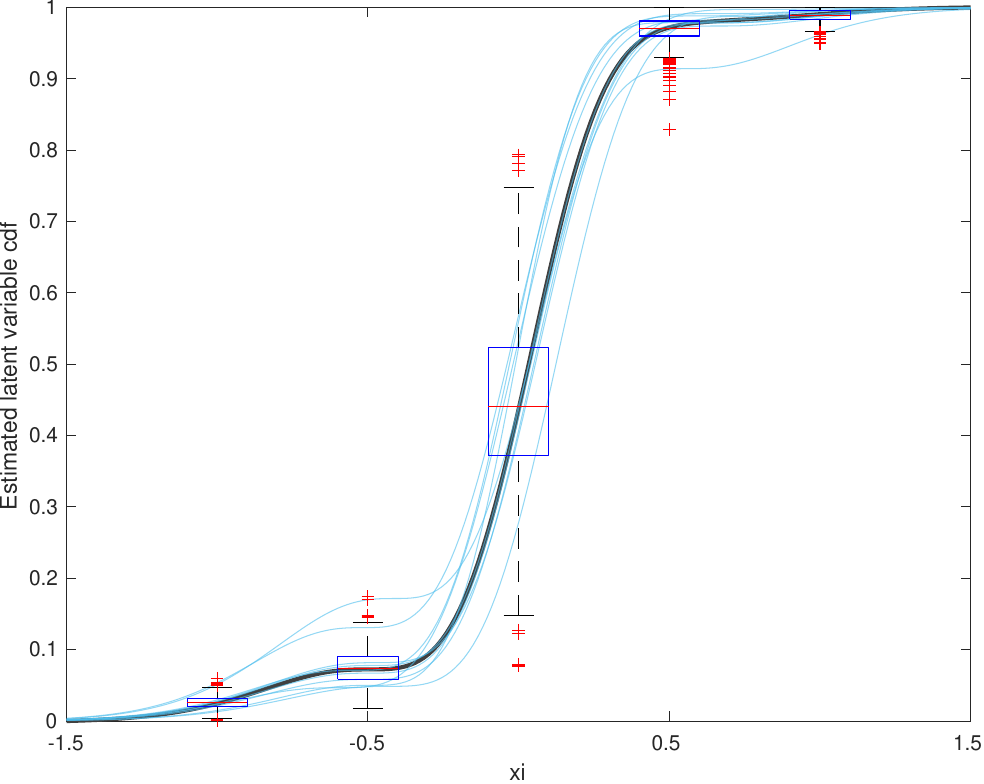}
    \caption*{\emph{Panel 2(a)}}
  \end{subfigure}%
  \hfill%
  \begin{subfigure}[!t]{0.33\textwidth}
    \centering
    \includegraphics[width=\textwidth]{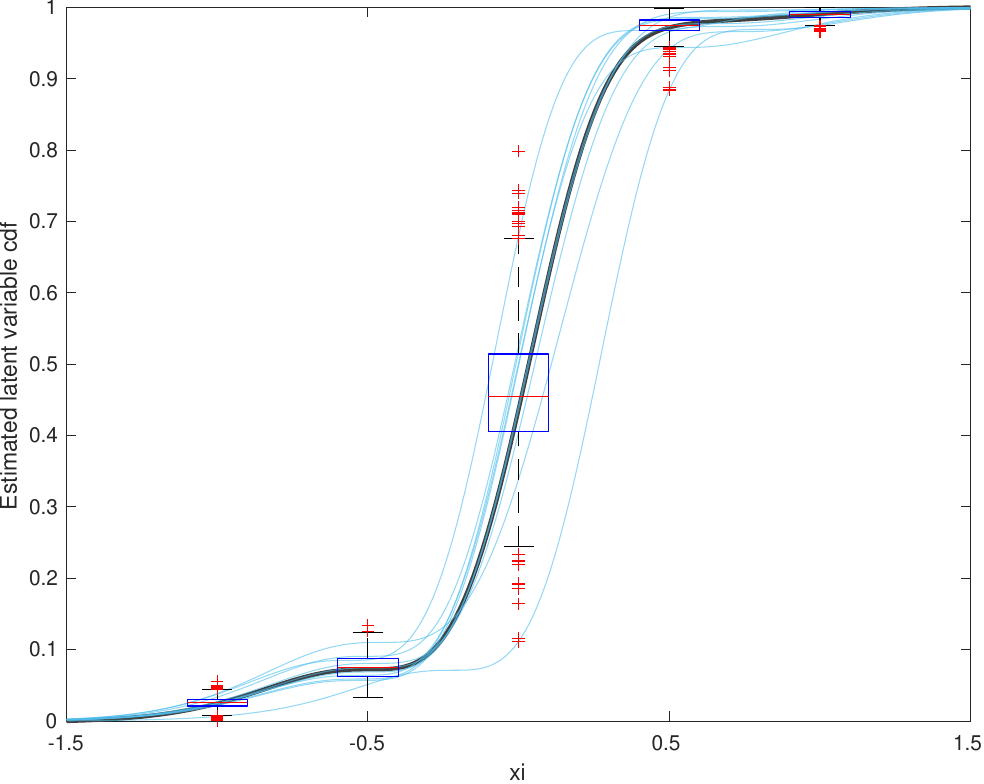}
    \caption*{\emph{Panel 2(b)}}
  \end{subfigure}%
  \hfill%
  \begin{subfigure}[!t]{0.33\textwidth}
    \centering
    \includegraphics[width=\textwidth]{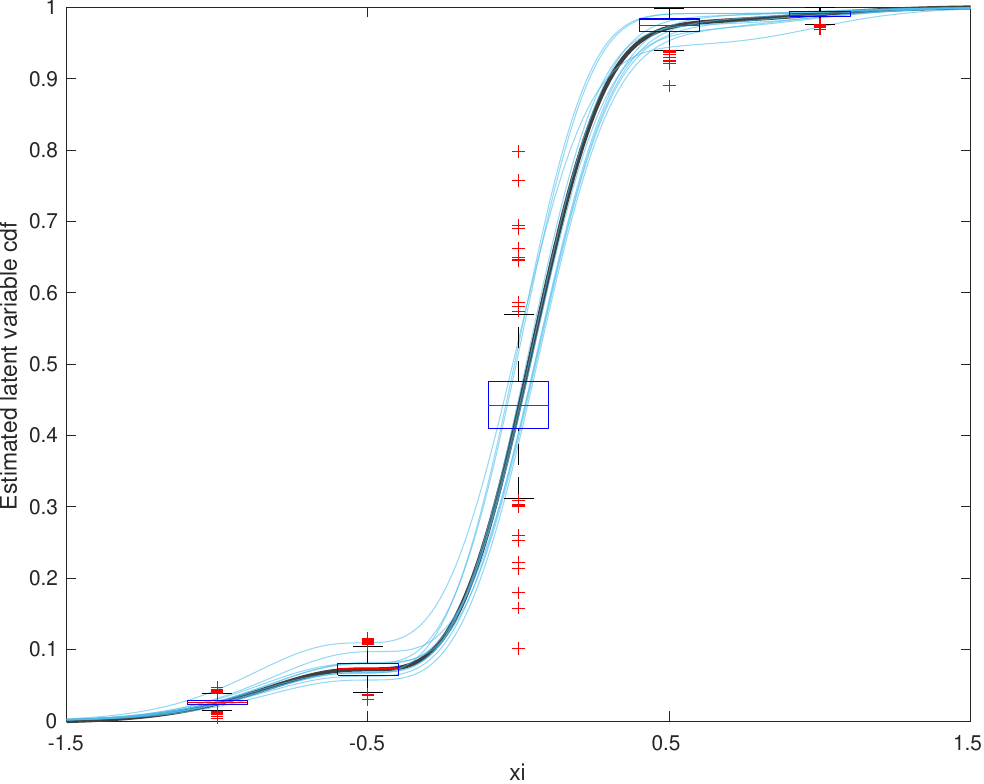}
    \caption*{\emph{Panel 2(c)}}
  \end{subfigure}
    \begin{subfigure}[!t]{0.33\textwidth}
    \centering
    \includegraphics[width=\textwidth]{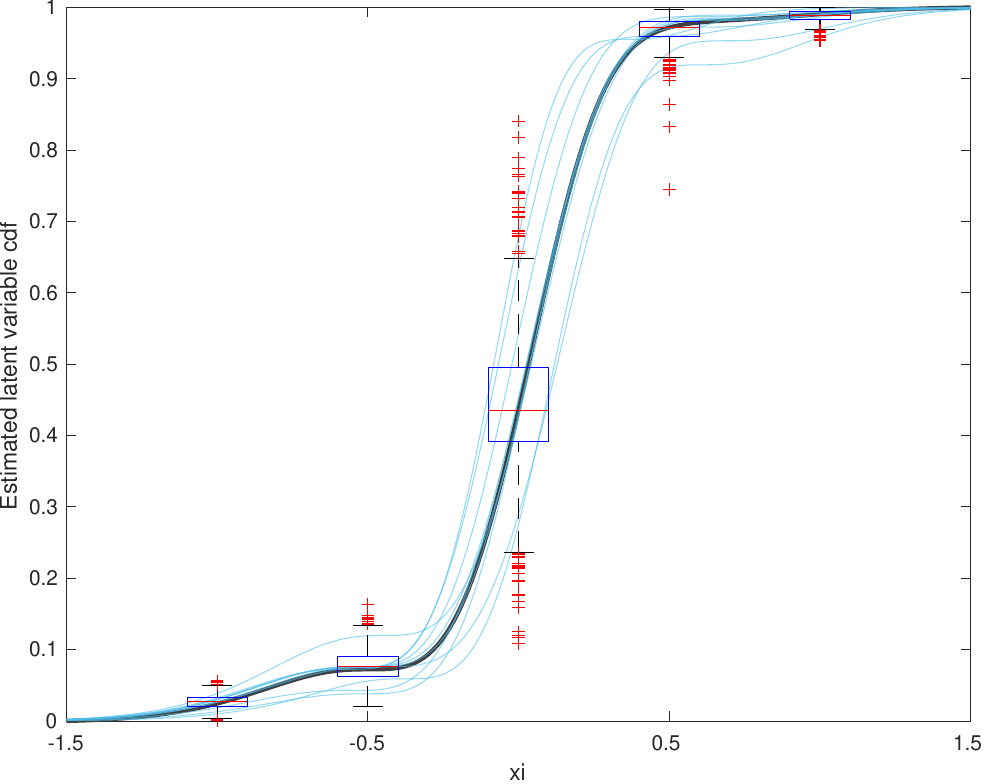}
    \caption*{\emph{Panel 3(a)}}
  \end{subfigure}%
  \hfill%
  \begin{subfigure}[!t]{0.33\textwidth}
    \centering
    \includegraphics[width=\textwidth]{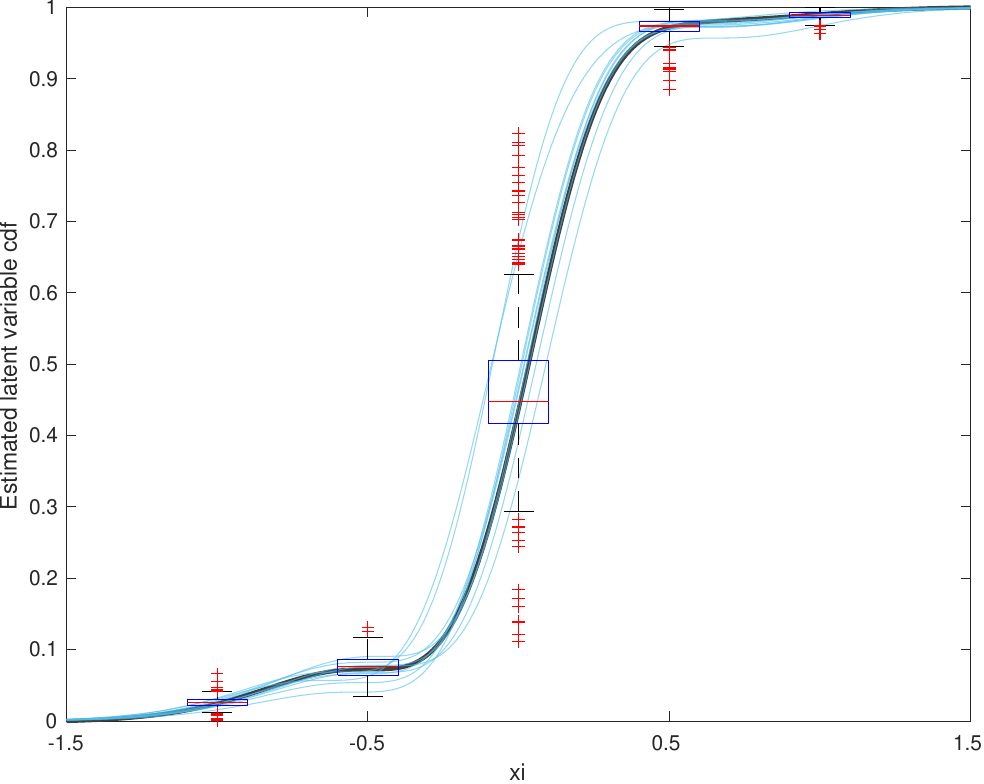}
    \caption*{\emph{Panel 3(b)}}
  \end{subfigure}%
  \hfill%
  \begin{subfigure}[!t]{0.33\textwidth}
    \centering
    \includegraphics[width=\textwidth]{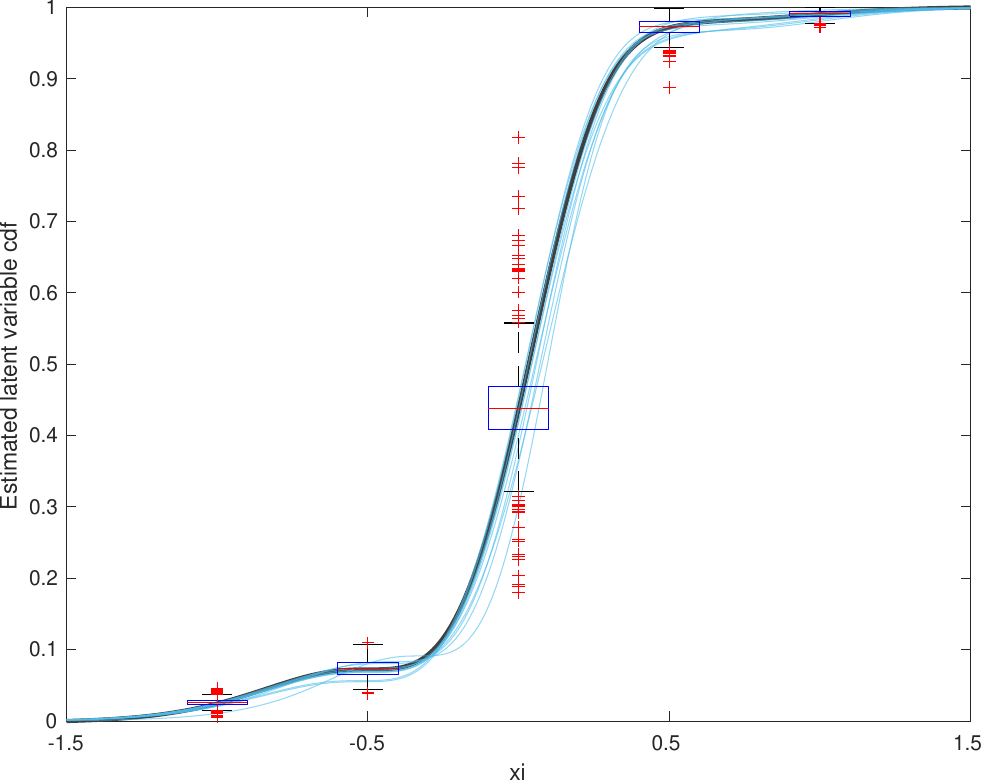}
    \caption*{\emph{Panel 3(c)}}
  \end{subfigure}
  \caption{Monte Carlo simulation results for $F_\xi$. Panels 1, 2, and 3 correspond to the tuning parameter $\kappa=1$, $3.14$, and $5$, respectively; Subpanels \emph{(a)}, \emph{(b)}, and \emph{(c)} correspond to sample sizes $N=1000$, $2000$, and $4000$, respectively.} \label{fig:montecarlo1}
\end{figure}

\begin{figure}[h!]
  \centering
  \begin{subfigure}[!t]{0.33\textwidth}
    \centering
    \includegraphics[width=\textwidth]{figures/cdfepsilon_N=1000_k=4_kappa=100_R=500.pdf}
    \caption*{\emph{Panel 1(a)}}
  \end{subfigure}%
  \hfill%
  \begin{subfigure}[!t]{0.33\textwidth}
    \centering
    \includegraphics[width=\textwidth]{figures/cdfepsilon_N=2000_k=5_kappa=100_R=500.pdf}
    \caption*{\emph{Panel 1(b)}}
  \end{subfigure}%
  \hfill%
  \begin{subfigure}[!t]{0.33\textwidth}
    \centering
    \includegraphics[width=\textwidth]{figures/cdfepsilon_N=4000_k=6_kappa=100_R=500.pdf}
    \caption*{\emph{Panel 1(c)}}
  \end{subfigure}
  \begin{subfigure}[!t]{0.33\textwidth}
    \centering
    \includegraphics[width=\textwidth]{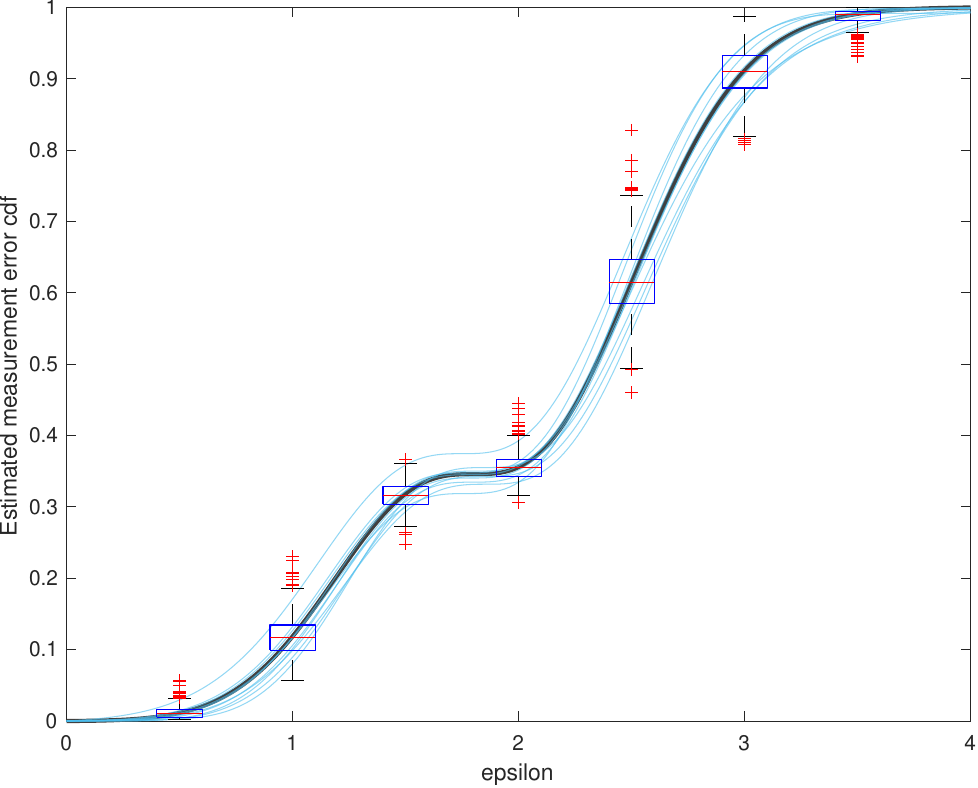}
    \caption*{\emph{Panel 2(a)}}
  \end{subfigure}%
  \hfill%
  \begin{subfigure}[!t]{0.33\textwidth}
    \centering
    \includegraphics[width=\textwidth]{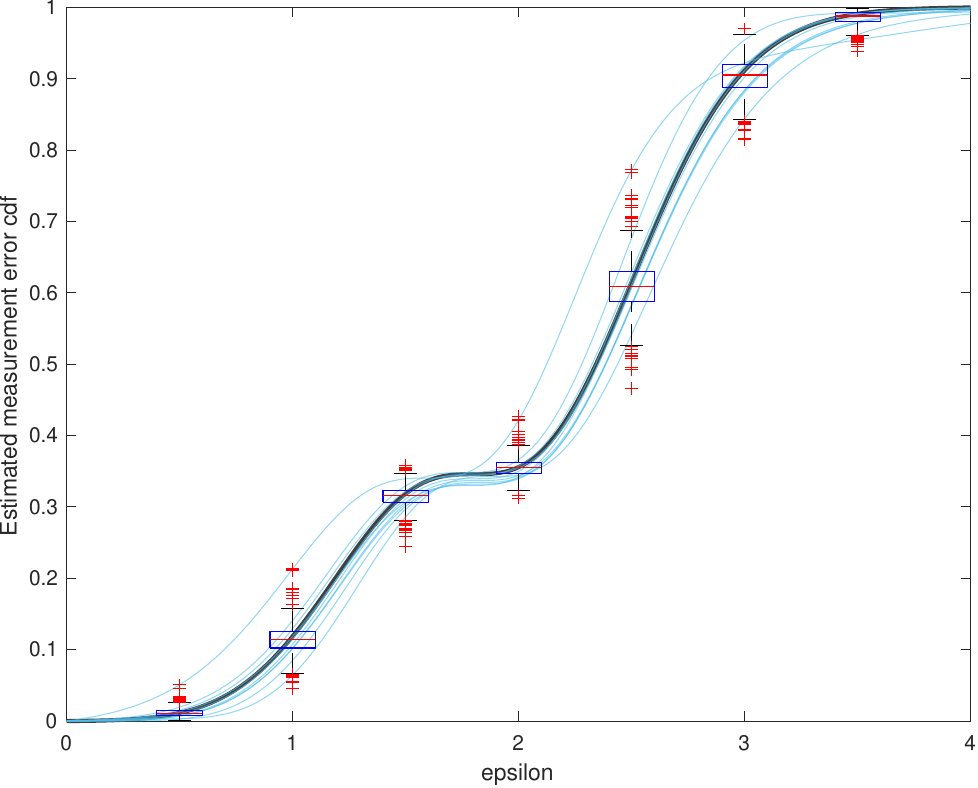}
    \caption*{\emph{Panel 2(b)}}
  \end{subfigure}%
  \hfill%
  \begin{subfigure}[!t]{0.33\textwidth}
    \centering
    \includegraphics[width=\textwidth]{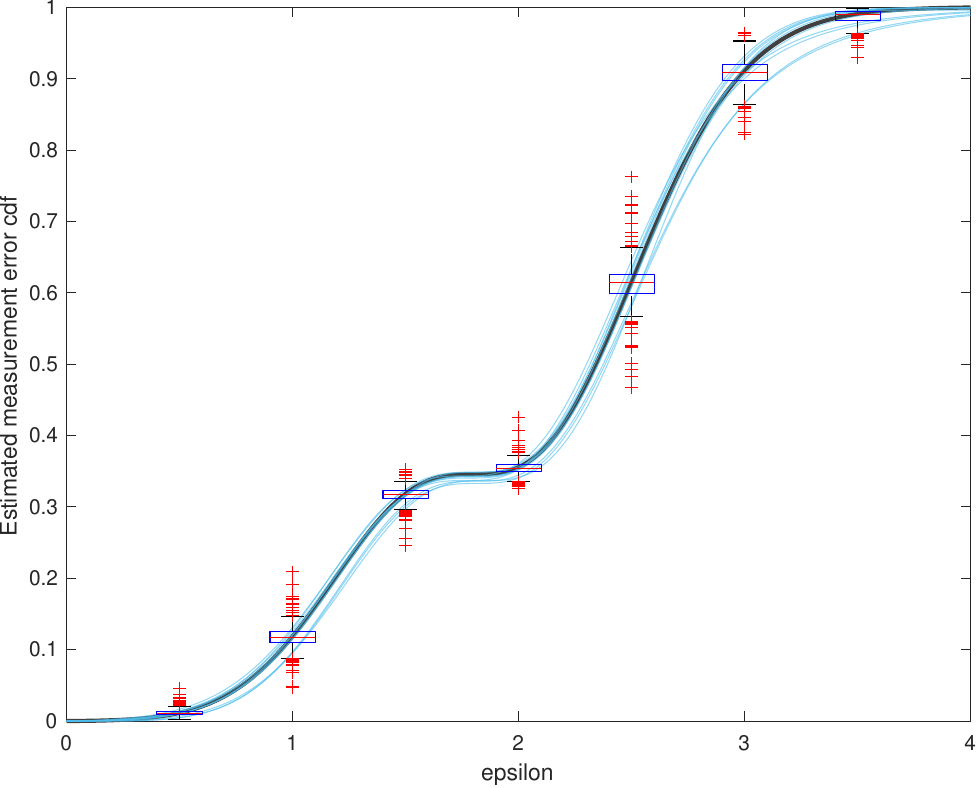}
    \caption*{\emph{Panel 2(c)}}
  \end{subfigure}
    \begin{subfigure}[!t]{0.33\textwidth}
    \centering
    \includegraphics[width=\textwidth]{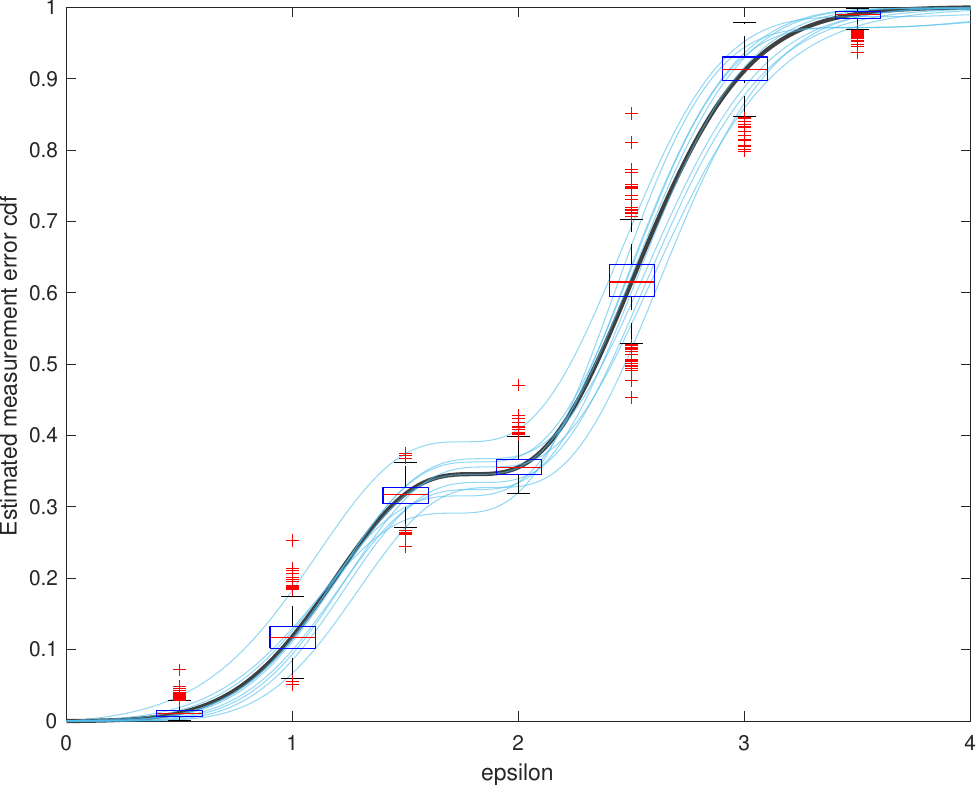}
    \caption*{\emph{Panel 3(a)}}
  \end{subfigure}%
  \hfill%
  \begin{subfigure}[!t]{0.33\textwidth}
    \centering
    \includegraphics[width=\textwidth]{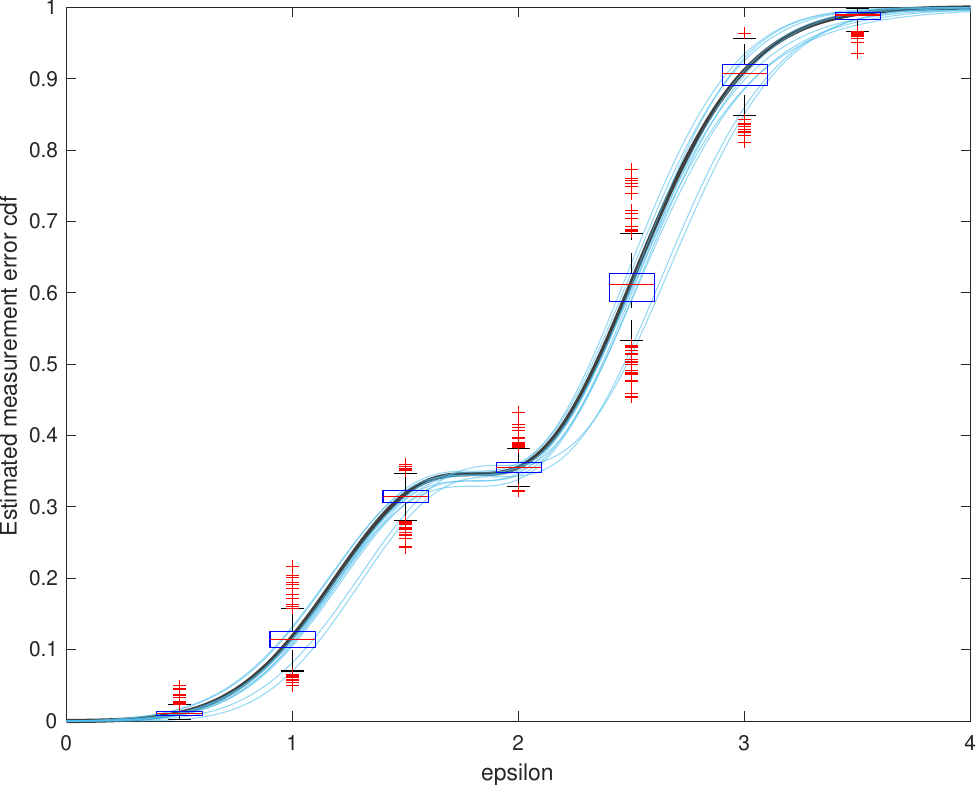}
    \caption*{\emph{Panel 3(b)}}
  \end{subfigure}%
  \hfill%
  \begin{subfigure}[!t]{0.33\textwidth}
    \centering
    \includegraphics[width=\textwidth]{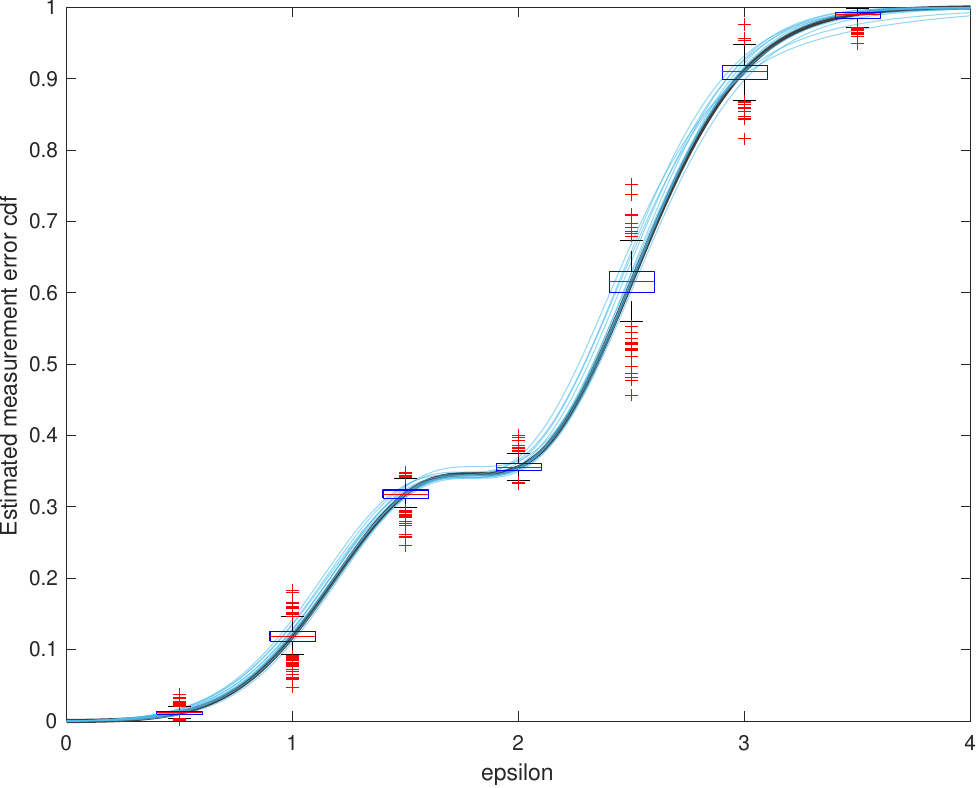}
    \caption*{\emph{Panel 3(c)}}
  \end{subfigure}
  \caption{Monte Carlo simulation results for $F_\varepsilon$. Panels 1, 2, and 3 correspond to the tuning parameter $\kappa=1$, $3.14$, and $5$, respectively; Subpanels \emph{(a)}, \emph{(b)}, and \emph{(c)} correspond to sample sizes $N=1000$, $2000$, and $4000$, respectively.} \label{fig:montecarlo2}
\end{figure}

\end{appendices}

\end{document}